\documentclass[]{elsarticle}

\usepackage{amsfonts, bm, amsmath, graphicx, amssymb, amsthm, mathdots, mathtools}
\usepackage[utf8]{inputenc}
\usepackage[english]{babel}
\usepackage{enumitem}
\usepackage{float}
\usepackage[table]{xcolor}
\usepackage{hyperref}
\usepackage{comment}
\usepackage{subcaption} 
\usepackage{graphicx}
\usepackage{thmtools, thm-restate}

\theoremstyle{plain}
\newtheorem{teo}{Theorem}[section]
\newtheorem*{teo*}{Theorem}

\newtheorem{claim}[teo]{Claim}

\newtheorem{lema}[teo]{Lemma}
\newtheorem{defn}[teo]{Definition}
\theoremstyle{remark}

\newtheorem{remark}[teo]{\textbf{\textit{Remark}}}

\newcounter{cases}
\newcounter{subcases}[cases]
\newcounter{subsubcases}[subcases]
\newenvironment{mycases}
  {%
    \setcounter{cases}{0}%
    \setcounter{subcases}{0}%
    \setcounter{subsubcases}{0}%
    \def\case
      {%
        \par\noindent
        \refstepcounter{cases}%
        \textit{\textbf{Case (\thecases})}
      }%
    \def\subcase
      {%
        \par\noindent
        \refstepcounter{subcases}%
        \textbf{\textit{Case (\thecases.\thesubcases) }}
      }%
     \def\subsubcase
      {%
        \par\noindent
        \refstepcounter{subsubcases}%
        \textbf{\textit{Case (\thecases.\thesubcases.\thesubsubcases) }}
      }%
  }
  {%
    \par
  }
\renewcommand*\thecases{\arabic{cases}}
\renewcommand*\thesubcases{\arabic{subcases}}
\renewcommand*\thesubsubcases{\arabic{subsubcases}}

\newcommand*{\QED}{\hfill\ensuremath{\square}}%
\newcommand*{\tagg}{\mathrm{tag}}%
\newcommand{\fsc}{$\mathcal{F}_{sc}$}%

\setlist[itemize]{itemsep=0.1pt, topsep=0pt, nolistsep}

\let\oldbibliography\thebibliography
\renewcommand{\thebibliography}[1]{%
  \oldbibliography{#1}%
  \setlength{\itemsep}{0pt}%
}

\definecolor{dark-blue}{RGB}{60,80,170}
\definecolor{dark-red}{RGB}{240,50,60}
\definecolor{dark-orange}{RGB}{255,147,23}
\definecolor{dark-green}{RGB}{65, 173, 57}

\begin{document}

\begin{frontmatter}

\date{}
\title{Forbidden induced subgraph characterization of circle graphs within split graphs}

\author[DC,ICC]{Flavia~Bonomo-Braberman}\ead{fbonomo@dc.uba.ar}

\author[IC,DM,DII]{Guillermo~Dur\'an}\ead{gduran@dm.uba.ar}

\author[ICC,IC]{Nina~Pardal}\ead{npardal@ic.fcen.uba.ar}

\author[UNS,INMABB]{Mart\'{\i}n~D.~Safe}\ead{msafe@uns.edu.ar}

\address[DC]{Universidad de Buenos Aires. Facultad de Ciencias Exactas y Naturales. Departamento de Computaci\'on. Buenos Aires, Argentina.}
\address[ICC]{CONICET-Universidad de Buenos Aires. Instituto de Investigaci\'on en Ciencias de la Computaci\'on (ICC). Buenos Aires, Argentina.}
\address[IC]{CONICET-Universidad de Buenos Aires. Instituto de C\'alculo (IC). Buenos Aires, Argentina.}
\address[DM]{Universidad de Buenos Aires. Facultad de Ciencias Exactas y Naturales. Departamento de Matem\'atica. Buenos Aires, Argentina.}
\address[DII]{Departamento de Ingenier{\'{\i}}a Industrial, Facultad de Ciencias F{\'{\i}}sicas y Matem{\'a}ticas, Universidad de Chile, Santiago, Chile.}
\address[UNS]{Departamento de Matem\'atica, Universidad Nacional del Sur (UNS), Bah\'ia Blanca, Argentina}
\address[INMABB]{INMABB, Universidad Nacional del Sur (UNS)-CONICET, Bah\'ia Blanca, Argentina}

\begin{abstract}

A graph is \emph{circle} if its vertices are in correspondence with a family of chords in a circle in such a way that every two distinct vertices are adjacent if and only if the corresponding chords have nonempty intersection. 
Even though there are diverse characterizations of circle graphs, a structural characterization by minimal forbidden induced subgraphs for the entire class of circle graphs is not known, not even restricted to split graphs (which are the graphs whose vertex set can be partitioned into a clique and a stable set). In this work, we give a characterization by minimal forbidden induced subgraphs of circle graphs, restricted to split graphs.

\end{abstract}


\end{frontmatter}

\section{Introduction}

A graph is \emph{circle} if its vertices are in correspondence with a family of chords in a circle in such a way that two vertices are adjacent if and only if the corresponding chords have nonempty intersection.
These graphs were defined by Even and Itai~\cite{E-I-circ} in 1971 to solve an ordering problem stated by Knuth, using the minimum number of parallel intermediate stacks without the restriction on completion of loading before unloading. They proved that this problem can be translated into the problem of finding the chromatic number of a circle graph.
In 1985, Naji~\cite{Naji} characterized circle graphs in terms of the solvability of a system of linear equations, yielding a $O(n^7)$-time recognition algorithm for this class. More recently, it was shown in~\cite{G-P-T-C-14} that circle graphs can be recognized in almost linear time. For a survey on circle graphs, see~\cite{D-G-S-14}.

All graphs in this work are simple, undirected, with no loops or multiple edges. Let $G=(V,E)$ be a graph.
Given $V' \subseteq  V$, the \emph{subgraph of $G$ induced by} $V'$, denoted $G\left[V'\right]$, is the graph  whose vertex set is $V'$ and whose edge set consists of all the edges in $E$ that have both endpoints in $V'$.
The \emph{neighborhood} of a vertex $u\in V$ is the subset $N_G(u)$ (or simply $N(u)$) consisting of the vertices of $G$ that are adjacent to $u$.
The \emph{complement} of $G$ is the graph $\overline{G}$ having $V$ as vertex set and such that every two distinct vertices of $\overline{G}$ are adjacent if and only if they are non-adjacent in $G$.
The \emph{local complement} of $G$ with respect to a vertex $u \in V$ is the graph $G*u$ that arises from $G$ by replacing the induced subgraph $G\left[N(u)\right]$ by its complement. Two graphs $G$ and $H$ are \emph{locally equivalent} if and only if $G$ arises from $H$ by a finite sequence of local complementations. Circle graphs were characterized by Bouchet~\cite{Bou-circle-obs} in 1994 in terms of forbidden induced subgraphs of some locally equivalent graph. Inspired by this result, Geelen and Oum~\cite{Geelen-Oum-circle} gave a new characterization of circle graphs in terms of \emph{pivoting}. The result of \emph{pivoting a graph $G$ with respect to an edge $uv$} is the graph $G \times uv = G * u * v * u$. 
Let $G_1$ and $G_2$ be two graphs such that $\vert V(G_i)\vert \geq 3$, for each $i=1,2$, and assume that $V(G_1)\cap V(G_2)=\emptyset$. Let $v_i$ be a distinguished vertex of $G_i$, for each $i=1,2$. The \emph{split composition} of $G_1$ and $G_2$ with respect to $v_1$ and $v_2$ is the graph $G_1*G_2$ whose vertex set is $V(G_1*G_2)=(V(G_1)\cup V(G_2))\setminus\{v_1,v_2\}$ and whose edge set is $E(G_1*G_2)=E(G_1-\{v_1\})\cup E(G_2-\{v_2\})\cup\{uv:u\in N_{G_1}(v_1)\mbox{ and }v\in
N_{G_2}(v_2)\}$. The vertices $v_1$ and $v_2$ are called the \emph{marker vertices}. We say that $G$ has a \emph{split decomposition} if there exist two graphs $G_1$ and $G_2$ with $\vert V(G_i)\vert\geq 3$, for each $i=1,2$, such that $G=G_1*G_2$ with respect to some pair of marker vertices; if so, $G_1$ and $G_2$
are called the \emph{factors} of the split decomposition.
Those graphs that do not admit a split decomposition are called \emph{prime graphs}. Notice that, if any of the factors of a split decomposition admits a split decomposition, we can continue the process until every factor is prime, a star or a complete graph. 
Bouchet proved that circle graphs are closed under split composition~\cite{Bou-circle}.


In spite of all these results, no characterizations for the entire class of circle graphs by forbidden induced subgraphs is known. Some partial results in this direction were found by restricting the problem to the classes of $P_4$-tidy graphs, tree-cographs and linear-domino graphs~\cite{B-D-G-S-circle}.

In this work, we consider the problem of characterizing circle graphs by minimal forbidden induced subgraphs restricted to split graphs. The motivation to study circle graphs restricted to this particular graph class comes from \emph{chordal graphs}, which are those graphs that contain no induced cycle of length greater than $3$, and constitute a widely studied graph class with very nice structural properties. 
Something similar happens with the class of \emph{split graphs}, which is an interesting subclass of chordal graphs. 
Split graphs are those graphs whose vertex set can be partitioned into a clique and a stable set.
Equivalently, split graphs are those chordal graphs whose complement is also a chordal graph.
Hence, studying those split graphs that are also circle is a good first step towards a characterization of those chordal graphs that are also circle. If $G$ is a split graph, the pair $(K,S)$ is a \emph{split partition} of $G$ if $\{K,S\}$ is a partition of the vertex set of $G$ and the vertices of $K$ (resp.\ $S$) are pairwise adjacent (resp.\ nonadjacent); i.e., $K$ is a clique and $S$ is a stable set. We denote it $G=(K,S)$.

Let us consider a split graph $G$, and suppose that $G$ is \emph{minimally non-circle}; i.e., $G$ is not circle but every proper induced subgraph of $G$ is circle.
Since $G$ is not circle, in particular $G$ is not a permutation graph. \emph{Permutation graphs} are exactly those comparability graphs whose complement graph is also a comparability graph~\cite{E-P-L-permut}. Comparability graphs were characterized by forbidden induced subgraphs in~\cite{Gal-comp}. 
This characterization of comparability graphs leads to a forbidden induced subgraph characterization for the class of permutation graphs.
Hence, given that permutation graphs are a subclass of circle graphs (see, e.g.~\cite[p.\ 252]{Go-perf-2ed}), in particular $G$ is not a permutation graph. Using the list of minimal forbidden induced subgraphs for permutation graphs and the fact that $G$ is a split graph, we conclude that $G$ contains either a \emph{tent}, a \emph{$4$-tent}, a \emph{co-$4$-tent} or a \emph{net} as an induced subgraph (see Figure~\ref{fig:forb_permsplit_base}). It is not difficult to see that these four graphs are also circle graphs. 

\begin{figure}[h]
\centering
\includegraphics[scale=.35]{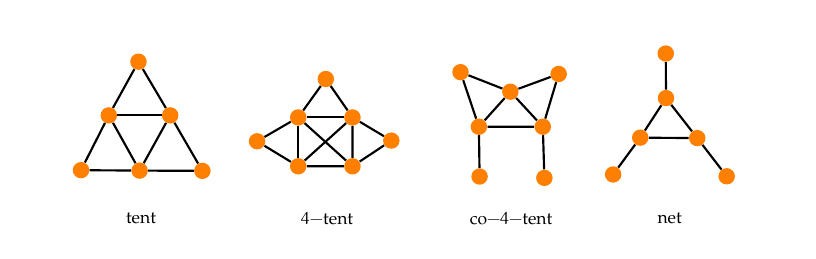}
\caption{The forbidden induced subgraphs for permutation graphs within split graphs.
} \label{fig:forb_permsplit_base}
\end{figure}

In Figure~\ref{fig:forb_graphs} we define some graph families that will be central throughout the sequel. The \emph{odd $k$-suns with center} are defined for each odd $k\geq 3$. The \emph{even $k$-sun} is defined for each even $k\geq 4$. We denote by \fsc\ the graph class consisting of the graphs belonging to any of the families depicted in Figure~\ref{fig:forb_graphs}. None of the graphs in \fsc\ is a circle graph (see Lemma~\ref{lem:nocircle}).

The theorem below, which is the main result of this work, gives the characterization of circle graphs by minimal forbidden induced subgraphs, restricted to split graphs.

\begin{teo} \label{teo:circle_split_caract}
Let $G$ be a split graph. Then, $G$ is a circle graph if and only if $G$ contains none of the graphs in \fsc \ (depicted in Figure~\ref{fig:forb_graphs}) as induced subgraph.
\end{teo}

\begin{figure}[h]
\centering
\includegraphics[scale=.27]{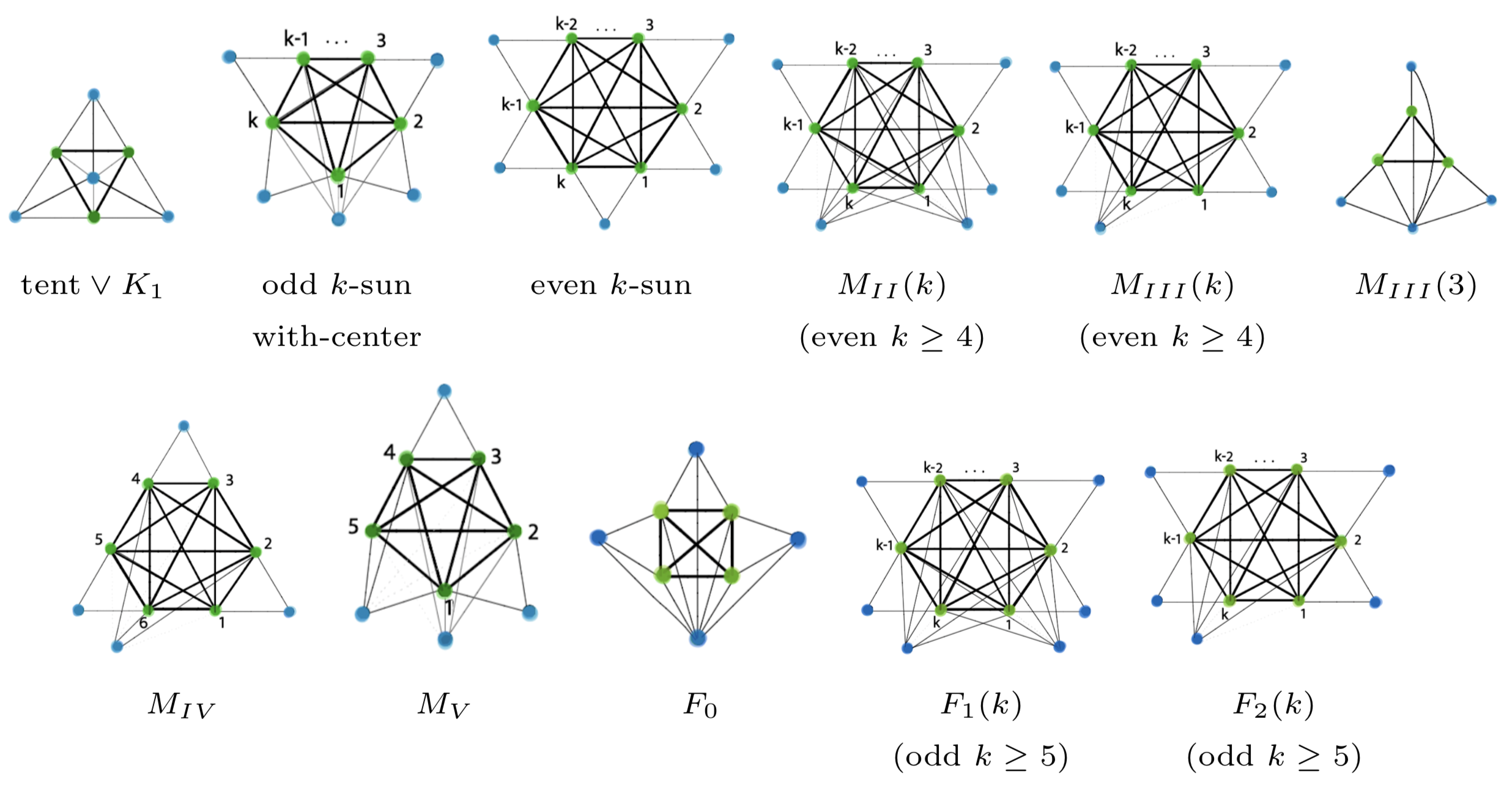}
\caption{The family \fsc \ of graphs. The $3$-sun with center is also known as tent-with-center.} \label{fig:forb_graphs}
\end{figure}



This work is organized as follows. In Section~\ref{section:partitions}, given a split graph $G$ with split partition $(K,S)$ and an induced subgraph $H$ of $G$ isomorphic to tent, $4$-tent or co-$4$-tent, we introduce partitions of $K$ and $S$ according to the adjacencies and prove that these partitions are well defined. In Section~\ref{section:split_circle_graphs}, we address the problem of characterizing the forbidden induced subgraphs of a circle graph that contains an induced tent, $4$-tent, co-$4$-tent or net, and we finish by giving the guidelines to draw a circle model for each case. In each subsection we address a case for proof of Theorem~\ref{teo:circle_split_caract}, as explained in Subsection~\ref{subsec:circle5}. In Section~\ref{sec:final}, we set out some final remarks and future challenges about structural characterizations of circle graphs.

\section{Preliminaries} \label{section:partitions}

Throughout this section, we define some subsets in both $K$ and $S$ depending on whether $G$ contains an induced tent, $4$-tent or co-$4$-tent $H$ as an induced subgraph. We prove that these subsets induce a partition of both $K$ and $S$.
In each case, we first partition the vertices in the complete set $K$ into subsets, according to the adjacencies with the vertices of $V(T) \cap S$, and then we partition the vertices of the independent set $S$ into subsets, according to the adjacencies with the partition defined on $K$.
We give the full proof when $G$ contains an induced tent, and state which parts of $K$ and $S$ are nonempty when $G$ contains a $4$-tent or a co-$4$-tent, since the proof is very similar in these cases. For more details on this, see~\cite{P20}.
These partitions will be useful in Section~\ref{section:split_circle_graphs}, when we give the proof of the characterization by forbidden induced subgraphs for split circle graphs.
Notice that we do not consider the case in which $G$ contains an induced net. We explain in detail in Section~\ref{subsec:circle5} how this case can be reduced to one of the other cases. 

Let $A=(a_{ij})$ be a $n\times m$ $(0,1)$-matrix.
We denote by $a_{i.}$ and $a_{.j}$ the $i$th row and the $j$th column of matrix $A$. From now on, we associate each row $a_{i.}$ with the set of columns in which $a_{i.}$ has a $1$. For example, the \emph{intersection} of two rows $a_{i.}$ and $a_{j.}$ is the subset of columns in which both rows have a $1$.
Two rows $a_{i.}$ and $a_{k.}$ are \emph{disjoint} if there is no $j$ such that $a_{ij} = a_{kj} = 1$.
We say that $a_{i.}$ is \emph{contained} in $a_{k.}$ if for each $j$ such that $a_{ij} = 1$ also $a_{kj} = 1$. We say that $a_{i.}$ and $a_{k.}$ are \emph{nested} if $a_{i.}$ is contained in $a_{k.}$ or $a_{k.}$ is contained in $a_{i.}$.
We say that a row $a_{i.}$ is \emph{empty} if every entry of $a_{i.}$ is $0$, and we say that $a_{i.}$ is \emph{nonempty} if there is at least one entry of $a_{i.}$ equal to $1$.
We say that two nonempty rows \emph{overlap} if they are non-disjoint and non-nested.
For every nonempty row $a_{i.}$, let $l_i = \min\{ j \colon\,a_{ij} = 1 \}$ and $r_i = \max\{ j \colon\,a_{ij} = 1 \}$ for each $i\in\{1,\ldots,n\}$.
Finally, we say that $a_{i.}$ and $a_{k.}$ \emph{start} (resp.\ \emph{end}) \emph{in the same column} if $l_i = l_k$ (resp.\ $r_i = r_k$), and we say $a_{i.}$ and $a_{k.}$ \emph{start (end) in different columns}, otherwise.

Let $G$ be a split graph with split partition $(K,S)$, $n=\vert S\vert$, and $m=\vert K\vert$.
Let $s_1, \ldots, s_n$ and $v_1, \ldots, v_m$ be linear orderings of $S$ and $K$, respectively. Let $A= A(S,K)$ be the $n\times m$ matrix defined by $A(i,j)=1$ if $s_i$ is adjacent to $v_j$ and $A(i,j)=0$, otherwise. From now on, we associate the row (resp.\ column) of the matrix $A(S,K)$ with the corresponding vertex in the independent set (resp.\ vertex in the complete set) of the  partition.

\subsection{Partitions of $S$ and $K$ for a graph containing an induced tent} \label{subsec:tent_partition}

Let $G=(K,S)$ be a split graph where $K$ is a clique and $S$ is an independent set. Let $H$ be an induced subgraph of $G$ isomorphic to a tent. Let $V(T)=\{k_1,$ $k_3,$ $k_5,$ $s_{13},$ $s_{35},$ $s_{51}\}$ where $k_1,$ $k_3,$ $k_5\in K$, $s_{13},$ $s_{35},$ $s_{51}\in S$, and the neighbors of $s_{ij}$ in $H$ are precisely $k_i$ and $k_j$.

We introduce sets $K_1,K_2,\ldots,K_6$ as follows.
\begin{itemize}
 \item For each $i\in\{1,3,5\}$, let $K_i$ be the set of vertices of $K$ whose neighbors in $V(T)\cap S$ are precisely $s_{(i-2)i}$ and $s_{i(i+2)}$ (where subindexes are modulo~$6$).
 \item For each $i\in\{2,4,6\}$, let $K_i$ be the set of vertices of $K$ whose only neighbor in $V(T)\cap S$ is $s_{(i-1)(i+1)}$ (where subindexes are modulo~$6$).
\end{itemize}
See Figure~\ref{fig:tent_ext} for a graphic idea of this. Notice that $K_1$, $K_3$ and $K_5$ are always nonempty sets.

\vspace{2mm}

We say a vertex $v$ is \emph{complete to} the set of vertices $X$ if $v$ is adjacent to every vertex in $X$, and we say $v$ is \emph{anticomplete to} $X$ if $v$ has no neighbor in $X$.
We say that $v$ is \emph{adjacent to} $X$ if $v$ has at least one neighbor in $X$. Notice that complete to $X$ implies adjacent to $X$ if and only if $X$ is nonempty.
For $v$ in $S$, let $N_i(v) = N(v) \cap K_i$.
Given two vertices $v_1$ and $v_2$ in $S$, we say that $v_1$ and $v_2$ are \emph{nested} if either $N(v_1) \subseteq N(v_2)$ or $N(v_2) \subseteq N(v_1)$.
In particular, given $i \in \{1, \ldots, 6\}$, if either $N_{i}(v_1) \subseteq N_i(v_2)$ or $N_{i}(v_2) \subseteq N_i(v_1)$, then we say that
$v_1$ and $v_2$ are \emph{nested in} $K_i$.
Additionally, if $N(v_1) \subseteq N(v_2)$, then we say that \emph{$v_1$ is contained in $v_2$}.


\begin{figure}[h!]
    \begin{center}
        \includegraphics[scale=.23]{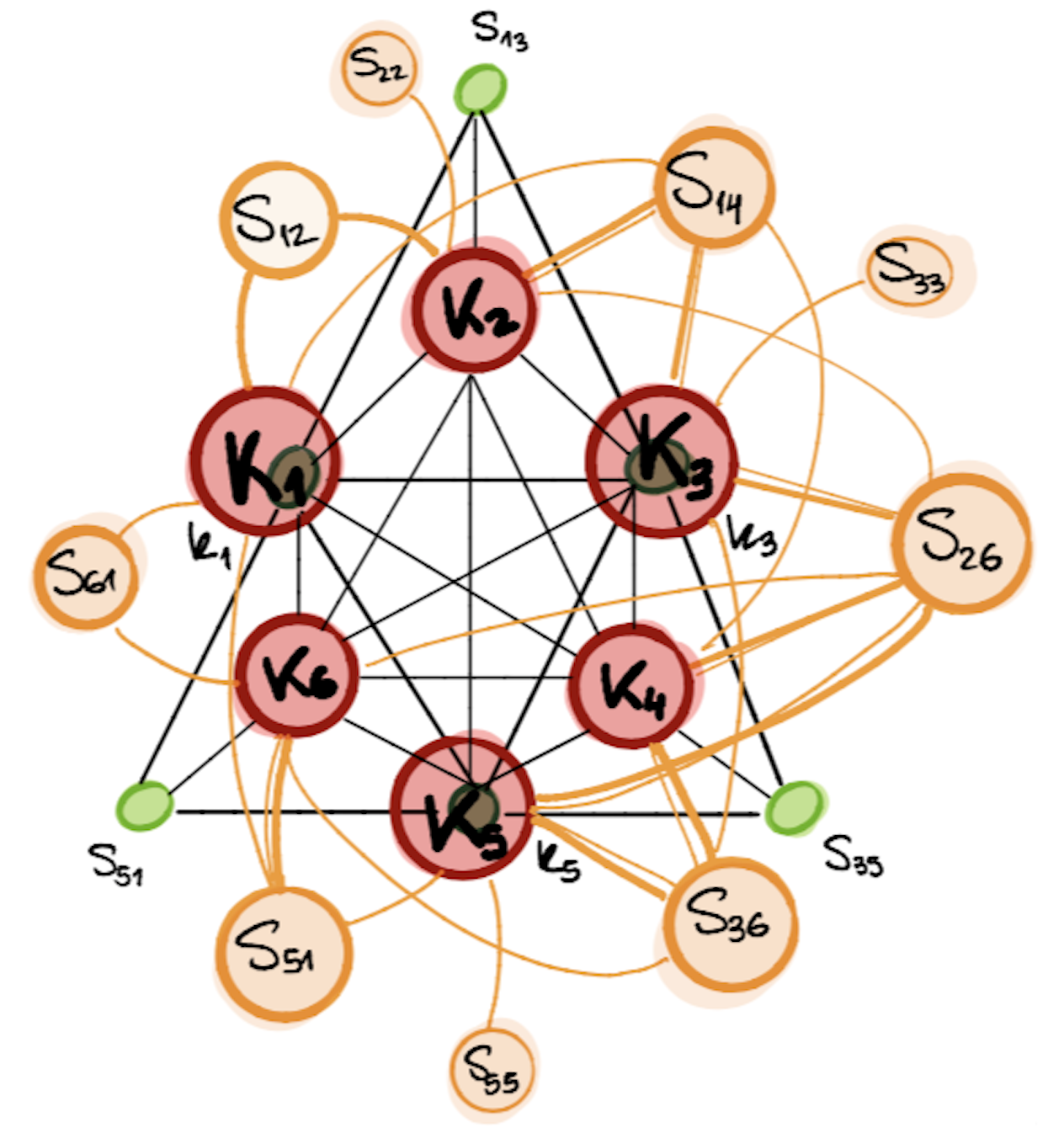}
    \end{center}
    \caption{Tent $H$ and the split graph $G$ according to the given extensions.} \label{fig:tent_ext}
\end{figure}

Given a graph $H$, the graph $G$ is $H$-free if it does not contain $H$ as induced subgraph. For a family of graphs $\mathcal{H}$, a graph $G$ is $\mathcal{H}$-free if $G$ is $H$-free for every $H \in \mathcal{H}$.

\begin{lema} \label{lema:tent_1}
    If $G$ is \fsc-free, then $\{K_1,K_2,\ldots,K_6\}$ is a partition of $K$.
\end{lema}

\begin{proof}
    Every vertex of $K$ is adjacent to precisely one or two vertices of $V(T) \cap S$, for if not we find either a tent${}\vee{}K_1$ or a $3$-sun with center as induced subgraph of $G$, a contradiction.
\end{proof}

Let $i,j\in\{1,\ldots,6\}$ and let $S_{ij}$ be the set of vertices of $S$ that are adjacent to some vertex in $K_i$ and some vertex in $K_j$, are complete to $K_{i+1}$,$K_{i+2},\ldots,K_{j-1}$, and are anticomplete to $K_{j+1}$,$K_{j+2},\ldots,K_{i-1}$ (where subindexes are modulo~$6$).
The following claims are necessary to prove Lemma~\ref{lema:tent_2}, that states, on the one hand, which sets of $\{S_{ij}\}_{i,j \in \{1, \ldots, 6\}}$ may be nonempty, and, on the other hand, and that the sets $\{S_{ij}\}_{i,j \in \{1, \ldots, 6\}}$ indeed induce a partition of $S$.
This shows that the adjacencies of a vertex of $S$ have a circular structure with respect to the defined partition of $K$.

    \begin{claim} \label{claim:tent_1}
        If $G$ is \fsc-free, then there is no vertex $v$ in $S$ such that $v$ is simultaneously adjacent to $K_1$, $K_3$ and $K_5$.
        Moreover, there is no vertex $v$ in $S$ adjacent to $K_2$, $K_4$ and $K_6$ such that $v$ is anticomplete to any two of $K_j$, for $j \in \{1, 3, 5 \}$.
    \end{claim}

    Let $v$ in $S$ and let $w_i$ in $K_i$ for each $i\in\{1,3,5\}$, such that $v$ is adjacent to each $w_i$. Hence, $\{w_1,$ $w_3,$ $w_5,$ $s_{13},$ $s_{35},$ $s_{51},$ $v\}$ induce in $G$ a $3$-sun with center, a contradiction.

    To prove the second statement, let $w_i$ in $K_i$ such that $v$ is adjacent to $w_i$ for every $i\in\{2,4,6\}$. Suppose that $v$ is anticomplete to $K_3$ and $K_5$. Thus, we find a $4$-sun induced by the set $\{ w_2$, $k_3$, $k_5$, $w_6$, $s_{13}$, $s_{35}$, $s_{51}$, $v \}$.
    If instead $v$ is anticomplete to $K_1$ and $K_3$, then we find a $4$-sun induced by $\{ k_1$, $k_3$, $w_4$, $w_6$, $s_{13}$, $s_{35}$, $s_{51}$, $v \}$, and if $v$ is anticomplete to $K_1$ and $K_5$, then a $4$-sun is induced by $\{ k_1$, $w_2$, $w_4$, $k_5$, $s_{13}$, $s_{35}$, $s_{51}$, $v \}$. \QED

    \begin{claim} \label{claim:tent_2}
        If $G$ is \fsc-free and $v$ in $S$ is adjacent to $K_i$ and $K_{i+3}$, then $v$ is complete to $K_j$, either for $j \in \{ i+1, i+2 \}$ or for $j \in \{ i-1, i-2\}$.
    \end{claim}

    Let $w_i$ in $K_i$, $w_{i+3}$ in $K_{i+3}$ such that $v$ is adjacent to $w_i$ and $w_{i+3}$. Notice that the statement is exactly the same for $i = j$ and for $i = j+3$, so let us assume that $i$ is even.

    If $w_j$ in $K_j$ is a non-neighbor of $v$ for each $j \in \{ i-1, i+1 \}$, then we find an induced $M_{III}(3)$. Hence, $v$ is complete to $K_j$ for at least one of $j \in \{i-1, i+1 \}$.
    Suppose that $v$ is complete to $K_{i+1}$. If $K_{i+2} = \emptyset$, then the claim holds. Hence, suppose that $K_{i+2} \neq \emptyset$ and suppose $w_{i+2}$ in $K_{i+2}$ is a non-neighbor of $v$. In particular, since $v$ is adjacent to $w_{i+3}$ and $k_{i+1}$, then $v$ is anticomplete to $K_{i-1}$ by Claim~\ref{claim:tent_1}. However, in this case we find $M_{III}(3)$ induced by $\{s_{(i-1)(i+3)}$, $s_{(i+3)(i-1)}$, $v$, $k_{i-1}$, $k_{i+1}$, $w_{i+2}$, $w_{i+3} \}$.
    It follows analogously if instead $v$ is complete to $K_{i-1}$ and is not complete to $K_{i-2}$, for we find the same induced subgraphs. Notice that the proof is independent on whether $K_j = \emptyset$ or not, for every even $j$. \QED

    \begin{claim} \label{claim:tent_3odd}
        If $G$ is \fsc-free and $v$ in $S$ is adjacent to $K_i$ and $K_{i+2}$, with $i$ odd, then $v$ is complete to $K_{i+1}$.
    \end{claim}

    Given the symmetry of the odd-indexed and even-indexed sets $K_j$, we analyze the case in which $v$ is adjacent to $K_1$ and $K_3$. Let $w_1,w_3$ be the respective neighbors. By Claim~\ref{claim:tent_1}, $v$ is anticomplete to $K_5$.
    If $v$ is nonadjacent to some vertex $w_2$ in $K_2$, then the set $\{ s_{35}$, $v$,  $s_{51}$, $w_{1}$, $w_{3}$, $k_{5}$, $w_{2} \}$ induces a $3$-sun with center. Hence, $v$ is complete to $K_2$. \QED

    \begin{claim} \label{claim:tent_3even}
        If $G$ is \fsc-free and $v$ in $S$ is adjacent to $K_i$ and $K_{i+2}$, with $i$ even, then either $v$ is complete to $K_{i+1}$ and one of $\{K_{i-1}, K_{i+3}\}$, or $v$ is complete to $K_j$ for $j \in \{ i-1, i-2, i-3 \}$.
    \end{claim}

    Given the symmetry of the odd-indexed and even-indexed sets $K_j$, we analyze the case in which $v$ is adjacent to $K_2$ and $K_4$. Let $w_2,w_4$ be the respective neighbors.
    First, notice that $v$ is complete to either $K_1$ or $K_5$, for if not we find a $4$-sun induced by $\{ s_{13}$, $s_{51}$, $s_{35}$, $v$, $w_2$, $w_1$, $w_5$, $w_4 \}$, where $w_1$ and $w_5$ are non-neighbors of $v$ in $K_1$ and $K_5$, respectively.
    Suppose that $v$ is complete to $K_1$. If $v$ is not complete to $K_3$, then $v$ is complete to $K_5$ and $K_6$, for if not there is $M_{III}(3)$ induced by $\{ s_{13}$, $s_{51}$, $v$, $k_1$, $w_3$, $w_4$, $w_j \}$ for both $j=5,6$, where $w_\ell$ is a non-neighbor of $v$ in $K_\ell$, for $\ell \in \{3,5,6\}$. \QED

    \begin{remark}
        As a consequence of the previous claims we also proved that, if $G$ is \fsc-free, then:
        \begin{itemize}
            \item For each $i \in \{1, 2, \ldots, 6\}$, the sets $S_{i(i-1)}$ are empty, for if not, there is a vertex $v$ in $S$ such that $v$ is adjacent to $K_1$, $K_3$ and $K_5$ (contradicting Claim~\ref{claim:tent_1}). Moreover, the same holds for $S_{i(i-2)}$, for each $i \in \{1, 3, 5\}$.
            \item For each $i \in \{2, 4, 6 \}$, the sets $S_{i(i+2)}$ are empty since every vertex $v$ in $S$ such that $v$ is adjacent to $K_i$ and $K_{i+2}$ is necessarily complete to either $K_{i-1}$ or $K_{i+3}$ (Claim~\ref{claim:tent_3even}).
        \end{itemize}
    \end{remark}

    \begin{claim} \label{claim:tent_4}
        If $G$ is \fsc-free, then for each $i \in \{1, 3, 5\}$, every vertex in $S_{i(i+3)} \cup S_{(i+3)i}$ is complete to $K_i$.
    \end{claim}

    We will prove this claim without loss of generality for $i = 1$.

    Let $v$ in $S_{14}$. By definition, $v$ is adjacent to $k_3$ and nonadjacent to $k_5$. Towards a contradiction, let $w_{11}$ and $w_{12}$ in $K_1$ such that $v$ is nonadjacent to $w_{11}$ and $v$ is adjacent to $w_{12}$, and let $w_4$ in $K_4$ such that $v$ is adjacent to $w_4$. In this case, we find $F_0$ induced by the set $\{ s_{13}$, $s_{35}$, $v$, $w_{11}$, $w_{12}$, $k_3$, $w_4$, $k_5 \}$.

    Analogously, if $v$ is in $S_{41}$, then $F_0$ is induced by $\{ s_{35}$, $s_{51}$, $v$, $w_{11}$, $w_{12}$, $k_3$, $w_4$, $k_5 \}$. \QED

    \vspace{3mm}

The following lemma is a straightforward consequence of Claims~\ref{claim:tent_1} to~\ref{claim:tent_4}.

\begin{lema} \label{lema:tent_2}
Let $G=(K,S)$ be a split graph that contains an induced tent. If $G$ is \fsc-free, then all the following assertions hold:
 \begin{itemize}
  \item $\{S_{ij}\}_{i,j\in\{1,2,\ldots,6\}}$ is a partition of $S$.
  \item For each $i\in\{1,3,5\}$, $S_{i(i-1)}$ and $S_{i(i-2)}$ are empty.
  \item For each $i\in\{2,4,6\}$, $S_{i(i-1)}$ and $S_{i(i+2)}$ are empty.
  \item For each $i\in\{1,3,5\}$, $S_{i(i+3)}$ and $S_{(i+3)i}$ are complete to $K_i$.
 \end{itemize}

\begin{table}[h]
\begin{center}
\small{
    \begin{tabular}{ c | c c c c c c }
         \hline
         $i\setminus j$ & 1 & 2 & 3 & 4 & 5 & 6 \\
          \hline
         1 & \checkmark & \checkmark & \checkmark & \textcolor{dark-orange}{\checkmark} & $\emptyset$ & $\emptyset$ \\
         2 & $\emptyset$ & \checkmark & \checkmark & $\emptyset$ & \checkmark & \checkmark \\
         3 & $\emptyset$ & $\emptyset$ & \checkmark & \checkmark & \checkmark & \textcolor{dark-orange}{\checkmark} \\
         4 & \textcolor{dark-orange}{\checkmark} & \checkmark & $\emptyset$ & \checkmark & \checkmark & $\emptyset$ \\
         5 & \checkmark & \textcolor{dark-orange}{\checkmark} & $\emptyset$ & $\emptyset$ & \checkmark & \checkmark \\
         6 & \checkmark & $\emptyset$ & \textcolor{dark-orange}{\checkmark} & \checkmark & $\emptyset$ & \checkmark \\
    \end{tabular}}
\caption{The (possibly) nonempty parts of $S$ in the tent case. The orange checkmarks denote those $S_{ij}$ for which every vertex is complete to $K_i$ or $K_j$.}
\end{center}
\end{table}
\end{lema}

\subsection{Partitions of $S$ and $K$ for a graph containing an induced 4-tent} \label{subsec:4tent_partition}

Let $G=(K,S)$ be a split graph where $K$ is a clique and $S$ is an independent set. Let $H$ be an induced subgraph of $G$ isomorphic to a $4$-tent. Let $V(T)=\{k_1,$ $k_2,$ $k_4,$ $k_5,$ $s_{12},$ $s_{24},$ $s_{45}\}$ where $k_1,k_2,k_4,k_5\in K$, $s_{12},s_{24},s_{45}\in S$, and the neighbors of $s_{ij}$ in $H$ are precisely $k_i$ and $k_j$.

We introduce sets $K_1,K_2,\ldots,K_6$ as follows.
\begin{itemize}
 \item Let $K_1$ be the set of vertices of $K$ whose only neighbor in $V(T)\cap S$ is $s_{12}$.
 Analogously, let $K_3$ be the set of vertices of $K$ whose only neighbor in $V(T)\cap S$ is $s_{24}$, and let $K_5$ be the set of vertices of $K$ whose only neighbor in $V(T)\cap S$ is $s_{45}$.
 \item For each $i\in \{2, 4\}$, let $K_i$ be the set of vertices of $K$ whose neighbors in $V(T)\cap S$ are precisely $s_{ji}$ and $s_{ik}$, for $i=2$, $j=1$ and $k=2$ or $i=4$, $j=2$ and $k=5$.

 \item Let $K_6$ be the set of vertices of $K$ that are anticomplete to $V(T)\cap S$.
\end{itemize}

Let $i,j \in \{1, \ldots, 6\}$ and let $S_{ij}$ defined as in the previous section. We denote by $S_{[ij}$ (resp.\ $S_{ij]}$) the set of vertices in $S$ that are adjacent to $K_j$ and complete to $K_i$, $K_{i+1}, \ldots, K_{j-1}$ (resp.\ adjacent to $K_i$ and complete to $K_{i+1}, \ldots, K_{j-1}, K_j$). 
We denote by $S_{[ij]}$ the set of vertices in $S$ that are complete to $K_i, \ldots, K_j$.

Consider those vertices in $S$ that are complete to $K_2, \ldots, K_5$ and are adjacent to al least one vertex in both $K_1$ and $K_6$. We consider these vertices divided into two distinct subsets: we denote by $S_{[16}$ the subset that contains those vertices that are complete to $K_1, K_2, \ldots, K_5$, and $S_{16}$ to the subset of those that are adjacent but not complete to $K_1$. Furthermore, we denote by $S_{65}$ the vertices in $S$ that are adjacent but not complete to $K_5$.

In an analogous way as in the tent case, we obtain the following lemma for a split graph that contains an induced $4$-tent. The details of the proof can be found in~\cite{P20}.

\begin{lema} \label{lema:4tent_particion}
Let $G=(K,S)$ be a split graph that contains an induced $4$-tent and contains no induced tent. If $G$ is \fsc-free, then all of the following assertions hold:
 \begin{itemize}
 \item $\{K_1, K_2, \ldots, K_6\}$ is a partition of $K$.
  \item $\{S_{ij}\}_{i,j\in\{1,2,\ldots,6\}}$ is a partition of $S$.
  \item For each $i\in\{2,3,4,5\}$, $S_{i1}$ is empty.
  \item For each $i\in\{3,4,5\}$, $S_{i2}$ is empty.
  \item The subsets $S_{43}$, $S_{53}$ and $S_{54}$ are empty. 
  \item The following subsets coincide: $S_{13}= S_{[13}$, $S_{14}=S_{14]}$, $S_{25}=S_{[25}$, $S_{26}=S_{[26}$, $S_{35}=S_{35]}$, $S_{46} = S_{[46}$, $S_{62} = S_{62]}$ and $S_{64} = S_{64]}$.
 \end{itemize}

\begin{table}[h!]
    \begin{center}
    \small{
        \begin{tabular}{ c | c c c c c c }
             \hline
             $i\setminus j$ & 1 & 2 & 3 & 4 & 5 & 6 \\
              \hline
             1 & \checkmark & \checkmark & \textcolor{dark-orange}{\checkmark} & \textcolor{dark-orange}{\checkmark} & \checkmark & \checkmark \\
             2 & $\emptyset$ & \checkmark & \checkmark & \checkmark & \textcolor{dark-orange}{\checkmark} & \textcolor{dark-orange}{\checkmark} \\
             3 & $\emptyset$ & $\emptyset$ & \checkmark & \checkmark & \textcolor{dark-orange}{\checkmark} & \checkmark \\
             4 & $\emptyset$ & $\emptyset$ & $\emptyset$ & \checkmark & \checkmark & \textcolor{dark-orange}{\checkmark} \\
             5 & $\emptyset$ & $\emptyset$ & $\emptyset$ & $\emptyset$ & \checkmark & \checkmark \\
             6 & \checkmark & \textcolor{dark-orange}{\checkmark} & \checkmark & \textcolor{dark-orange}{\checkmark} & \checkmark & \checkmark \\
        \end{tabular}}
    \end{center}
    \caption{The (possibly) nonempty parts of $S$ in the $4$-tent case. The orange checkmarks denote those sets $S_{ij}$ complete to either $K_i$ or $K_j$.} 
\end{table}

\end{lema}

\subsection{Partitions of $S$ and $K$ for a graph containing an induced co-4-tent} \label{subsec:co4tent_partition}

Let $G=(K,S)$ be a split graph where $K$ is a clique and $S$ is an independent set, and suppose that $G$ contains no induced tent or $4$-tent.
Let $H$ be an induced subgraph of $G$ isomorphic to a co-$4$-tent. Let $V(T)=\{k_1,$ $k_3,$ $k_5,$ $s_{13},$ $s_{35},$ $s_1,$ $s_5\}$ where $k_1,k_3,k_5\in K$, $s_{13},$ $s_{35},$ $s_1,$ $s_5$ in $S$ such that the neighbors of $s_{ij}$ in $H$ are precisely $k_i$ and $k_j$ and the neighbor of $s_i$ in $H$ is precisely $k_i$.

We introduce sets $K_1,K_2,\ldots,K_{15}$ as follows.
\begin{itemize}
 \item Let $K_1$ be the set of vertices of $K$ whose only neighbors in $V(T)\cap S$ are $s_1$ and $s_{13}$. Analogously, let $K_5$ be the set of vertices of $K$ whose only neighbors in $V(T)\cap S$ are $s_5$ and $s_{35}$, and let $K_3$ be the set of vertices of $K$ whose only neighbors in $V(T)\cap S$ are $s_{13}$ and $s_{35}$. Let $K_{13}$ be the set of vertices of $K$ whose only neighbors in $V(T)\cap S$ are $s_1$ and $s_5$, $K_{14}$ be the set of vertices of $K$ whose only neighbors in $V(T)\cap S$ are $s_{13}$ and $s_5$ and $K_{15}$ be the set of vertices of $K$ whose only neighbors in $V(T)\cap S$ are $s_1$ and $s_{35}$.
 \item Let $K_2$ be the set of vertices of $K$ whose neighbors in $V(T)\cap S$ are precisely $s_1$, $s_{13}$ and $s_{35}$, and let $K_4$ be the set of vertices of $K$ whose neighbors in $V(T)\cap S$ are precisely $s_5$, $s_{13}$ and $s_{35}$. Let $K_9$ be the set of vertices of $K$ whose neighbors in $V(T)\cap S$ are precisely $s_1$, $s_{13}$ and $s_5$, and let $K_{10}$ be the set of vertices of $K$ whose neighbors in $V(T)\cap S$ are precisely $s_1$, $s_{35}$ and $s_5$.
 \item Let $K_6$ be the set of vertices of $K$ whose only neighbor in $V(T)\cap S$ is precisely $s_{35}$, and let $K_8$ be the set of vertices of $K$ whose only neighbor in $V(T)\cap S$ is precisely $s_{13}$. Let $K_{11}$ be the set of vertices of $K$ whose only neighbor in $V(T)\cap S$ is precisely $s_1$, and let $K_{12}$ be the set of vertices of $K$ whose only neighbor in $V(T)\cap S$ is precisely $s_5$.
 \item Let $K_7$ be the set of vertices of $K$ that are anticomplete to $V(T) \cap S$.

\end{itemize}

\begin{remark} \label{obs:co4tent_1}
If $K_4 = \emptyset$, then there is a split decomposition of $G$. Let us consider the subset $K_5$ on the one hand, and on the other hand a vertex $u \not\in G$ such that $u$ is complete to $K_5$ and is anticomplete to $V(G) \setminus K_5$. Let $G_1$ and $G_2$ be the subgraphs induced by the vertex subsets $V_1 = V(G) \setminus S_{55}$ and $V_2 = \{u\} \cup K_5 \cup S_{55}$, respectively. Hence, $G$ is the result of the split composition of $G_1$ and $G_2$ with respect to $K_5$ and $u$. The same holds if $K_2 = \emptyset$ considering the subgraphs induced by the vertex subsets $V_1 = V(G) \setminus S_{11}$ and $V_2 = \{u\} \cup K_1 \cup S_{11}$, where in this case $u$ is complete to $K_1$ and is anticomplete to $V(G) \setminus K_1$.

If we consider $H$ a minimally non-circle graph, then $H$ is a prime graph, for if not one of the factors should be non-circle and thus $H$ would not be minimally non-circle~\cite{Bou-circle}. Hence, in order to characterize those circle graphs that contain an induced co-$4$-tent, we will assume without loss of generality that $G$ is a prime graph, and therefore $K_2 \neq \emptyset$ and $K_4 \neq \emptyset$. 
\end{remark}

In an analogous way as in the tent case, we obtain the following lemma for a split graph that contains an induced co-$4$-tent. The details of the proof can be found in~\cite{P20}.

\begin{lema} \label{lema:co4tent_particion}
Let $G=(K,S)$ be a split graph that contains an induced co-$4$-tent and contains no induced tent or $4$-tent. If $G$ is \fsc-free, then all the following assertions hold:
 \begin{itemize}
  \item $K_9, \ldots, K_{15}$ are empty sets and $\{K_1, K_2, \ldots, K_8\}$ is a partition of $K$.
  \item $\{S_{ij}\}_{i,j\in\{1,2,\ldots,8\}}$ is a partition of $S$.
  \item For each $i\in\{2,3,4,5,6,7,8 \}$, $S_{i1}$ is empty.
  \item For each $i\in\{3,4,5,6,7 \}$, $S_{i2}$ is empty.
  \item For each $i\in\{4,5,6,7 \}$, $S_{i3}$ is empty, and $S_{56}$ is also empty.
  \item For each $i\in\{3,4,5,6 \}$, $S_{i7}$ is empty.
  \item For each $i\in\{2,3,4,5,6,7 \}$, $S_{i8}$ is empty.
  \item The subsets $S_{64}$, $S_{54}$ and $S_{56}$ are empty.
  \item The following subsets coincide: $S_{1i}= S_{[1i}$ for $i=3,4,8$;
  $S_{16}=S_{16]}$, $S_{25}=S_{25]}$, $S_{27}=S_{[27}$, $S_{35}=S_{35]}$, $S_{46} = S_{[46}$, $S_{82} = S_{82]}$ and $S_{85} = S_{[85}$ (as the case may be, according to whether $K_i \neq \emptyset$ or not, for $i=6,7,8$).
 \end{itemize}

\end{lema}

Since $S_{18} = S_{[18}$, we will consider these vertices as those in $S_{87}$ that are complete to $K_7$ and $S_{18} = \emptyset$. Moreover, those vertices that are complete to $K_1, \ldots, K_6, K_8$ and are adjacent to $K_7$ will be considered as in $S_{76]}$, thus $S_{87}$ is the set of vertices of $S$ that are complete to $K_1, \ldots, K_7$ and are adjacent but not complete to $K_8$.
These results are summarized in Table~\ref{fig:tabla_co4tent_1}.

\begin{table}[h!]
\begin{center}
\small{
    \begin{tabular}{ c | c c c c c c c c}
         \hline
         $i\setminus j$ & 1 & 2 & 3 & 4 & 5 & 6 & 7 & 8 \\
          \hline
         1 & \checkmark & \checkmark & \textcolor{orange}{\checkmark} & \textcolor{orange}{\checkmark} & $\emptyset$ & \textcolor{orange}{\checkmark} & \checkmark & $\emptyset$ \\
         2 & $\emptyset$ & \checkmark & \checkmark & $\emptyset$ & \textcolor{orange}{\checkmark} & \checkmark & \textcolor{orange}{\checkmark} & $\emptyset$ \\
         3 & $\emptyset$ & $\emptyset$ & \checkmark & \checkmark & \textcolor{orange}{\checkmark} & \checkmark & $\emptyset$ & $\emptyset$ \\
         4 & $\emptyset$ & $\emptyset$ & $\emptyset$ & \checkmark & \checkmark & \textcolor{orange}{\checkmark} & $\emptyset$ & $\emptyset$ \\
         5 & $\emptyset$ & $\emptyset$ & $\emptyset$ & $\emptyset$ & \checkmark & $\emptyset$ & $\emptyset$ & $\emptyset$ \\
         6 & $\emptyset$  & $\emptyset$  & $\emptyset$  & $\emptyset$  & $\emptyset$ & \checkmark & $\emptyset$  &  $\emptyset$  \\
        7 & $\emptyset$ & $\emptyset$ & $\emptyset$ & \textcolor{orange}{\checkmark} & \checkmark & \checkmark & \checkmark & $\emptyset$ \\
        8 & $\emptyset$ & \textcolor{orange}{\checkmark} & \checkmark & \checkmark & \textcolor{orange}{\checkmark} & \checkmark & \checkmark & \checkmark \\
    \end{tabular}}
\end{center}
\caption{The (possibly) nonempty parts of $S$ in the co-$4$-tent case. The orange checkmarks denote those subsets $S_{ij}$ that are either complete to $K_i$ or $K_j$.} \label{fig:tabla_co4tent_1}
\end{table}

\section{Characterization by forbidden induced subgraphs of circle graphs within split graphs} \label{section:split_circle_graphs}

In this section, we give the proof of Theorem~\ref{teo:circle_split_caract}. The proof strongly relies on a characterization of \emph{$2$-nested matrices} given in~\cite{P20}. The $2$-nested matrices are those matrices having an ordering of its columns such that the ones in each row appear in at most two blocks and for which there is certain color assignment for every row using $2$ colors (for the details, see~Definition~\ref{def:2-nested}).

%
%

This section is organized as follows. First, we define and give the characterization of nested and $2$-nested matrices by forbidden subconfigurations (for the complete proof, see~\cite{P20},~\cite{PDS20}).
Afterwards, in Sections~\ref{subsec:circle2} to~\ref{subsec:circle5} we prove the characterization given in Theorem~\ref{teo:circle_split_caract}, which gives the complete list of forbidden induced subgraphs for those split graphs that are also circle. This proof is divided into four cases, depending on whether the split graph contains an induced tent, $4$-tent, co-$4$-tent or net.


A $(0,1)$-matrix has the \emph{consecutive-ones property (C1P) for the rows} if there is a permutation of its columns such that the ones in each row appear consecutively. We say that a matrix $B$ is a \emph{subconfiguration} of a matrix $A$ if $B$ equals some submatrix of $A$ up to permutation of rows and/or columns. Given a set of rows $R$ of a matrix $A$, we say that $R$ \emph{induces a matrix} $B$ if $B$ is a subconfiguration of the submatrix of $A$ given by selecting only those rows in $R$. Tucker characterized all the minimal forbidden subconfigurations for the C$1$P, later known as \emph{Tucker matrices} (a graphic representation of which can be found in~\cite{Tuc-c1p}).

\begin{defn} \label{def:nested_m}
    Let $A$ be a $(0,1)$-matrix. We say $A$ is \emph{nested} if there is a consecutive-ones ordering for the rows and every two rows are disjoint or nested.
\end{defn}

\begin{defn} \label{def:nested_g}
A split graph $G = (K,S)$ is \emph{nested} if and only if $A(S,K)$ is a nested matrix.
\end{defn}


\begin{defn} \label{def:enriched_matrix}
    Let $A$ be a $(0,1)$-matrix. We say $A$ is an \emph{enriched matrix} if all of the following conditions hold:
    \begin{enumerate}[noitemsep,nolistsep]
        \itemsep0em
        \item Each row of $A$ is either unlabeled or labeled with one of the following labels: L or R or LR. We say that a row is an \emph{LR-row  (resp.\ L-row, R-row)} if it is labeled with LR (resp.\ L, R). An unlabeled row is called U-row.
        \item Each row of $A$ is either uncolored or colored with either blue or red.
        \item The only colored rows may be those labeled with L or R, and those LR-rows having a $0$ in every column.
        \item The LR-rows having a $0$ in every column are all colored with the same color.
    \end{enumerate}
\end{defn}

The $0$-gem, $1$-gem and $2$-gem are the following enriched matrices:
 \[  \scriptsize{    \begin{pmatrix}
         1 1 0 \cr
        0 1 1
        \end{pmatrix}, \qquad
        \bordermatrix{ & \cr
         & 1 0 \cr
         & 1 1 }\ , \qquad
        \bordermatrix{ & \cr
        \textbf{LR} & 1 1 0 \cr
        \textbf{LR} & 1 0 1  }\ } \]
respectively.

\begin{defn}  \label{def:gems}
    Let $A$ be an enriched matrix. We say that $A$ \emph{contains a gem} (resp.\ \emph{doubly-weak gem}) if it contains a $0$-gem (resp.\ a $2$-gem) as a subconfiguration.
    We say that $A$ \emph{contains a weak gem} if it contains a $1$-gem such that, either the first is an L-row (resp.\ R-row) and the second is a U-row, or the first is an LR-row and the second is a non-LR-row.
    We say that a $2$-gem is \emph{badly-colored }if the entries in the column in which both rows have a $1$ are in blocks colored with the same color.
\end{defn}

\begin{teo}[\cite{P20, PDS20}] \label{teo:nested_caract}
A $(0,1)$-matrix is nested if and only if it contains no $0$-gem as a subconfiguration. 
\end{teo}


\begin{defn} \label{def:LR-orderable}
    Let $A$ be an enriched matrix. We say $A$ is \emph{LR-orderable} if there is a linear ordering $\Pi$ for the columns of $A$ such that each of the following assertions holds:
    \begin{itemize}
        \item $\Pi$ is a consecutive-ones ordering for every non-LR row of $A$.
        \item The ordering $\Pi$ is such that the ones in every nonempty row labeled with L (resp.\ R) start in the first column (resp.\ end in the last column).
        \item $\Pi$ is a consecutive-ones ordering for the complements of the LR-rows of $A$.
     \end{itemize}
Such an ordering is called an \emph{LR-ordering}.
\end{defn}

\begin{defn} \label{def:suitable_ordering}
For each row of $A$ labeled with L or LR and having a $1$ in the first column of $\Pi$, we define its \emph{L-block (with respect to $\Pi$)} as the maximal set of consecutive columns of $\Pi$ starting from the first one on which the row has a 1. \emph{R-blocks} are defined on an entirely analogous way.
For each unlabeled row of $A$, we say its \emph{U-block (with respect to $\Pi$)} is the set of columns having a $1$ in the row.
The blocks of $A$ with respect to $\Pi$ are its L-blocks, its R-blocks and its U-blocks. We say an \emph{L-block (resp.\ R-block, U-block) is colored} if there is a $1$-color assignment for every entry of the block.

    An LR-ordering $\Pi$ is \emph{suitable} if the L-blocks of those LR-rows with exactly two blocks are disjoint with every R-block, the R-blocks of those LR-rows with exactly two blocks are disjoint with the L-blocks and for each LR-row the intersection with any U-block is empty with either its L-block or its R-block.
\end{defn}

\begin{defn} \label{def:A*}
    Let $A$ be an enriched matrix and let $\Pi$ be a LR-ordering.
    We define $A^*$ as the enriched matrix that arises from $A$ by:
\begin{itemize}
    \item Replacing each LR-row by its complement.
    \item Adding two distinguished rows: both rows have a $1$ in every column, one is labeled with L and the other is labeled with R.
\end{itemize}
\end{defn}

\begin{defn}  \label{def:tagged_matrixA}
A \emph{tagged matrix} is a $(0,1)$-matrix, each of whose rows are either uncolored or colored with blue or red, together with a set of at most two distinguished columns. The distinguished columns will be referred to as \emph{tag columns}.
Let $A$ be an enriched matrix. We define the \emph{tagged matrix of $A$} as a tagged matrix, denoted by $A_\tagg$, whose underlying matrix is obtained from $A$ by adding two columns, $c_L$ and $c_R$, such that:
(1) the column $c_L$ has a $1$ if $f$ is labeled L or LR and $0$ otherwise,
(2) the column $c_R$ has a $1$ if $f$ is labeled R or LR and $0$ otherwise, and
(3) the set of distinguished columns of $A_{\tagg}$ is $\{ c_L, c_R\}$.
We denote $A^*_\tagg$ to the tagged matrix of $A^*$. By simplicity we will consider column $c_L$ as the first and column $c_R$ as the last column of $A_\tagg$ and $A^*_\tagg$.
\end{defn}

Notice that for every enriched matrix, the only colored rows are those labeled with L or R and those empty LR-rows. Moreover, for every LR-orderable matrix, there is an ordering of the columns such that every row labeled with L (resp.\ R) starts in the first column (resp.\ ends in the last column), and thus all its $1$'s appear consecutively.
Thus, if an enriched matrix is also LR-orderable, then the given coloring induces a partial block bi-coloring, in which every empty LR-row remains the same, whereas for every nonempty colored labeled row, we color all its $1$'s with the color given in the definition of the matrix.

\begin{defn} \label{def:2-nested}
    Let $A$ be an enriched matrix. We say $A$ is \emph{$2$-nested} if there exists an LR-ordering $\Pi$ of the columns and an assignment of colors red or blue to the blocks of $A$ such that all of the following conditions hold:
\begin{enumerate}[noitemsep,nolistsep]
    \item If an LR-row has an L-block and an R-block, then they are colored with distinct colors. \label{item:2nested1}
    \item For each colored row $r$ in $A$, any of its blocks is colored with the same color as $r$ in $A$. \label{item:2nested2}
    \item If an L-block of an LR-row is properly contained in the L-block of an L-row, then both blocks are colored with different colors. \label{item:2nested3}
    \item Every L-block of an LR-row and any R-block are disjoint. The same holds for an R-block of an LR-row and any L-block. \label{item:2nested4}
    \item If an L-block and an R-block are not disjoint, then they are colored with distinct colors.    \label{item:2nested5}
    \item Each two U-blocks colored with the same color are either disjoint or nested. \label{item:2nested6}
    \item If an L-block and a U-block are colored with the same color, then either they are disjoint or the U-block is contained in the L-block. The same holds replacing L-block for R-block. \label{item:2nested7}
    \item If two distinct L-blocks of non-LR-rows are colored with distinct colors, then every LR-row has an L-block. The same holds replacing L-block for R-block.  \label{item:2nested8}
    \item If two LR-rows overlap, then the L-block of one and the R-block of the other are colored with the same color. \label{item:2nested9}
\end{enumerate}

An assignment of colors red and blue to the blocks of $A$ that satisfies all these properties is called a \emph{(total) block bi-coloring}.
\end{defn}

These matrices admit the following characterization by forbidden subconfigurations.

\begin{teo}[\cite{P20, PDS20}]\label{teo:2-nested_caract_bymatrices}
    Let $A$ be an enriched matrix. Then, $A$ is $2$-nested if and only if $A$ contains none of the following listed matrices or their dual matrices as subconfigurations:
    \begin{itemize}
    \item $M_0$, $M_{II}(4)$, $M_V$ or $S_0(k)$ for every even $k \geq 4$ (See Figure~\ref{fig:forb_M_chiquitas})
    \item Every enriched matrix in the family $\mathcal{D}$ (See Figure~\ref{fig:forb_D}) 
     \item Every enriched matrix in the family $\mathcal{F}$ (See Figure~\ref{fig:forb_F}) 
    \item Every enriched matrix in the family $\mathcal{S}$ (See Figure~\ref{fig:forb_S})
    \item Every enriched matrix in the family $\mathcal{P}$ (See Figure~\ref{fig:forb_P})
    \item Monochromatic gems, monochromatic weak gems, badly-colored doubly-weak gems
    \end{itemize}
and $A^*$ contains no Tucker matrices and none of the enriched matrices in $\mathcal{M}$ or their dual matrices as subconfigurations. (See Figure~\ref{fig:forb_LR-orderable}).
\end{teo}

\vspace{-7mm}
\begin{figure}[H]
    \centering
    \footnotesize{
    \begin{align*}
            M_0 = \scriptsize{ \bordermatrix{ & \cr
                         & 1 0 1 1 \cr
                         & 1 1 1 0  \cr
                         & 0 1 1 1 }\ }
            &&
            M_{II}(4) = \scriptsize{ \bordermatrix{ & \cr
                         & 0 1 1 1 \cr
                         & 1 1 0 0  \cr
                         & 0 1 1 0 \cr
                         & 1 1 0 1 }\ }
            &&
            M_V = \scriptsize{ \bordermatrix{ & \cr
                         & 1 1 0 0 0 \cr
                         & 0 0 1 1 0  \cr
                         & 1 1 1 1 0 \cr
                         & 1 0 0 1 1 }\ }
            &&
            S_0(k)&= \scriptsize{ \begin{pmatrix}
                1 11...11\\
                110...00\\
                011...00\\
                .   .   .   .   . \\
                .   .   .   .   . \\
                .   .   .   .   . \\
                000...11\\
                100...01\\
            \end{pmatrix} }
    \end{align*}}
    \caption{The matrices $M_0$, $M_{II}(4)$, $M_V$ and $S_0(k) \in \{0,1\}^{((k+1)\times k}$ for any even $k \geq 4$.}
    \label{fig:forb_M_chiquitas}
    \end{figure}

\begin{figure}[H]
    \centering
    \includegraphics[scale=.33]{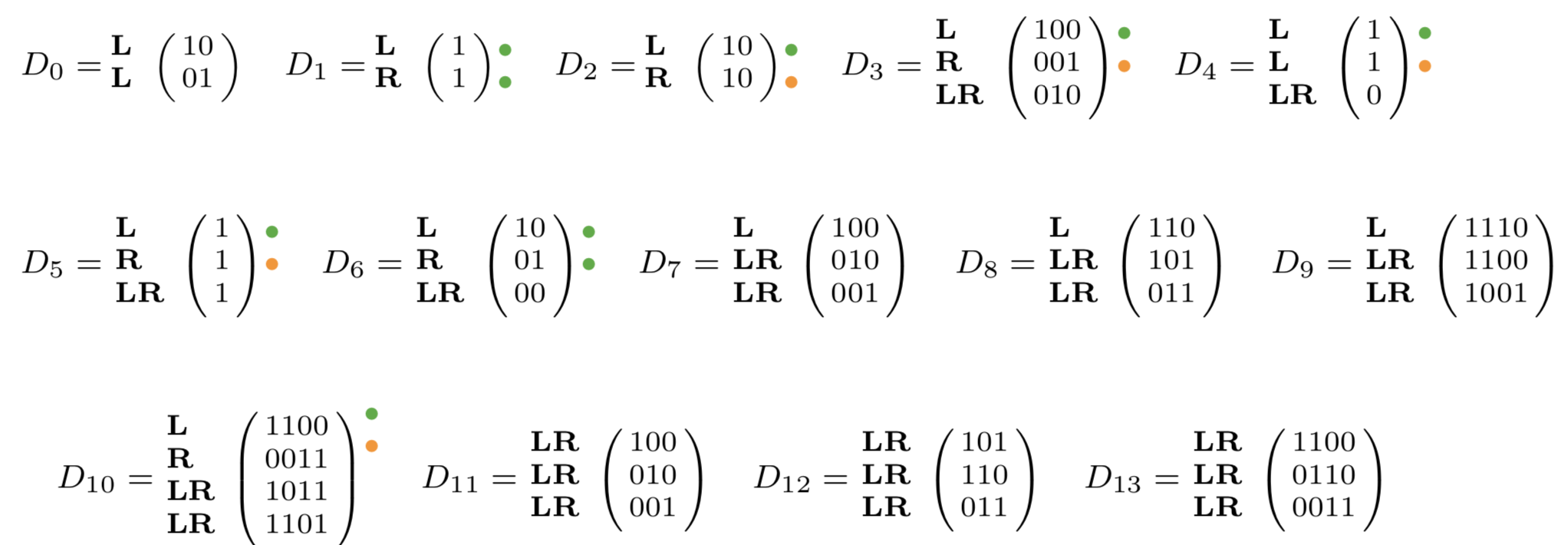}
    \caption{The family of enriched matrices $\mathcal{D}$.}
    \label{fig:forb_D}
\end{figure}

\begin{figure}[H]
\centering
\footnotesize{
    \begin{align*}
            F_0= \scriptsize{  \begin{pmatrix}
                11100\\
                01110\\
                00111\\
            \end{pmatrix} }
            &&
            F_1(k)= \scriptsize{ \begin{pmatrix}
                011...111\\
                111...110\\
                000...011\\
                000...110\\
                .   .   .   .   . \\
                .   .   .   .   . \\
                .   .   .   .   . \\
                110...000\\
            \end{pmatrix} }
            &&
            F_2(k)= \scriptsize{  \begin{pmatrix}
                0111...10\\
                1100...00\\
                0110...00\\
                .   .   .   .   . \\
                .   .   .   .   . \\
                .   .   .   .   . \\
                0000...11\\
            \end{pmatrix}  }
            &&
              F'_0=  \scriptsize{ \bordermatrix{ & \cr
                \textbf{L (LR)} & 1 1 0 0 \cr
                 &1 1 1 0 \cr
                &0 1 1 1 }\ }
            \end{align*}
            \begin{align*}
            F''_0= \scriptsize{ \bordermatrix{ & \cr
                \textbf{L} & 1 1 0 \cr
                 &1 1 1 \cr
                \textbf{R} &0 1 1 }\     }
            &&
            F'_1(k)= \scriptsize{  \bordermatrix{ & \cr
                &11\ldots1111\cr
                \textbf{L (LR)}&11\ldots1110\cr
                &00\ldots0011\cr
                &00\ldots0110\cr
                & \iddots \cr
                \textbf{L (LR)}&10 \ldots 0000 }\ }
            &&
            F'_2(k)= \scriptsize{ \bordermatrix{ & \cr
                &111\ldots10\cr
                \textbf{L (LR)}&100\ldots00\cr
                &110\ldots00\cr
                &  \ddots \cr
                &000\ldots11 }\ }
        \end{align*}}
        \vspace{-2mm}
    \caption{The enriched matrices of the family $\mathcal{F}$.}  \label{fig:forb_F}
\end{figure}

The matrices $\mathcal{F}$ represented in Figure~\ref{fig:forb_F} are defined as follows:
$F_1(k) \in \{0,1\}^{k \times (k-1)}$, $F_2(k) \in \{0,1\}^{k \times k}$, $F'_1(k) \in \{0,1\}^{k \times (k-2)}$ and $F'_2(k) \in \{0,1\}^{k \times (k-1)}$, for every odd $k \geq 5$. In the case of $F'_0$, $F'_1(k)$ and $F'_2(k)$, the labeled rows may be either L or LR indistinctly, and in the case of their dual matrices, the labeled rows may be either R or LR indistinctly.


The matrices $\mathcal{S}$ in Figure~\ref{fig:forb_S} are defined as follows.
If $k$ is odd, then $S_1(k) \in \{0,1\}^{(k+1) \times k}$ for $k \geq 3$, and if $k$ is even, then $S_1(k) \in \{0,1\}^{k \times (k-2)}$ for $k\geq 4$.
The remaining matrices have the same size whether $k$ is even or odd: $S_2(k) \in \{0,1\}^{k \times (k-1)}$ for $k \geq 3$,
$S_3(k) \in \{0,1\}^{k \times (k-1)}$ for $k \geq 3$,
$S_5(k) \in \{0,1\}^{k \times (k-2)}$ for $k \geq 4$,
$S_4(k) \in \{0,1\}^{k \times (k-1)}$, $S_6(k) \in \{0,1\}^{k \times k}$ for $k\geq 4$,
$S_7(k) \in \{0,1\}^{k \times (k+1)}$ for every $k \geq 3$
and $S_8(2j) \in \{0,1\}^{2j \times (2j)}$ for $j \geq 2$.
If $k$ is even, then the first and last row of $S_2(k)$ and $S_3(k)$ are colored with the same color, and in $S_4(k)$ and $S_5(k)$ are colored with distinct colors.

\begin{figure}[H]
    \centering
    \footnotesize{
    \begin{align*}
            S_1(2j) &= \scriptsize{ \bordermatrix{ & \cr
                \textbf{L}& 1 0 \ldots 00  \cr
                              & 1 1 \ldots 00 \cr
                               &   \ddots  \cr
                               & 0 0 \ldots 11 \cr
                \textbf{LR} &0 0  \ldots 01 \cr
                \textbf{L} & 1  1  \ldots  11 }\ }
            &
            S_1(2j+1) &= \scriptsize{  \bordermatrix{ & \cr
                \textbf{L} & 1 0 \ldots 0 0 \cr
                            & 1 1 \ldots 0 0 \cr
                            & \ddots \cr
                            & 0 0 \ldots 1 1 \cr
                \textbf{LR} &0 0 \ldots 0 1 }\ }
            &
        S_2(k) &= \scriptsize{ \bordermatrix{ & \cr
                \textbf{L}& 1 0 \ldots 0 0 \cr
                            & 1 1  \ldots 0 0 \cr
                            &  \ddots \cr
                            &0 0 \ldots 1 1 \cr
                \textbf{L} & 1 1 \ldots 1 0 }\
                \begin{matrix}
             \textcolor{dark-orange}{\bullet} \\ \\ \\ \\ \\ \textcolor{dark-green}{\bullet} \\
                \end{matrix}  }
        \end{align*}
        \begin{align*}
         S_3(k) &=  \scriptsize{ \bordermatrix{ & \cr
             \textbf{L} & 1 0 \ldots 0 0 \cr
                              & 1 1  \ldots 0 0 \cr
                              &  \ddots \cr
                              &0 0 \ldots 1 1 \cr
            \textbf{R} & 0 0 \ldots 0 1 }\
                \begin{matrix}
            \textcolor{dark-orange}{\bullet} \\ \\ \\ \\ \\ \textcolor{dark-green}{\bullet} \\
                \end{matrix} }
                &
        S_4(k) &= \scriptsize{ \bordermatrix{ & \cr
        \textbf{LR}& 1 1 \ldots 1 1 \cr
        \textbf{L} & 1 0 \ldots 0 0 \cr
                        & 1 1  \ldots  0 0 \cr
                        &  \ddots \cr
                        & 0 0  \ldots 1 1 \cr
          \textbf{R}& 0 0 \ldots 0 1  }\
                \begin{matrix}
            \\ \textcolor{dark-orange}{\bullet} \\ \\ \\ \\ \\ \textcolor{dark-green}{\bullet} \\
                \end{matrix} }
        &
        S_5(k) &= \scriptsize{ \bordermatrix{ & \cr
        \textbf{L}& 1 0 \ldots 0 0 \cr
                      & 1 1 \ldots 0 0 \cr
                      &  \ddots \cr
                      & 0 0  \ldots 1 1 \cr
        \textbf{LR}& 1 1 \ldots 1 0 \cr
        \textbf{L}& 1 1 \ldots 1 1  }\
                \begin{matrix}
             \textcolor{dark-orange}{\bullet} \\ \\ \\ \\ \\ \\  \textcolor{dark-green}{\bullet} \\
                \end{matrix} }
        \end{align*}
        \begin{align*}
        S_6(3) &= \scriptsize{ \bordermatrix{ & \cr
        \textbf{LR} & 1 1 0 \cr
          \textbf{R} & 0 1 1  \cr
                          & 1 1 0 }\ }
        &
        S_6'(3) &= \scriptsize{ \bordermatrix{ & \cr
        \textbf{LR} & 1 1 0 \cr
          \textbf{R} & 0 1 1  \cr
                          & 1 1 1 }\ }
         &
        S_6(k) &= \scriptsize{ \bordermatrix{ & \cr
        \textbf{LR}& 1 1 1 \ldots 1 1 0 \cr
          \textbf{R}& 0 1 1 \ldots 1 1 1 \cr
                        & 1 1 0  \ldots 0 0 0  \cr
                        &  \ddots \cr
                        & 0 0 0 \ldots 0 1 1 }\
                \begin{matrix}
            \\ \\ \textcolor{dark-orange}{\bullet} \\ \\ \\ \\ \\ \\
                \end{matrix} }
    \end{align*}
    \begin{align*}
        S_7(3) &= \scriptsize{ \bordermatrix{ & \cr
        \textbf{LR} & 1 1 0 0 1 \cr
        \textbf{LR} & 1 0 0 1 1  \cr
                         & 1 1 1 0 0 }\ }
        &
        S_7(2j) &= \scriptsize{ \bordermatrix{ & \cr
        \textbf{LR}& 1 1 0 0 \ldots 0 0 0 \cr
        \textbf{LR}& 1 0 0 0  \ldots 0 0 1 \cr
                        & 0 1 1 0 \ldots 0 0 0  \cr
                        &  \ddots \cr
                        & 0 0 0 0 \ldots 0 1 1 }\ }
        &
        S_8(2j) &= \scriptsize{ \bordermatrix{ & \cr
        \textbf{LR}& 1 0 0 \ldots 0 0 1 \cr
                        & 1 1 0  \ldots 0 0 0 \cr
                        &  \ddots \cr
                        & 0 0 0 \ldots 0 1 1 }\ }
    \end{align*}}
\caption{The family of matrices $\mathcal{S}$ for every $j\geq2$ and every odd $k \geq 3$} \label{fig:forb_S}
\end{figure}


In the matrices $\mathcal{P}$, the integer $l$ represents the number of unlabeled rows between the first row and the first LR-row. The matrices $\mathcal{P}$ described in Figure~\ref{fig:forb_P} are defined as follow:
$P_0(k,0) \in \{0,1\}^{k \times k}$ for every $k \geq 4$, $P_0(k,l) \in \{0,1\}^{k \times (k-1)}$ for every $k \geq 5$ and $l >0$;
$P_1(k,0) \in \{0,1\}^{k \times (k-1)}$ for every $k \geq 5$, $P_1(k,l) \in \{0,1\}^{k \times (k-2)}$ for every $k \geq 6$, $l > 0$;
$P_2(k,0) \in \{0,1\}^{k \times (k-1)}$ for every $k \geq 7$, $P_2(k,l) \in \{0,1\}^{k \times (k-2)}$ for every $k \geq 8$ and $l > 0$.
If $k$ is even, then the first and last row of every matrix in $\mathcal{P}$ are colored with distinct colors.

\begin{figure}[H]
    \centering
    \footnotesize{
        \begin{align*}
        P_0(k,0) =  \scriptsize{ \bordermatrix{ & \cr
        \textbf{L} & 1 1 0 0 0  \ldots 0 0 0   \cr
        \textbf{LR} & 1 0 0 1 1 \ldots 1 1 1 \cr
                         & 0 0  1  1  0  \ldots 0 0 0 \cr
                         &  \ddots  \cr
                        & 0 0 0 0 0 \ldots  0 1 1  \cr
        \textbf{R} & 0 0 0 0 0 \ldots 0 0 1  }\
     \begin{matrix}
      \textcolor{dark-green}{\bullet} \\ \\ \\ \\ \\ \\ \textcolor{dark-green}{\bullet}
      \end{matrix} }
        &&
        P_0(k,l) = \scriptsize{  \bordermatrix{ & \cr
        \textbf{L}& 1 0 0 \ldots 0 0 0 0 \ldots 0 \cr
                      & 1 1 0 \ldots 0 0 0 0  \ldots 0 \cr
                      &  \ddots \cr
                     & 0 0 0 \ldots 1 1 0 0 \ldots 0 \cr
    \textbf{LR} & 1 1 1 \ldots 1 0 0 1 \ldots 1 \cr
                     & 0 0 0 \ldots 0 0 1 1  \ldots 0 \cr
                     & \ddots \cr
                    & 0 0 0 \ldots 0 0 \ldots 0 1 1 \cr
        \textbf{R} & 0 0 0 \ldots 0 0 \ldots 0 0 1 }\
    \begin{matrix}
      \textcolor{dark-green}{\bullet} \\ \\ \\ \\ \\ \\ \\ \\ \\ \\ \textcolor{dark-green}{\bullet}
      \end{matrix} }
       \end{align*}
    \begin{align*}
        P_1(k,0) = \scriptsize{ \bordermatrix{ & \cr
        \textbf{L} & 1 1 0 0  \ldots 0 0 0   \cr
        \textbf{LR} & 1 0 1 1 \ldots 1 1 1 \cr
        \textbf{LR} & 1 1 0 1 \ldots 1 1 1 \cr
                         & 0 0  1  1  0  \ldots 0 0 0 \cr
                         &  \ddots  \cr
                        & 0 0 0 0 0 \ldots  0 1 1  \cr
        \textbf{R} & 0 0 0 0 \ldots 0 0 1  }\
     \begin{matrix}
      \textcolor{dark-green}{\bullet} \\ \\ \\ \\ \\ \\ \\ \textcolor{dark-green}{\bullet}
      \end{matrix} }
    &&
            P_1(k,l) = \scriptsize{ \bordermatrix{ & \cr
        \textbf{L}& 1 0 0 \ldots 0 0 0 0 \ldots 0 \cr
                      & 1 1 0 \ldots 0 0 0 0  \ldots 0 \cr
                      &  \ddots \cr
                     & 0 0 0 \ldots 1 1 0 0 \ldots 0 \cr
    \textbf{LR} & 1 1 1 \ldots 1 0 1 1 \ldots 1 \cr
    \textbf{LR} & 1 1 1 \ldots 1 1 0 1 \ldots 1 \cr
                     & 0 0 0 \ldots 0 0 1 1  \ldots 0 \cr
                     & \ddots \cr
                    & 0 0 0 \ldots 0 0 \ldots 0 1 1 \cr
        \textbf{R} & 0 0 0 \ldots 0 0 \ldots 0 0 1 }\
    \begin{matrix}
      \textcolor{dark-green}{\bullet} \\ \\ \\ \\ \\ \\ \\ \\ \\ \\ \\ \textcolor{dark-green}{\bullet}
      \end{matrix} }
       \end{align*}
    \begin{align*}
        P_2(k,0) =  \scriptsize{ \bordermatrix{ & \cr
        \textbf{L} & 1 1 0 0 0 0 \ldots 0 0 0 \cr
        \textbf{LR} & 1 0 1 1 1 1 \ldots 1 1 1 \cr
        \textbf{LR} & 1 1 1 0 1 1 \ldots 1 1 1 \cr
        \textbf{LR} & 1 1 0 1 1 1 \ldots 1 1 1 \cr
        \textbf{LR} & 1 1 1 0 0 1 \ldots 1 1 1 \cr
                         & 0 0  0 0 1  1 \ldots 0 0 0 \cr
                         &  \ddots  \cr
                        & 0 0 0 0 0 \ldots  0 1 1  \cr
        \textbf{R} & 0 0 0 0 0 \ldots 0 0 1  }\
     \begin{matrix}
      \textcolor{dark-green}{\bullet} \\ \\ \\ \\ \\ \\ \\ \\ \\ \\ \textcolor{dark-green}{\bullet}
      \end{matrix} }
    &&
            P_2(k,l) = \scriptsize{ \bordermatrix{ & \cr
        \textbf{L}& 1 0 0 \ldots 0 0 0 0 0 \ldots 0 \cr
                      & 1 1 0 \ldots 0 0 0 0 0  \ldots 0 \cr
                      &  \ddots \cr
                     & 0 0 0 \ldots 1 1 0 0 0  \ldots 0 \cr
    \textbf{LR} & 1 1 1 \ldots 1 0 0 1 1 \ldots 1 \cr
    \textbf{LR} & 1 1 1 \ldots 1 1 1 0 1 \ldots 1 \cr
    \textbf{LR} & 1 1 1 \ldots 1 1 0 1 1 \ldots 1 \cr
    \textbf{LR} & 1 1 1 \ldots 1 1 0 0 1 \ldots 1 \cr
                     & 0 0 0 \ldots 0 0 0 1 1 \ldots 0 \cr
                     & \ddots \cr
                    & 0 0 0 \ldots 0 0 0 \ldots 0 1 1 \cr
        \textbf{R} & 0 0 0 \ldots 0 0 0  \ldots 0 0 1 }\
    \begin{matrix}
      \textcolor{dark-green}{\bullet} \\ \\ \\ \\ \\ \\ \\ \\\ \\ \\ \\ \\ \\ \\ \textcolor{dark-green}{\bullet}
      \end{matrix} }
    \end{align*}}
    \caption{The family of enriched matrices $\mathcal{P}$ for every odd $k$.}
    \label{fig:forb_P}
    \end{figure}


\begin{figure}[h]
    \centering
  \footnotesize{
    \begin{align*}
            M_2'(k) = \scriptsize{ \bordermatrix{ & \cr
                        & 1 1 1 \ldots 1 1 1\cr
           \textbf L &  1 0 0 \ldots 0 0 0\cr
                        & 1  1 0  \ldots  0 0 0\cr
                        & \ddots  \cr
                        & 0 0 0 \ldots 1 1 0 \cr
           \textbf L & 1 1 1 \ldots 1 0 1 }\ }
            &&
            M_2''(k) = \scriptsize{ \bordermatrix{ & \cr
        \textbf R & 1  1  1 \ldots  1 1 1\cr
        \textbf L & 1  0 0 \ldots  0 0 0\cr
                    & 1  1 0  \ldots  0 0 0\cr
                        & \ddots \cr
                        & 0 0 0  \ldots  1 1 0 \cr
         \textbf R &  0 0 0 \ldots  0 1  0 \cr
     \textbf{L} & 1  1 1 \ldots  1 1  1  }\  }
             \end{align*}
            \begin{align*}
            M_3'(k) = \scriptsize{ \bordermatrix{ & \cr
            \textbf L &  1 0  0   \ldots   0   0   0 \cr
                &   1  1 0   \ldots   0   0   0 \cr
                & \ddots \cr
                & 0 0  0   \ldots   1   1   0 \cr
                & 1 1 1   \ldots   1   0   1  }\     }
            &&
            M_3''(k) = \scriptsize{ \bordermatrix{ & \cr
                & 1   1  0  \ldots  0 0  \cr
                & 0  1  1   \ldots  0 0  \cr
                & \ddots \cr
                & 0 0 0   \ldots  1 1 \cr
                \textbf R & 0 1 1  \ldots  1  0  }\ }
            &&
            M_4' =  \scriptsize{ \bordermatrix{ & \cr
                    \textbf L & 1  0  0  0 0 \cr
                    & 0  1   1   0  0 \cr
                    & 0  0  0  1   1 \cr
                    & 1  0  1  0  1 }\ }
            \end{align*}
            \begin{align*}
            M_4'' = \scriptsize{  \bordermatrix{ & \cr
                    \textbf L & 1  0  0  0 \cr
                    \textbf R & 0  1  0  0 \cr
                                & 0  0  1  1 \cr
                                 &  1  1   0   1 }\ }
            &&
            M_5' = \scriptsize{ \bordermatrix{ & \cr
                    & 1   1   0   0   \cr
                    & 0   0   1   1   \cr
                    \textbf R & 1   0   0   1  \cr
                    & 1   1   1   }\ }
            &&
            M_5'' = \scriptsize{ \bordermatrix{ & \cr
                    \textbf L & 1  0   0   0 \cr
                    & 0  1   1   0 \cr
                    & 1   0   1   1 \cr
                    \textbf L & 1  1   1  0 }\ }
    \end{align*}
       }
    \caption{The enriched matrices in family $\mathcal{M}$: $M_{2}'(k)$, $M_{3}'(k)$, $M_{3}''(k)$, $M_{3}'''(k)$ for $k \geq 4$, and $M_{2}''(k)$ for $k \geq 5$. }  \label{fig:forb_LR-orderable}
\end{figure}

The following lemmas, theorem and definition will be useful in the sequel. 
\begin{lema}[\cite{P20, PDS20}]  \label{lema:2-nested_if}
    Let $A$ be an enriched matrix. Then, $A$ is \emph{$2$-nested} if $A$ with its partial block bi-coloring contains none of the matrices listed in Theorem~\ref{teo:2-nested_caract_bymatrices} and the given partial block bi-coloring of $A$ can be extended to a total block bi-coloring of $A$.
\end{lema}

\begin{lema}[\cite{P20, PDS20}] \label{lema:B_ext_2-nested}
Let $A$ be an enriched matrix. If $A$ with its partial block bi-coloring contains none of the matrices listed in Theorem~\ref{teo:2-nested_caract_bymatrices} and $B$ is obtained from $A$ by extending its partial coloring to a total block bi-coloring, then $B$ is $2$-nested if and only if for each LR-row its L-block and R-block are colored with distinct colors and $B$ contains no monochromatic gems, monochromatic weak gems or badly-colored doubly-weak gems as subconfigurations.
\end{lema}

\begin{teo}[\cite{P20, PDS20}] \label{teo:hay_suitable_ordering}
If $A$ is admissible, LR-orderable and contains no $M_0$, $M_{II}(4)$, $M_V$ or $S_0(k)$ for every even $k \geq 4$, then there is at least one suitable LR-ordering.
\end{teo}

\begin{defn} \label{def:admissibility}
An enriched matrix $A$ is admissible if $A$ is $\{\mathcal{D}, \mathcal{S}, \mathcal{P}\}$-free. 
\end{defn}

The proof of Theorem~\ref{teo:circle_split_caract} is organized as follows. In each of the following sections we consider a split graph $G$ that contains an induced subgraph $H$, where $H$ is either a tent, a $4$-tent, a co-$4$-tent or a net (Sections~\ref{subsec:circle2},~\ref{subsec:circle3},~\ref{subsec:circle4} and~\ref{subsec:circle5}, respectively), and each of these is a case of the proof of Theorem~\ref{teo:circle_split_caract}.
Using the partitions of $K$ and $S$ described in the previous section, we define one enriched $(0,1)$-matrix for each partition $K_i$ of $K$ and four auxiliary non-enriched $(0,1)$-matrices that will help us give a circle model for $G$. At the end of each section, we prove that $G$ is circle if and only if these enriched matrices are $2$-nested and the four non-enriched matrices are nested, giving the guidelines for a circle model in each case.
From now on, we will refer indistinctly to a row $r$ (resp.\ a column $c$) of a matrix and the vertex in the independent (resp.\ complete) set of the split partition of $G$ whose adjacency is represented by the row (resp.\ column).

\subsection{Split circle graphs containing an induced tent} \label{subsec:circle2}

In this section we address the first case of the proof of Theorem~\ref{teo:circle_split_caract}, which is the case where $G$ contains an induced tent.
This section is subdivided as follows. In Section~\ref{subsubsec:tent2}, we use the partitions of $K$ and $S$ given in Section~\ref{subsec:tent_partition} to define the matrices $\mathbb{A}_i$ for each $i=1, 2, \ldots, 6$ and prove some properties that will be useful further on.
In Section~\ref{subsubsec:tent3}, the main results are the necessity of the $2$-nestedness of each $\mathbb A_i$ for $G$ to be a \fsc-free graph and the guidelines to obtain a circle model for a \fsc-free split graph $G$ containing an induced tent given in Theorem~\ref{teo:finalteo_tent}.


\subsubsection{Matrices $\mathbb A_1,\mathbb A_2,\ldots,\mathbb A_6$} \label{subsubsec:tent2}

Let $G=(K,S)$ and $H$ as in Section~\ref{subsec:tent_partition}. For each $i\in\{1,2,\ldots,6\}$, let $\mathbb A_i$ be an enriched $(0,1)$-matrix having one row for each vertex $s\in S$ such that $s$ belongs to $S_{ij}$ or $S_{ji}$ for some $j\in\{1,2,\ldots,6\}$, and one column for each vertex $k\in K_i$ and such that the entry corresponding to the row $s$ and the column $k$ is $1$ if and only if $s$ is adjacent to $k$ in $G$. For each $j\in\{1,2,\ldots,6\}-\{i\}$, we mark those rows corresponding to vertices of $S_{ji}$ with L and those corresponding to vertices of $S_{ij}$ with R.

Moreover, we color some of the rows of $\mathbb A_i$ as follows.
\begin{itemize}
 \item If $i\in\{1,3,5\}$, then we color each row corresponding to a vertex $s\in S_{ij}$ for some $j\in\{1,2,\ldots,6\}-\{i\}$ with color red and each row corresponding to a vertex $s\in S_{ji}$ for some $j\in\{1,2,\ldots,6\}-\{i\}$ with color blue.
 \item If $i\in\{2,4,6\}$, then we color each row corresponding to a vertex $s\in S_{ij}\cup S_{ji}$ for some $j\in\{1,2,\ldots,6\}$ with color red if $j=i+1$ or $j=i-1$ (modulo $6$) and with color blue otherwise.
\end{itemize}

Example:
\[ \mathbb A_3 = \scriptsize{ \bordermatrix{ & K_3\cr
        S_{34}\ \textbf{R} & \cdots \cr
                S_{35}\ \textbf{R} & \cdots \cr
                S_{33}\            & \cdots \cr
                S_{13}\ \textbf{L} & \cdots \cr
                S_{23}\ \textbf{L} & \cdots}\
                \begin{matrix}
                \textcolor{red}{\bullet} \\ \textcolor{red}{\bullet} \\ \\ \textcolor{blue}{\bullet} \\ \textcolor{blue}{\bullet}
                \end{matrix} } \qquad\qquad
   \mathbb A_4 = \scriptsize{ \bordermatrix{ & K_4\cr
        S_{34}\ \textbf{L} & \cdots \cr
                S_{45}\ \textbf{R} & \cdots\cr
                S_{44}\            & \cdots \cr
                S_{14}\ \textbf{L} & \cdots \cr
                S_{64}\ \textbf{L} & \cdots \cr
                S_{41}\ \textbf{R} & \cdots \cr
                S_{42}\ \textbf{R} & \cdots }\
                \begin{matrix}
                \textcolor{red}{\bullet} \\ \textcolor{red}{\bullet} \\ \\ \textcolor{blue}{\bullet} \\ \textcolor{blue}{\bullet} \\ \textcolor{blue}{\bullet} \\ \textcolor{blue}{\bullet}
                \end{matrix} } \]

\begin{figure}[h!]
  \begin{subfigure}[b]{0.5\textwidth}
    \includegraphics[width=\textwidth]{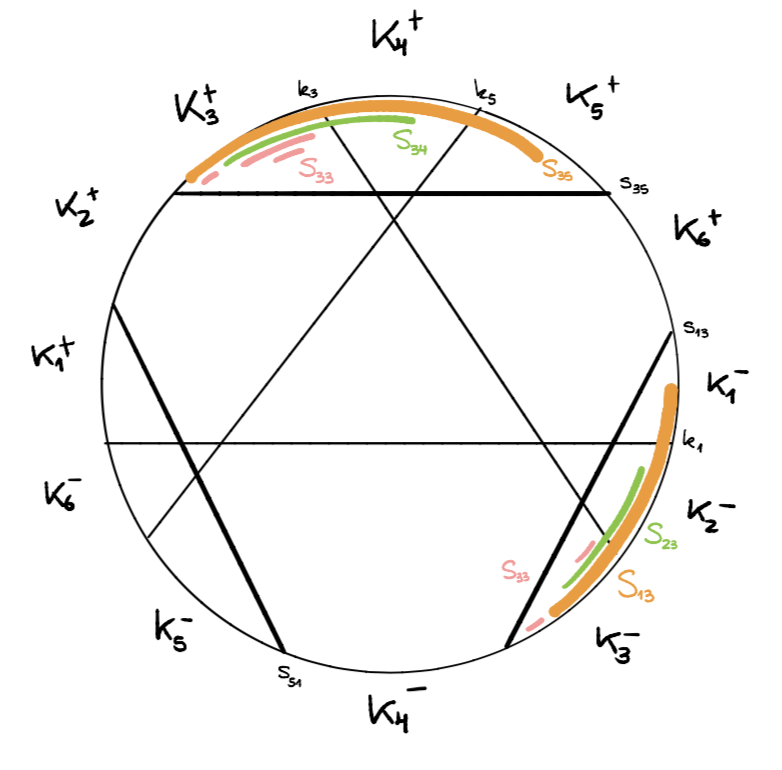}
    \caption{$\mathbb{A}_3$}
    \label{fig:modelA3}
  \end{subfigure}
  \hfill
  \begin{subfigure}[b]{0.48\textwidth}
    \includegraphics[width=\textwidth]{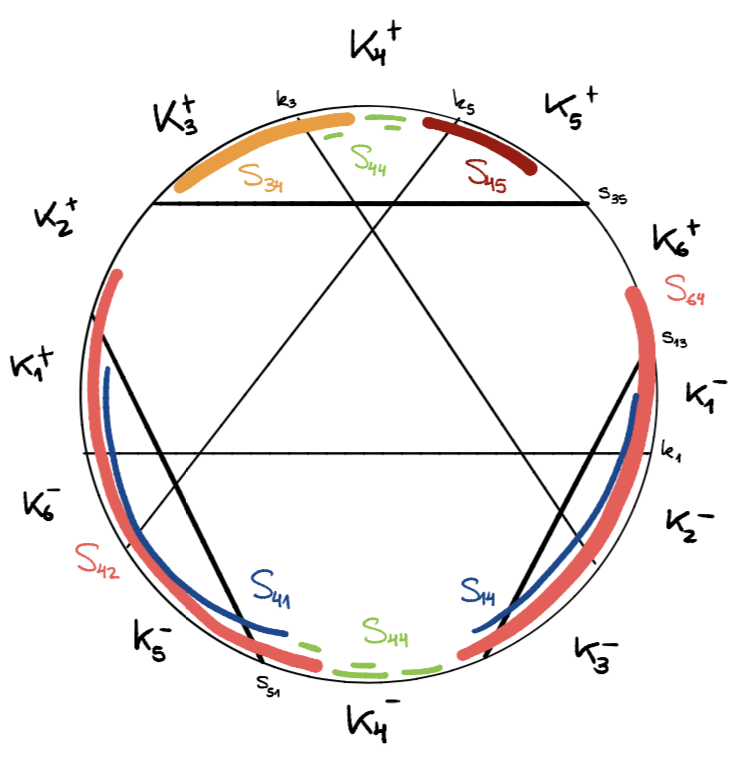}
    \caption{$\mathbb{A}_4$}
    \label{fig:f2}
  \end{subfigure}
  \caption{Sketch model of $G$ with some of the chords associated to rows in $\mathbb{A}_3$ and $\mathbb{A}_4$.}
\end{figure}

The following results are useful in the sequel.

\begin{claim} \label{claim:tent_0} 
Let $v_1$ in $S_{ij}$ and $v_2$ in $S_{ik}$, for $i,j,k \in \{1, 2, \ldots, 6 \}$ such that $i \neq j,k$. If $\mathbb A_i$ contains no $D_0$ for each $i \in \{1, 2, \ldots, 6\}$, then the following assertions hold:
    \begin{itemize}
        \item If $j \neq k$, then $v_1$ and $v_2$ are nested in $K_i$. Moreover, if $j=k$, then $v_1$ and $v_2$ are nested in both $K_i$ and $K_j$.
        \item For each $i \in \{1, 2, \ldots, 6 \}$, there is a vertex $v^*_i$ in $K_i$ such that for every $j \in \{1, 2, \ldots, 6 \} - \{i\}$ and every $s$ in $S_{ij}$, the vertex $s$ is adjacent to $v^*_i$.
    \end{itemize}
\end{claim}

Let $v_1$, $v_2$ in $S_{ij}$, for some $i, j \in \{1, \ldots, 6\}$. Towards a contradiction, suppose without loss of generality that $v_1$ and $v_2$ are not nested in $K_i$. 
Since $v_1$ and $v_2$ are both adjacent to at least one vertex in $K_i$, then there are vertices $w_1$, $w_2$ in $K_i$ such that $w_1$ is adjacent to $v_1$ and nonadjacent to $v_2$, and $w_2$ is adjacent to $v_2$ and nonadjacent to $v_1$.
Moreover, since $v_1$ and $v_2$ lie in $S_{ij}$ and $i \neq j$, it follows from the definition of $\mathbb A_i$ that the corresponding rows are labeled with the same letter and colored with the same color.
Therefore, we find $D_0$ induced by the rows corresponding to $v_1$ and $v_2$, and the columns $w_1$ and $w_2$, which results in a contradiction. 
The proof is analogous by symmetry for $K_j$ and if $j\neq k$. Moreover, the second statement of the claim follows from the previous argument and the fact that there is a C$1$P for the columns of $\mathbb A_i$. \QED


\subsubsection{Split circle equivalence} \label{subsubsec:tent3}

In this section, we will use the matrix theory developed in~\cite{P20,PDS20} to characterize the forbidden induced subgraphs that arise in a split graph that contains an induced tent when this graph is not a circle graph. We will start by proving that, given a split graph $G$ that contains an induced tent, if $G$ is \fsc-free, then $\mathbb A_i$ is $2$-nested for each $i=1, 2, \ldots, 6$.

\begin{lema} \label{lema:no2nested_prohibido_tent}
    If $\mathbb A_i$ is not $2$-nested, for some $i \in \{ 1, \ldots, 6 \}$, then $G$ contains an induced subgraph of the families depicted in Figure~\ref{fig:forb_graphs}. 
\end{lema}

\begin{proof}

Based on the symmetry of the subsets $K_i$ of $K$, it suffices to see what happens when $i = 3$ or $i= 4$, since the proof depends solely on the parity of $i$. 


The proof is organized as follows. First, we assume that $\mathbb A_i$ is not admissible, thus $A$ contains one of the forbidden subconfigurations in $\mathcal{D}$, $\mathcal{S}$ or $\mathcal{P}$. Once we reach a contradiction, we will assume that $\mathbb A^*_i$ contains either a Tucker matrix or one of the forbidden subconfigurations in $\mathcal{M}$, once again reaching a contradiction. The next step is to assume that $\mathbb A_i$ contains no monochromatic gems, monochromatic weak-gems nor badly colored doubly-weak-gems, and finally that $\mathbb A_i$ is not $2$-nested. 

Notice that, if $G$ is \fsc-free, then each $\mathbb A_i$ ($i=1, \ldots, 6$) contains no $M_0$ ($3$-sun with center), $M_{II}(4)$, $M_V$ or $S_0(k)$ for every even $k \geq 4$ ($k$-sun), since these matrices are matrices $A(S,K)$ of graphs in \fsc. Hence, we assume that $\mathbb A_i$ contains none of the matrices depicted in Figure~\ref{fig:forb_M_chiquitas}.

\begin{mycases}
\case \textit{Suppose first that $\mathbb A_i$ is not admissible.} Since $\mathbb{A}_i$ contains no LR-rows, then $\mathbb A_i$ contains either $D_0$, $D_1$, $D_2$ or $S_2(k)$, $S_3(k)$ for some $k \geq 3$.

\subcase \textit{Suppose $\mathbb A_i$ contains $D_0$. }
Let $v_0$ and $v_1$ in $S$ be the vertices whose adjacency is represented by the first and second row of $D_0$, respectively, and let $k_{i1}$ and $k_{i2}$ in $K_i$ be the vertices whose adjacency is represented by the first and second column of $D_0$, respectively. Both rows of $D_0$ are labeled with the same letter and the coloring given to each row is indistinct. We assume without loss of generality that both rows are labeled with L, due to the symmetry of the problem.

\subsubcase \textit{Suppose first that $i = 3$.} In this case, $v_1$ and $v_2$ lie in $S_{34}$ or $S_{35}$. 
By Claim~\ref{claim:tent_0} there is a vertex $k_4$ in $K_4$ (resp.\ $k_5$ in $K_5$) adjacent to every vertex in $S_{34}$ (resp.\ $S_{35}$).
Thus, if both $v_1$ and $v_2$ lie in $S_{35}$, 
then we find a $M_{III}(3)$ induced by $\{ k_5$, $k_{31}$, $k_{32}$, $v_1$, $v_2$, $s_{51}$, $k_1 \}$. If instead both $v_1$ and $v_2$ lie in $S_{34}$, then we find a $3$-sun with center induced by $\{ k_4$, $k_{31}$, $k_{32}$, $v_1$, $v_2$, $s_{35}$, $s_{13} \}$.

Suppose that $v_1$ in $S_{34}$ and $v_2$ in $S_{35}$. Let $k_4$ in $K_4$ adjacent to $v_1$, 
let $k_5$ in $K_5$ adjacent to $v_2$ and let $k_1$ be any vertex in $K_1$.
Thus, $v_1$ and $v_2$ are nonadjacent to $k_1$ and $v_1$ is nonadjacent to $k_5$. Hence, we find a $4$-sun induced by the set $\{ s_{13}$, $s_{51}$, $ v_1$, $v_2$, $k_1$, $k_{31}$, $k_4$, $k_5 \}$.

\subsubcase
\textit{Suppose now that $i = 4$. }Thus, the vertices $v_1$ and $v_2$ belong to either $S_{34}$, $S_{14}$ or $S_{64}$.
Suppose $v_1$ in $S_{34}$ and $v_2$ in $S_{14}$, and let $k_1$ in $K_1$ and $k_3$ in $K_3$ such that $v_1$ is adjacent to $k_3$. Since $v_2$ is complete to $K_3$, then $v_2$ is adjacent to $k_3$, and both $v_1$ and $v_2$ are nonadjacent to $k_1$. Hence, we find $M_{II}(4)$ induced by $\{ s_{13}$, $s_{35}$, $v_1$, $v_2$, $k_3$, $k_{41}$, $k_{42}$, $k_1 \}$. The same holds if $v_2$ lies in $S_{64}$.

If instead $v_1$ and $v_2$ lie in $S_{34}$, then we find a $M_{III}(3)$ induced by the set $\{ k_3$, $k_{41}$, $k_{42}$, $v_1$, $v_2$, $s_{13}, k_1 \}$.

Finally, if $v_1$ and $v_2$ lie in $S_{14} \cup S_{64}$, then we find a tent${}\vee{} K_1$ induced by $\{ k_3$, $k_1$, $k_{41}$, $k_{42}$, $v_1$, $v_2$, $s_{35} \}$, where $k_1$ in $K_1$ is adjacent to $v_1$ and $v_2$.

\subcase \textit{Suppose $\mathbb A_i$ contains $D_1$. }
Both rows of $D_1$ are labeled with distinct letters and are colored with the same color. Let $v_1$ and $v_2$ in $S$ be the vertices whose adjacency is rep\-re\-sent\-ed by the first and second row of $D_1$, respectively, and let $k_i$ in $K_i$ be the vertex whose adjacency is represented by the column of $D_1$.
We assume without loss of generality that $v_1$ is labeled with L and $v_2$ is labeled with R.
It follows from the definition of $\mathbb A_i$ that, if $i$ is odd, then there are no two rows labeled with distinct letters and colored with the same color, thus we assume that $i$ is even and hence $i=4$. 

In this case, either $v_1$ in $S_{34}$ and $v_2$ in $S_{45}$, or $v_1$ in $S_{14} \cup S_{64}$ and $v_2$ in $S_{41} \cup S_{42}$.

If $v_1$ in $S_{34}$ and $v_2$ in $S_{45}$, then we find a $4$-sun induced by $\{ v_1$, $v_2$, $s_{13}$, $s_{51}$, $k_1$, $k_3$, $k_4$, $k_5 \}$, where $k_3$ in $K_3$ is adjacent to $v_1$ and nonadjacent to $v_2$, $k_4$ in $K_4$ is adjacent to both $v_1$ and $v_2$, $k_5$ in $K_5$ is adjacent to $v_2$ and nonadjacent to $v_1$, and $k_1$ in $K_1$ is nonadjacent to both $v_1$ and $v_2$.

Suppose that $v_1$ lies in $S_{14}$ and $v_2$ lies in $S_{41}$. We find a tent${}\vee{}K_1$ induced by $\{ v_1$, $v_2$, $s_{35}$, $k_1$, $k_3$, $k_4$, $k_5 \}$, where $k_1$, $k_3$, $k_4$ and $k_5$ are vertices analogous as those described in the previous paragraph. Analogously, we find the same forbidden induced subgraph in $G$ if $v_1$ in $S_{64}$ or $v_2$ in $S_{42}$.

\subcase \textit{Suppose $D_2$ in $\mathbb A_i$. }
Let $v_1$ and $v_2$ in $S$ be the vertices whose adjacency is represented by the first and second row of $D_2$, respectively, and let $k_{i1}$ and $k_{i2}$ in $K_i$ be the vertices whose adjacency is represented by the first and second column of $D_2$, respectively.
Both rows of $D_2$ are labeled with distinct letters and colored with distinct colors, for the ``same color'' case is covered since we proved that there is no $D_1$ as a submatrix of $\mathbb A_i$. We assume without loss of generality that $v_1$ is labeled with L and $v_2$ is labeled with R.

\subsubcase \textit{Suppose that $i= 4$.} Thus, $v_1$ in $S_{34}$ and $v_2$ in $S_{41} \cup S_{42}$.
We find a $3$-sun with center induced by $\{ v_1$, $v_2$, $s_{13}$, $k_1$, $k_3$, $k_{41}$, $k_{42} \}$, where $k_1$ in $K_1$ is adjacent to $v_2$ and nonadjacent to $v_1$ and $k_3$ in $K_3$ is adjacent to $v_1$ and nonadjacent to $v_2$. We find the same forbidden induced subgraph if $v_2$ in $S_{41}$ or $S_{42}$.

\subsubcase \textit{Suppose that $i=3$.} In this case, $v_1$ in $S_{13} \cup S_{23}$, and $v_2$ in $S_{34} \cup S_{35}$.

Suppose first that $K_2 \neq \emptyset$. If $v_1$ in $S_{23}$ and $v_2$ in $S_{34}$, then we find $M_{II}(4)$ induced by $\{ v_1$, $v_2$, $s_{13}$, $s_{35}$, $k_2$, $k_4$, $k_{31}$, $k_{32} \}$. 
If instead $v_2$ in $S_{35}$, then we find $M_{II}(4)$ induced by the same subset of vertices with the exception of $k_4$, considering an analogous vertex $k_5$ in $K_5$. 
Moreover, the same forbidden induced subgraph can be found if $v_1$ in $S_{13}$, as long as $K_2 \neq \emptyset$.

If instead $K_2 = \emptyset$, then necessarily $v_1$ in $S_{13}$. If $v_2$ in $S_{35}$, then we find a $3$-sun with center induced by the subset $\{ v_1$, $v_2$, $s_{51}$, $k_1$, $k_5$, $k_{31}$, $k_{32} \}$. 
If $v_2$ in $S_{34}$, then we find $M_{III}(4)$ induced by $\{ v_1$, $v_2$, $s_{51}$, $s_{13}$, $k_1$, $k_4$, $k_5$, $k_{31}$, $k_{32} \}$. 

\subcase \textit{Suppose $S_2(j)$ is a subconfiguration of $\mathbb A_i$ for some $j \geq 3$.} Let $v_1, v_2, \ldots, v_j$ be the vertices in $S$ represented by the rows of $S_2(j)$ and $k_{i1}$, $k_{i2}, \ldots, k_{i(j-1)}$ be the vertices in $K_i$ that represent columns $1$ to $j-1$ of $S_2(j)$. Notice that $v_1$ and $v_j$ are labeled with the same letter, and depending on whether $j$ is odd or even, then $v_1$ and $v_j$ are colored with distinct colors or with the same color, respectively. We assume without loss of generality that $v_1$ and $v_j$ are both labeled with L.

\subsubcase \textit{Suppose $j$ is odd. }If $i =3$, then there are no vertices $v_1$ and $v_j$ labeled with the same letter and colored with distinct colors as in $S_2(j)$.
Thus, let $i= 4$. In this case, $v_1$ in $S_{34}$ and $v_j$ in $S_{14} \cup S_{64}$.
Let $k_3$ in $K_3$ be a vertex adjacent to both $v_1$ and $v_j$, and let $k_1$ in $K_1$ adjacent to $v_j$. Thus, we find $F_1(j+2)$ induced by $\{ s_{13}$, $s_{35}$, $v_1$, $\ldots$, $v_j$, $k_1$, $k_3$, $k_{i1}$, $\ldots$, $k_{i(j-1)} \}$.

\subsubcase  \textit{Suppose $j$ is even. }We split this in two cases, depending on the parity of $i$.
If $i=3$, then $v_1$ and $v_j$ lie in $S_{13} \cup S_{23}$.
Suppose that $v_1$ in $S_{13}$ and $v_j$ in $S_{23}$. Let $k_2$ in $K_2$ adjacent to $v_1$ and $v_j$. Hence, we find $F_1(j+2)$ induced by the subset $\{ v_1$, $\ldots$, $v_j$, $k_2$, $k_{i2}$, $\ldots$, $k_{i(j-1)}$, $s_{35} \}$. The same holds if both $v_1$ and $v_j$ lie in $S_{23}$.
If instead $v_1$ and $v_j$ both lie in $S_{13}$, then we find $F_1(j+2)$ induced by the same subset but replacing $k_2$ for a vertex $k_1$ in $K_1$ adjacent to both $v_1$ and $v_j$.

Suppose now that $i=4$. In this case, $v_1$ and $v_j$ lie in $S_{14} \cup S_{64}$. In either case, there is a vertex $k_1$ in $K_1$ that is adjacent to both $v_1$ and $v_j$. We find $F_1(j+1)$ induced by $\{ v_1$, $\ldots$, $v_j$, $k_1$, $k_{i1}$, $\ldots$, $k_{i(j-1)}$, $ s_{35} \}$.

\subcase \textit{Suppose $S_3(j)$ is a subconfiguration of $\mathbb A_i$ for some $j \geq 3$. }
Let $v_1, v_2, \ldots, v_j$ be the vertices represented by the rows of $S_3(j)$ and $k_{i1}, \ldots, k_{i(j-1)}$ be the vertices represented by columns $1$ to $j-1$ of $S_3(j)$. Notice that $v_1$ and $v_j$ are labeled with distinct letters, and depending on whether $j$ is odd or even, $v_1$ and $v_j$ are either colored with distinct colors or with the same color, respectively. We assume without loss of generality that $v_1$ is labeled with L and $v_j$ is labeled with R.

\subsubcase \textit{Suppose $j$ is odd. }If $i=3$, then $v_1$ lies in $S_{34} \cup S_{35}$, and $v_j$ lies in $S_{13} \cup S_{23}$.
If $v_1$ lies in $S_{34}$ and $v_j$ lies in $S_{23}$, then we find $F_1(j+2)$ induced by $\{ v_1$, $\ldots$, $v_j$, $k_2$, $k_4$, $k_{i1}$, $\ldots$, $k_{i(j-1)}$, $s_{35}$, $s_{13} \}$. 
If $v_1$ lies in $S_{34}$ and $v_j$ lies in $S_{13}$, then we find $F_1(j+2)$ induced by $\{ v_1$, $\ldots$, $v_j$, $k_1$, $k_4$, $k_{i1}$, $\ldots$, $k_{i(j-1)}$, $s_{35}$, $s_{13} \}$.  
If instead $v_1$ lies in $S_{35}$ and $v_j$ lies in $S_{23}$, then we find $F_1(j+2)$ induced by $\{ v_1$, $\ldots$, $v_j$, $k_2$, $k_5$, $k_{i1}$, $\ldots$, $k_{i(j-1)}$, $s_{35}$, $s_{13} \}$. 

If instead $i=4$, then $v_1$ in $S_{34}$ and $v_j$ in $S_{41} \cup S_{42}$. In either case, we find a $(j+1)$-sun induced by $\{ v_1$, $\ldots$, $v_j$, $k_{i1}$, $\ldots$, $k_{i(j-1)}$, $k_1$, $k_3$, $s_{13} \}$. 

\subsubcase \textit{Suppose $j$ is even.} If $i = 3$, then there no two rows in $\mathbb A_3$ labeled with distinct letters and colored with the same color. Hence, assume $i=4$ and thus either $v_1$ in $S_{34}$ and $v_j$ in $S_{45}$, or $v_1$ in $S_{14} \cup S_{64}$ and $v_j$ in $S_{41} \cup S_{42}$.

If $v_1$ in $S_{34}$ and $v_j$ in $S_{45}$, then we find a $(j+2)$-sun induced by $\{ v_1$, $\ldots$, $v_j$, $k_1$, $k_3$, $k_5$, $k_{i1}$, $\ldots$, $k_{i(j-1)}$, $s_{13}$, $s_{51} \}$, where $k_1$ in $K_1$ is nonadjacent to both $v_1$ and $v_j$. 

If instead $v_1$ in $S_{14} \cup S_{64}$ and $v_j$ in $S_{41} \cup S_{42}$, then we find a $j$-sun induced by $\{ v_1$, $\ldots$, $v_j$, $k_1$, $k_{i1}$, $\ldots$, $k_{i(j-1)} \}$. 

Therefore, $\mathbb A_i$ is admissible for every $j \in \{1, 2, \ldots, 6 \}$.
\vspace{1mm}


\case \textit{$\mathbb A_i$ is admissible but not LR-orderable. }
Then $\mathbb A_i$ contains a Tucker matrix, or one of the following submatrices: $M_4'$, $M_4''$, $M_5'$, $M_5''$, $M'_2(k)$, $M''_2(k)$, $M_3'(k)$, $M_3''(k)$ or their corresponding dual matrices, for some $k \geq 4$.
We assume throughout the rest of the proof that, for each pair of vertices $x$ and $y$ in the $S_{ij}$ of $S$, there are vertices $k_i$ in $K_i$ and $k_j$ in $K_j$ such that both $x$ and $y$ are adjacent to $k_i$ and $k_j$. This follows from Claim~\ref{claim:tent_0} and the fact that $\mathbb A_i$ is admissible.

Suppose there is $M_I(j)$ as a submatrix of $\mathbb A_i$. Let $v_1, \ldots, v_j$ be the vertices of $S$ represented by rows $1$ to $j$ of $M_I(k)$, and let $k_{i1}, \ldots, k_{ij}$ be the vertices in $K$ represented by columns $1$ to $j$. Thus, if $j$ is even, then we find either a $j$-sun induced by $\{v_1$, $\ldots$, $v_j$, $k_{i1}$, $\ldots$, $k_{ij} \}$, and if $j$ is odd, then we find a $j$-sun with center induced by the subset $\{v_1$, $\ldots$, $v_j$, $k_{i1}$, $\ldots$, $k_{ij}$, $s_{i(i+2)} \}$.

For any other Tucker matrix, we find the homonym forbidden induced sub\-graph induced by the subset $\{v_1$, $\ldots$, $v_j$, $k_{i1}$, $\ldots$, $k_{ij} \}$.

Suppose that $\mathbb A_i$ contains either $M_4'$, $M_4''$, $M_5'$, $M_5''$, $M'_2(k)$, $M''_2(k)$, $M_3'(k)$, $M_3''(k)$ or one of their corresponding dual matrices, for some $k \geq 4$. Let $M$ be such a submatrix. Notice that, for any tag column $c$ of $M$ that denotes which vertices are labeled with L, there is a vertex $k'$ in either $K_{i-1}$ or $K_{i-2}$ such that the vertices represented by a labeled row in $c$ are adjacent to $k'$  in $G$.
    If instead the tag column $c$ denotes which vertices are labeled with R, then we find an analogous vertex $k''$ in either $K_{i+1}$ or $K_{i+2}$.

Depending on whether there is one or two tag columns in $M$, we find the homonym forbidden induced subgraph induced by the vertices in $S$ and $K$ represented by the rows and the non-tagged columns of $M$ plus one or two vertices $k'$ and $k''$ as those described in the previous paragraph.

\case \textit{ $\mathbb A_i$ is LR-orderable but contains some gem. }Moreover, since $\mathbb A_i$ contains no LR-rows, then $\mathbb A_i$ contains either a monochromatic gem or a monochromatic weak gem.

Let $v_1$ and $v_2$ in $S$ be the vertices represented by the rows of the monochromatic gem. Notice that both rows are labeled rows, since every unlabeled row in $\mathbb A_i$ is uncolored. It follows that a monochromatic gem or a monochromatic weak gem may be induced only by two rows labeled with L or R, and hence both are the same case.

\subcase If $i=3$, since both vertices are colored with the same color, then $v_1$ in $S_{34}$ and $v_2$ in $S_{35}$.
In that case, we find $D_0$ in $\mathbb A_i$ since both rows are labeled with the same letter, which results in a contradiction for we assumed that $\mathbb A_i$ is admissible. The same holds if both vertices belong to either $S_{34}$ or $S_{35}$.

\subcase If instead $i= 4$, then we have three possibilities. Either $v_1$ in $S_{14}$ and $v_2$ in $S_{64}$, or $v_1$ in $S_{34}$ and $v_2$ in $S_{45}$, or $v_1$ in $S_{14}$ and $v_2$ in $S_{41}$.
The first case is analogous to the $i=3$ case stated above.
For the second and third case, since both rows are labeled with distinct letters, then we find $D_1$ as a submatrix of $\mathbb A_i$. This results once more in a contradiction, for $\mathbb A_i$ is admissible.

Therefore, $\mathbb A_i$ contains no monochromatic gems, monochromatic weak gems or badly-colored doubly-weak gems.

\case \textit{ $\mathbb A_i$ contains none of the matrices listed in Theorem~\ref{teo:2-nested_caract_bymatrices}, but $\mathbb A_i$ is not $2$-nested.}  
It follows from Lemma~\ref{lema:2-nested_if} that, for every suitable LR-ordering and $2$-color assignment of all the blocks of $\mathbb A_i$ that extends the given pre-coloring, we find either a monochromatic gem or a monochromatic weak gem. Moreover, at least one of the rows of such a gem is unlabeled, for there are no LR-rows in $\mathbb A_i$.
Consider the columns of the matrix $\mathbb A_i$ ordered according to a suitable LR-ordering. 
Suppose there is a monochromatic gem given by the rows $v_j$ and $v_{j+1}$, and suppose that both rows are colored with red.
If is not possible to color these two rows with distinct colors, then there is at least one more row $v_{j-1}$ colored with blue and forces $v_j$ to be colored with red. If $v_{j-1}$ is unlabeled, then $v_{j-1}$ and $v_j$ overlap. If $v_{j-1}$ is labeled with L or R, then $v_j$ and $v_{j-1}$ induce a weak gem.

If $v_{j-1}$ forces the coloring only on $v_j$, let $v_{j+2}$ be a distinct row that forces $v_{j+1}$ to be colored with red. Suppose first that $v_{j+2}$ forces the coloring only to the row $v_{j+1}$. Hence, there is a submatrix as the following in $\mathbb A_i$:

\vspace{-5mm}
\[ \small{ \bordermatrix{ & \cr
        v_{j-1}\  & 1 1 0 0 0 \cr
        v_{j}\  & 0 1 1 0 0  \cr
        v_{j+1}\  & 0 0 1 1 0  \cr
        v_{j+2}\  &  0 0 0 1 1  }\ }
    \begin{matrix}
    \textcolor{blue}{\bullet} \\ \textcolor{red}{\bullet} \\ \textcolor{red}{\bullet} \\ \textcolor{blue}{\bullet}
    \end{matrix}
        \]

If there are no other rows forcing the coloring of $v_{j-1}$ and $v_{j+2}$, then this submatrix can be colored blue-red-blue-red. Since this is not possible, there is a sequence of overlapping rows $v_l, \ldots, v_{j-2}$ and $v_{j+3}, \ldots, v_k$ such that each row forces the coloring of the next one, and this sequence includes$v_{j-1}$, $v_j$, $v_{j+1}$ and $v_{j+2}$. Moreover, suppose that this is the longest sequence of vertices with this property. Hence, $v_l$ and $v_k$ are labeled rows, for if not we could color again the rows and thus extending the pre-coloring, which results in a contradiction.
However, we find either $S_2(k-l+1)$ or $S_3(k-l+1)$ in $\mathbb A_i$, and this also results in a contradiction, for $\mathbb A_i$ is admissible.

Suppose now that $v_{j-1}$ forces the red color on both $v_j$ and $v_{j+1}$. If $v_{j-1}$ is unlabeled, then $v_{j-1}$ overlaps with both $v_j$ and $v_{j+1}$. Since $v_j$ and $v_{j+1}$ overlap, either $v_j[r_j] = v_{j+1}[r_j] = 1$ or $v_j[l_j] = v_{j+1}[l_j] = 1$. Suppose without loss of generality that $v_j[r_j] = v_{j+1}[r_j] = 1$.
Since $v_{j-1}$ overlaps with $v_j$, then either $v_{j-1}[l_j] = 1$ or $v_{j-1}[r_j] = 1$, and the same holds for $v_{j-1}[l_{j+1}] = 1$ or $v_{j-1}[r_{j+1}] = 1$.
If $v_{j-1}[l_j] = 1$, then $v_{j-1}[l_{j+1}] = 1$ and $v_j[l_{j+1}] = 1$, and thus we find $F_0$ induced by $\{ v_{j-1}$, $v_j$, $v_{j+1}$, $l_{j-1}$, $l_{j+1}-1$, $l_{j+1}$, $r_j$, $r_j + 1 \}$, which is a contradiction.
Analogously, if $v_{j-1}[r_j] = 1$, then $v_{j-1}[l_{j+1}] = 1$ and $v_{j-1}[l_j] = 1$, and thus we find $F_0$ induced by $\{ v_{j-1}$, $v_j$, $v_{j+1}$, $l_j$, $l_{j+1}$, $r_j$, $r_j +1$, $r_{j-1} \}$.

It follows analogously if $v_{j-1}$ is labeled with L or R, except that we find $F'_0$ instead of $F_0$ as a subconfiguration in $\mathbb A_i$. Moreover, the proof is analogous if $v_j$ and $v_{j+1}$ induce a weak-gem instead of a gem. 
\end{mycases}
Therefore, we reached a contradiction in every case and thus $\mathbb A_i$ is $2$-nested.
\end{proof}

Let $G= (K,S)$ and $H$ as in Section~\ref{subsec:tent_partition}, and the matrices $\mathbb A_i$ for each $i=1, 2, \ldots, 6$ as in the previous subsection.
Suppose $\mathbb A_i$ is $2$-nested for each $i =1, 2, \ldots, 6$. Hence, there is a suitable LR-ordering $\Pi_i$ for each $i =1, 2, \ldots, 6$ and a $2$-coloring extension $\chi_i$ of the given block bi-coloring. Since $\mathbb A_i$ contains no LR-rows, then every row in each matrix $\mathbb A_i$ is colored with either red or blue. 

Let $\Pi$ be the ordering of the vertices of $K$ given by concatenating the LR-orderings $\Pi_1$, $\Pi_2$, $\ldots$, $\Pi_6$. Let $A=A(S,K)$ and consider the columns of $A$ ordered according to $\Pi$. Given $s$ in $S_{ij}$, we denote throughout the following by $s_i$ the row corresponding to $s$ in $\mathbb A_i$.
For each vertex $s$ in $S_{ij}$, if $i \leq j$, then $s_i$ in $\mathbb A_i$ and $s_j$ in $\mathbb A_j$ are colored with the same color. Thus, we consider the row corresponding to $s$ in $A$ colored with that color. Notice that, if $i < l < j$, then $s$ is complete to each $K_l$. 
If instead $i >j$, then $s_i$ in $\mathbb A_i$ and $s_j$ in $\mathbb A_j$ are colored with distinct colors. Moreover, the row corresponding to $s$ in $A$ has an L-block and an R-block.
Thus, we consider its L-block colored with the same color assigned to $s_i$ and the R-block colored with the same color assigned to $s_j$.
Notice that the assignment of distinct colors in $\mathbb A_i$ and $\mathbb A_j$ makes sense, since we are describing vertices whose chords must have one of its endpoints drawn in the $K_i^+$ portion of the circle and the other endpoint in the $K_j^-$ portion of the circle. 

Let $s \in S$. Hence, $s$ lies in $S_{ij}$ for some $i,j\in \{1,2,\ldots,6\}$. Notice that, a row representing a vertex $s$ in $S_{ii}$ is entirely colored with the same color. 

\begin{defn} \label{def:matrices_A_por_colores}
We define the $(0,1)$-matrix $\mathbb A_r$ as the matrix obtained by considering only those rows representing vertices in $S \setminus \bigcup_{i=1}^6 S_{ii}$ and adding two distinct columns $c_L$ and $c_R$ such that the entry $\mathbb A_r (s, k)$ is defined as follows:
\begin{itemize}
    \item If $i<j$ and $s_i$ is colored with red, then the entry $\mathbb A_r (s, k)$ has a $1$ if $s$ is adjacent to $k$ and a $0$ otherwise, for every $k$ in $K$, and $\mathbb A_r (s, c_R)$ = $\mathbb A_r (s, c_L) = 0$.
    \item If $i>j$ and $s_i$ is colored with red, then the entry $\mathbb A_r (s, k)$ has a $1$ if $s$ is adjacent to $k$ and a $0$ otherwise, for every $k$ in $K_i \cup \ldots K_6$, and $\mathbb A_r (s, c_R) = 1$, $\mathbb A_r (s, c_L) = 0$. If instead $s_j$ is colored with red, then the entry $\mathbb A_r (s, k)$ has a $1$ if $s$ is adjacent to $k$ and a $0$ otherwise, for every $k$ in $K_1 \cup \ldots K_j$, and $\mathbb A_r (s, c_R) = 0$, $\mathbb A_r (s, c_L) = 1$.
\end{itemize}

The matrix $\mathbb A_b$ is defined in an entirely analogous way, changing red for blue in the definition.


\vspace{2mm}
We define the $(0,1)$-matrix $\mathbb A_{r-b}$ as the submatrix of $A$ obtained by considering only those rows corresponding to vertices $s$ in $S_{ij}$ with $i>j$ for which $s_i$ is colored with red.
The matrix $\mathbb A_{b-r}$ is defined as the submatrix of $A$ obtained by considering those rows corresponding to vertices $s$ in $S_{ij}$ with $i>j$ for which $s_i$ is colored with blue.

\end{defn}


\begin{lema} \label{lema:matrices_union_son_nested}
    Suppose that $\mathbb A_i$ is $2$-nested for every $=1, 2, \ldots, 6$. If $\mathbb A_r$, $\mathbb A_b$, $\mathbb A_{r-b}$ or $\mathbb A_{b-r}$ are not nested, then $G$ contains either tent${}\vee{}K_1$ or $F_0$ as induced subgraphs. 
\end{lema}

\begin{proof}
    Suppose first that $\mathbb A_r$ is not nested. Then, there are two rows $f_1$ and $f_2$ in $\mathbb A_r$ that induce a $0$-gem. Let $v_1$ in $S_{ij}$ and $v_2$ in $S_{lm}$ be the vertices corresponding to such rows in $G$.
    Since $\mathbb A_i$ is $2$-nested for every $=1, 2, \ldots, 6$, in particular there are no monochromatic gems in each $\mathbb A_i$. 
    First we need the following claim.
    \begin{claim}
        If $v_1$ and $v_2$ lie in the same $S_{ij}$, then $v_1$ and $v_2$ are nested.
    \end{claim}
    Suppose by simplicity that $i <j$. Towards a contradiction, suppose that $v_1$ and $v_2$ are not nested. Thus, there are vertices $k_{i1}$ and $k_{i2}$ in $K_i$, and $k_{j1}$ and $k_{j2}$ in $K_j$ such that $v_1$ is adjacent to $k_{i1}$, $k_{i2}$ and $k_{j1}$ and nonadjacent to $k_{j2}$, and $v_2$ is adjacent to $k_{i2}$, $k_{j1}$ and $k_{j2}$ and nonadjacent to $k_{i1}$. Let $l \in \{1,3, 5\}$ such that $v_i$ is nonadjacent to $K_l$. The existence of such index follows from Claim~\ref{claim:tent_1}, since there is no vertex of $S$ simultaneously adjacent to $K_1$, $K_3$ and $K_5$.
    We find $F_0$ induced by $\{ v_1$, $v_2$, $s*$, $k_{i1}$, $k_{i2}$, $k_{j1}$, $k_{j2}$, $k_{l} \}$, where $s*= s_{jl}$ or $s*= s_{(j+1)l}$ depending on the parity of $j$. \QED

    We assume from now on that $v_1$ and $v_2$ lie in distinct subsets of $S$. The rows in $\mathbb A_r$ represent vertices in the following subsets of $S$: $S_{34}$, $S_{45}$, $S_{35}$, $S_{36}$, $S_{25}$, $S_{26}$, $S_{42}$, $S_{52}$, $S_{51}$, $S_{61}$, $S_{64}$ or $S_{63}$. Notice that $S_{36} = S_{[36}$, $S_{25} = S_{25]}$.
    \begin{mycases}
    \case $v_1$ in $S_{34} \cup S_{45}$. Suppose that $v_1$ in $S_{34}$, thus $v_2$ in $S_{35}$ since $\mathbb A_4$ is admissible. We find $F_0$ induced by $\{ v_1$, $v_2$, $s_{13}$, $k_{1}$, $k_{31}$, $k_{32}$, $k_{4}$, $k_{5} \}$.
    It follows analogously if $v_1$ in $S_{45}$, for the only possibility is $v_2$ in $S_{35}$ since $S_{25}$ is complete to $K_5$.

    \case $v_1$ in $S_{35} \cup S_{36}$. Since $S_{36}$ is complete to $K_3$, $S_{25}$ is complete to $K_5$ and $\mathbb A_6$ is admissible, then necessarily $v_1$ in $S_{36}$. In that case, either $v_2$ in $S_{25} \cup S_{26}$ or $v_2$ in $S_{63} \cup S_{64}$. We find $F_0$ induced by $\{ v_1$, $v_2$, $s_{13}$, $k_{1}$, $k_{2}$, $k_{3}$, $k_{5}$, $k_{6} \}$ if $v_2$ in $S_{25}$, or $\{ v_1$, $v_2$, $s_{13}$, $k_{1}$, $k_{2}$, $k_{3}$, $k_{61}$, $k_{62} \}$ if $v_2$ in $S_{26}$. We find tent${}\vee{}K_1$ induced by $\{ v_1$, $v_2$, $s_{51}$, $k_{1}$, $k_{3}$, $k_{5}$, $k_{6} \}$ if $v_2$ in $S_{63} \cup S_{64}$.

    \case $v_1$ in $S_{25} \cup S_{26}$. Since $\mathbb A_2$ and $\mathbb A_6$ are admissible, then the only possibility is $v_1$ in $S_{25}$ and $v_2$ in $S_{26}$. We find $F_0$ induced by $\{ v_1$, $v_2$, $s_{13}$, $k_{1}$, $k_{21}$, $k_{22}$, $k_{5}$, $k_{6} \}$ and therefore $\mathbb A_r$ is nested.
    \end{mycases}
    Let us suppose that $\mathbb A_b$ is not nested.
    \begin{mycases}
    The rows in $\mathbb A_b$ represent vertices in the following subsets of $S$: $S_{12}$, $S_{13}$, $S_{23}$, $S_{14}$, $S_{41}$, $S_{42}$, $S_{52}$, $S_{51}$, $S_{56}$, $S_{61}$, $S_{64}$ or $S_{63}$.
    Notice that $S_{14} = S_{[14}$, $S_{52} = S_{[52}$ and $S_{41} = S_{41]}$.

    \case $v_1$ in $S_{12} \cup S_{23}$. Thus, $v_2$ in $S_{13}$ since $\mathbb A_2$ is admissible. If $v_1$ in $S_{12}$, then we find $F_0$ induced by $\{ v_1$, $v_2$, $s_{51}$, $k_{5}$, $k_{11}$, $k_{12}$, $k_{2}$, $k_{3} \}$. It follows analogously by symmetry if $v_1$ in $S_{23}$. 

    \case $v_1$ in $S_{13}$. Since $S_{14}$ is complete to $K_1$, the only possibility is $v_2$ in $S_{63}$. We find $F_0$ induced by $\{ v_1$, $v_2$, $s_{35}$, $k_{6}$, $k_{2}$, $k_{31}$, $k_{32}$, $k_{5} \}$.

    \case $v_1$ in $S_{14}$. Since $\mathbb A_4$ is admissible, then $v_2$ in $S_{63} \cup S_{64}$.
    We find $F_0$ induced by $\{ v_1$, $v_2$, $s_{35}$, $k_{6}$, $k_{1}$, $k_{3}$, $k_{4}$, $k_{5} \}$ if $v_2$ in $S_{63}$ and by $\{ v_1$, $v_2$, $s_{35}$, $k_{6}$, $k_{1}$, $k_{3}$, $k_{41}$, $k_{42} \}$ if $v_2$ in $S_{64}$.
    \end{mycases}

    Suppose now that $\mathbb A_{b-r}$ is not nested. The rows in $\mathbb A_{b-r}$ represent vertices in the following subsets of $S$: $S_{41}$, $S_{42}$, $S_{51}$, $S_{52}$ or $S_{61}$. Notice that $S_{41}=S_{41]}$ and $S_{52}=S_{[52}$.
    If $v_1$ in $S_{41}$ and $v_2$ in $S_{42}$, then we find $F_0$ induced by $\{ v_1$, $v_2$, $s_{13}$, $k_{41}$, $k_{42}$, $k_{1}$, $k_{2}$, $k_{3} \}$. The proof is analogous if the vertices lie in $S_{51} \cup S_{52}$.
    If instead $v_1$ in $S_{61}$, then $v_2$ in $S_{51}$. We find $F_0$ induced by $\{ v_1$, $v_2$, $s_{13}$, $k_{11}$, $k_{12}$, $k_{3}$, $k_{5}$, $k_{6} \}$ and therefore $\mathbb A_{b-r}$ is nested.

     Suppose that $\mathbb A_{r-b}$ is not nested. The rows in $\mathbb A_{r-b}$ represent vertices in $S_{63}$ or $S_{64}$. If $v_1$ in $S_{63}$ and $v_2$ in $S_{64}$, then we find $F_0$ induced by $\{ v_1$, $v_2$, $s_{51}$, $k_{5}$, $k_{61}$, $k_{62}$, $k_{3}$, $k_{4} \}$. 
    This finishes the proof and therefore $\mathbb A_{r}$, $\mathbb A_{b}$, $\mathbb A_{b-r}$ and $\mathbb A_{r-b}$ are nested.
\end{proof}

\begin{teo} \label{teo:finalteo_tent}
    Let $G=(K,S)$ be a split graph containing an induced tent. Then, the following are equivalent:
\begin{enumerate}
\item $G$ is circle;
\item $G$ is \fsc-free;
\item $\mathbb A_1,\mathbb A_2,\ldots,\mathbb A_6$ are $2$-nested and $\mathbb A_r$, $\mathbb A_b$, $\mathbb A_{b-r}$ and $\mathbb A_{r-b}$ are nested.
\end{enumerate}
\end{teo}

\begin{proof} It is not hard to see that $(1) \Rightarrow (2)$, and that
$(2) \Rightarrow (3)$ is a consequence of the previous lemmas. We will show $(3) \Rightarrow (1)$. Suppose that each of the matrices $\mathbb A_1,\mathbb A_2,\ldots,\mathbb A_6$ is $2$-nested, and that the matrices $\mathbb A_r$, $\mathbb A_b$, $\mathbb A_{b-r}$ and $\mathbb A_{r-b}$ are nested.
Let $\Pi_i$ be a suitable LR-ordering for the columns of $\mathbb A_i$ for each $i=1,2,\ldots,6$, and let $\Pi$ be the ordering obtained by concatenation of $\Pi_i$ for all the vertices in $K$.
Consider the circle divided into twelve pieces as in Figure~\ref{fig:modelA3}. For each vertex in $K_i$, we place a chord having one endpoint in $K_i^+$ and the other endpoint in $K_i^-$, considering the endpoints of the chords in both $K_i^+$ and $K_i^-$ ordered according to $\Pi_i$.
We denote by $a_i^-$ and $a_i^+$ the position on the circle of the endpoints of the chords corresponding to the first and last vertex of $K_i$ --ordered according to $\Pi_i$-- respectively. 
We denote by $s_{i,i+2}^+$ the placement of the chord corresponding to the vertex $s_{i,i+2}$ of the tent $H$, which lies between $a_{i-1}^+$ and $a_i^-$, and $s_{i,i+2}^+$ to the placement of the chord of the vertex $s_{i, i+2}$ that lies between $a_{i+1}^+$ and $a_{i+2}^-$.

Let us see how to place the chords for the vertices in every subset $S_{ij}$ of $S$.
\begin{claim} \label{claim:tent_placement}
The following assertions hold:
\begin{enumerate}[noitemsep,nolistsep]
    \item If $i \neq j$, then all the vertices in each $S_{ij}$ are nested.
    \item If $k \leq i$ and $j\leq l$, then every vertex in $S_{ij}$ is contained in every vertex of $S_{kl}$.
    \item For each $i \in \{2, 4, 6\}$, the vertex set $N_{K_i}(S_{(i-1)i}) \cap N_{K_i}(S_{i(i+1)})$ is empty.
    Moreover, $N_{K_i}(S_{ij}) \cap N_{K_i}(S_{(i+3)i}) = \emptyset$ and $N_{K_i}(S_{ij}) \cap N_{K_i}(S_{(i+2)i}) = \emptyset$, for $j=i+3, i+4$.
    \item For each $i\in\{1,3,5\}$, if $S_{i(i+3)} \neq \emptyset$, then $S_{(i-1)(i+2)} \neq \emptyset$, and viceversa.
\end{enumerate}
\end{claim}

Every statement follows directly from Claim~\ref{claim:tent_0} and the proof of Lemma~\ref{lema:matrices_union_son_nested}. \QED

Let us consider the subsets $S_{ij}$ such that $i \neq j$ and $i=1$ or $j=1$. These subsets are $S_{12}$, $S_{13}$, $S_{14}$, $S_{61}$, $S_{51}$ and $S_{j4}$. It follows from Claim~\ref{claim:tent_placement} that every vertex in $S_{12}$ is contained in every vertex of $S_{13}$, which is in turn contained in $S_{14}$, and the same holds for $S_{61}$, $S_{51}$ and $S_{41}$.
Furthermore, it follows that $S_{14} \neq \emptyset$ only if $S_{63} = \emptyset$ and that every vertex in $S_{14}$ is contained in every vertex in $S_{64}$.
Now consider those subsets $S_{ij}$ for which $i \neq j$ and either $i=2$ or $j=2$, namely $S_{12}$, $S_{23}$, $S_{25}$, $S_{26}$, $S_{42}$ and $S_{52}$. It follows from the previous claim that a vertex in $S_{12}$ is disjoint with any vertex in $S_{23}$. Moreover, every vertex in $S_{42}$ is nested in every vertex in $S_{52}$ and every vertex in $S_{25}$ is nested in every vertex in $S_{26}$.
The same analysis follows by symmetry for those subsets such that $i \neq j$ and either $i \in \{3, 4, 5, 6\}$ or $j \in \{3, 4, 5, 6\}$.
Furthermore, since $\mathbb A_i$ is nested for every $i \in \{1, \ldots, 6\}$, then every vertex in $S_{ii}$ is either contained or disjoint with every vertex in $S_{jk}$, for every $j,k \in \{1, \ldots, 6\}$ such that $j \leq i \leq k$.
Therefore, Claim~\ref{claim:tent_placement} ensures that we can place the chords corresponding to vertices of $S$ in such a way that they do not intersect. We consider the vertices in each $S_{ij}$ ordered by inclusion. The guidelines for each subset $S_{ij}$ are as follows.
\begin{itemize}
    \item For those vertices in $S_{i(i+1)}$: if $i =1, 2, 5$, then we place one endpoint in $K_i^-$ and the other endpoint in $K_{i+1}^-$. If $i =3, 4$, then we place one endpoint in $K_i^+$ and the other endpoint in $K_{i+1}^+$. If $i = 6$, then we place one endpoint in $K_6^-$ and the other endpoint in $K_1^+$.
    \item For the vertices in $S_{(i-1)(i+1)}$: if $i = 2$, then we place one endpoint in $K_1^-$ and the other endpoint in $K_3^-$. If $i =4$, then we place one endpoint in $K_3^+$ and the other endpoint in $K_5^+$. If $i =6$, then we place one endpoint in $K_5^-$ and the other endpoint in $K_1^+$.
    \item For $S_{(i-1)(i+2)}$: if $i = 1$, we place one endpoint in $K_6^+$, and the other endpoint between $s_{13}^-$ and the chord corresponding to $a_4^-$ in $K_4^-$. If $i = 2$, we place one endpoint between the chord corresponding to $a_6^+$ in $K_6^+$ and $s_{13}^+$, and the other endpoint in $K_4^-$. If $i = 3$, we place one endpoint in $K_2^+$, and the other endpoint between $s_{35}^-$ and the chord corresponding to $a_6^-$ in $K_6^+$. If $i = 4$, we place one endpoint between the chord corresponding to $a_2^+$ in $K_2^+$ and $s_{35}^+$, and the other endpoint in $K_6^+$. If $i = 5$, we place one endpoint in $K_4^-$, and the other endpoint between $s_{51}^-$ and the chord corresponding to $a_2^-$ in $K_2^+$. If $i = 6$, we place one endpoint between the chord corresponding to $a_4^+$ in $K_4^-$ and $s_{51}^+$, and the other endpoint in $K_2^+$.
    \item For $S_{i-2,i+2}$: if $i=2$, we place one endpoint in $K_6^+$ and the other endpoint in $K_4^-$. If $i=4$, we place one endpoint in $K_2^+$ and the other endpoint in $K_6^+$. If $i=6$, we place one endpoint in $K_4^-$ and the other endpoint in $K_2^+$.
\end{itemize}
This gives a circle model for the given split graph $G$.
\end{proof}

\subsection{Split circle graphs containing an induced 4-tent} \label{subsec:circle3}

In this section we address the second case of the proof of Theorem~\ref{teo:circle_split_caract}, which is when $G$ contains an induced $4$-tent and contains no induced tent. The main difference with the previous case is that one of the enriched matrices that we will define may contain LR-rows.
This section is subdivided as follows. In Section~\ref{subsubsec:4tent2}, we define the matrices $\mathbb{B}_i$ for each $i=1, 2, \ldots, 6$. 
In Section~\ref{subsubsec:4tent3}, we prove the necessity of the $2$-nestedness of each $\mathbb B_i$ for $G$ to be a \fsc-graph and give the guidelines to draw a circle model for a \fsc-free split graph $G$ containing an induced $4$-tent in Theorem~\ref{teo:finalteo_4tent}.

\subsubsection{Matrices $\mathbb B_1,\mathbb B_2,\ldots,\mathbb B_6$} \label{subsubsec:4tent2}

Let $G=(K,S)$ and $H$ as in Section~\ref{subsec:4tent_partition}.
For each $i\in\{1,2,\ldots,6\}$, let $\mathbb B_i$ be an enriched $(0,1)$-matrix having one row for each vertex $s\in S$ such that $s$ belongs to $S_{ij}$ or $S_{ji}$ for some $j\in\{1,2,\ldots,6\}$ and one column for each vertex $k\in K_i$ and such that such that the entry corresponding to row $s$ and column $k$ is $1$ if and only if $s$ is adjacent to $k$ in $G$. For each $j\in\{1,2,\ldots,6\}-\{i\}$, we label those rows corresponding to vertices of $S_{ji}$ with L and those corresponding to vertices of $S_{ij}$ with R, with the exception of $S_{[15]}$ and $S_{[16}$, whose vertices are labeled with LR.
As in the previous section, some of the rows of $\mathbb B_i$ are colored. Given the symmetry of the partitions of $K$, it suffices to study  $\mathbb B_i$ for $i=1,2,3,6$ (See Figure~\ref{fig:matricesB}). 

\begin{figure}[h]
\centering
\footnotesize{
\[ \mathbb B_1 = \tiny{ \bordermatrix{ & K_1\cr
        S_{12}\ \textbf{R} &\cdots \cr
                 S_{11}\            & \cdots \cr
                S_{14}\ \textbf{R} & \cdots \cr
                S_{15}\ \textbf{R} & \cdots \cr
                S_{16}\ \textbf{R} & \cdots \cr
                S_{61}\ \textbf{L} & \cdots }\
                \begin{matrix}
                \textcolor{red}{\bullet} \\ \\ \textcolor{blue}{\bullet} \\ \textcolor{blue}{\bullet} \\ \textcolor{blue}{\bullet} \\ \textcolor{blue}{\bullet}
                \end{matrix} } \qquad
   \mathbb B_2 = \tiny{ \bordermatrix{ & K_2\cr
        S_{12}\ \textbf{L} & \cdots \cr
                S_{22}\            & \cdots \cr
                S_{23}\ \textbf{R} & \cdots \cr
                S_{24}\ \textbf{R} &  \cdots }\
                \begin{matrix}
                \textcolor{red}{\bullet} \\ \\ \textcolor{blue}{\bullet} \\ \textcolor{blue}{\bullet}
                \end{matrix} } \qquad
   \mathbb B_3 =  \tiny{ \bordermatrix{ & K_3\cr
                S_{35}\ \textbf{R} & \cdots \cr
                S_{36}\ \textbf{R} & \cdots \cr
                S_{13}\ \textbf{L} & \cdots \cr
                S_{33}\            & \cdots \cr
                S_{34}\ \textbf{R} & \cdots \cr
                S_{23}\ \textbf{L} & \cdots }\
                \begin{matrix}
                 \textcolor{red}{\bullet} \\ \textcolor{red}{\bullet} \\ \textcolor{red}{\bullet} \\ \\ \textcolor{blue}{\bullet} \\ \textcolor{blue}{\bullet}
                \end{matrix} } \qquad
              \mathbb B_6 = \tiny{ \bordermatrix{ & K_6\cr
        S_{61}\ \textbf{R} & \cdots \cr
                S_{64}\ \textbf{R} & \cdots \cr
                S_{65}\ \textbf{R} & \cdots \cr
               S_{36}\ \textbf{L} & \cdots \cr
                S_{46}\ \textbf{L} & \cdots \cr
                S_{66}\            & \cdots \cr
                S_{62}\ \textbf{R} & \cdots \cr
                S_{63}\ \textbf{R} & \cdots \cr
                S_{26}\ \textbf{L} & \cdots \cr
                S_{56}\ \textbf{L} & \cdots \cr
                S_{16}\ \textbf{L} & \cdots \cr
                S_{[15]}\ \textbf{LR} & \cdots \cr
                S_{[16}\ \textbf{LR} & \cdots }\
                \begin{matrix}
                \textcolor{red}{\bullet} \\ \textcolor{red}{\bullet} \\ \textcolor{red}{\bullet} \\ \textcolor{red}{\bullet} \\ \textcolor{red}{\bullet} \\ \\ \textcolor{blue}{\bullet} \\ \textcolor{blue}{\bullet} \\ \textcolor{blue}{\bullet} \\ \textcolor{blue}{\bullet} \\ \textcolor{blue}{\bullet} \\ \\ \\
                \end{matrix} } \]}
\caption{The matrices $\mathbb{B}_1$, $\mathbb{B}_2$, $\mathbb{B}_3$ and $\mathbb{B}_6$.} \label{fig:matricesB}
\end{figure}

Since $S_{25}$, $S_{26}$, $S_{52}$ and $S_{62}$ are complete to $K_2$, then they are not considered for the definition of the matrix $\mathbb B_2$. The same holds for $S_{13}$ with regard to $\mathbb B_1$ and $S_{63}$ with regard to $\mathbb B_3$.
Notice that we consider $S_{16}$ and $S_{[16}$ as two distinct subsets of $S$. Moreover, every vertex in $S_{[16}$ is labeled with LR and every vertex in $S_{16}$ is labeled with L.
We consider $S_{65}$ as the subset of vertices in $S$ that are complete to $K_1, \ldots, K_4$, are adjacent to $K_5$ and $K_6$ but are not complete to $K_5$.
Furthermore, any vertex in $S_{[15]}$ is represented by an empty LR-row in $\mathbb B_6$. It follows from the definition of enriched matrix that every row corresponding to a vertex in $S_{[15]}$ in $\mathbb B_6$ must be colored with the same color.
We give more details on this further on in Section~\ref{subsubsec:4tent3}.


\begin{remark}
    Claim~\ref{claim:tent_0} holds for any two vertices $v_1$ in $S_{ij}$ and $v_2$ in $S_{ik}$, for every $\mathbb B_i$, $i \in \{1, \ldots, 6\}$. 
    and the proof is anal\-o\-gous as in the tent case.
\end{remark}

\subsubsection{Split circle equivalence} \label{subsubsec:4tent3}

In this section, we state a result analogous to Lemma~\ref{lema:no2nested_prohibido_tent}. The matrices $\mathbb B_i$ contain no LR-rows, for each $i \in \{1, \ldots, 5 \}$, hence the proof is very similar to the one given in Section~\ref{subsubsec:tent3} for the tent case. We leave out the details of the proof, which can be consulted in~\cite{P20}.
Afterwards, we state and prove a lemma analogous to Lemma~\ref{lema:equiv_circle_2nested_4tent_sinLR} but for the matrix $\mathbb B_6$.
The main difference between this matrix and the matrices $\mathbb B_i$ for each $i=1, 2, \ldots, 5$ is that $\mathbb B_6$ contains LR-rows

\begin{lema}[\cite{P20}]  \label{lema:equiv_circle_2nested_4tent_sinLR} 
    If $\mathbb B_i$ is not $2$-nested, for some $i \in \{ 1, \ldots, 5 \}$, then $G$ contains one of the forbidden induced subgraphs in Figure~\ref{fig:forb_graphs}. 
\end{lema}

Now we will focus on the matrix $\mathbb B_6$. First, we define how to color those rows that correspond to vertices in $S_{[15]}$, since we defined $\mathbb B_6$ as an enriched matrix and these are the only empty LR-rows in $\mathbb B_6$.
Remember that all the empty LR-rows must be colored with the same color. Hence, if there is at least one red row labeled with L or one blue row labeled with R (resp.\ blue row labeled with L or red row labeled with R), then we color every LR-row in $S_{[15]}$ with blue (resp.\ with red).
This gives a $1$-color assignment to each empty LR-row only if $G$ is \fsc-free.

\begin{lema} \label{lema:4tent_coloreoLRvacias}
    Let $G$ be a split graph that contains an induced $4$-tent and such that $G$ contains no induced tent, and let $\mathbb B_6$ as defined in the previous section. If $S_{[15]} \neq \emptyset$ and one of the following holds:
    \begin{itemize}
    \item There is at least one red row $f_1$ and one blue row $f_2$, both labeled with L (resp.\ R)
    \item There is at least one row $f_1$ labeled with L and one row $f_2$ labeled with R, both colored with red (resp.\ blue).
    \end{itemize}
Then, we find either $F_1(5)$, $3$-sun with center or $4$-sun as an induced subgraph of $G$.
\end{lema}

\begin{proof}
    We assume that $\mathbb B_6$ contains no $D_0$, for we will prove this in Lemma~\ref{lema:B6_2nested_4tent}.
    Let $v_1$, $v_2$ and $w$ be three vertices corresponding to rows of $\mathbb B_6$, $v_1$ is colored with red and labeled with L, $v_2$ is colored with blue and labeled with L and $w$ in $S_{[15]}$.
    Thus, $v_1$ in $S_{36} \cup S_{46}$ and $v_2$ in $S_{56} \cup S_{26} \cup S_{16}$.
    In either case, we find $F_1(5)$ induced by $\{ k_2$, $k_4$, $k_5$, $k_6$, $v_1$, $v_2$, $w$, $s_{24}$, $s_{45} \}$ or $\{ k_1$, $k_2$, $k_4$, $k_6$, $v_1$, $v_2$, $w$, $s_{12}$, $s_{24} \}$, depending on whether $v_2$ in $S_{56}$ or in $S_{26} \cup S_{16}$, respectively.
    Suppose now that $v_2$ corresponds to a red row labeled with R. Thus, $v_1$ in $S_{36} \cup S_{46}$ and $v_2$ in $S_{61} \cup S_{64} \cup S_{65}$.
     If $v_2$ in $S_{61}$, then there is a $4$-sun induced by $\{ k_1$, $k_2$, $k_4$, $k_6$, $v_1$, $v_2$, $s_{12}$, $s_{24} \}$. If instead $v_2$ in $S_{64} \cup S_{65}$, then we find a $3$-sun with center induced by $\{ k_6$, $k_1$, $k_4$, $k_5$, $v_1$, $v_2$, $w \}$.
    This finishes the proof since the other cases are analogous by symmetry.
\end{proof}

To prove the following lemma, we assume without loss of generality that $S_{[15]} = \emptyset$.

\begin{lema}\label{lema:B6_2nested_4tent}
    Let $G=(K,S)$ be a split graph containing an induced $4$-tent such that $G$ contains no induced tent and let $B = \mathbb B_6$.
     If $B$ is not $2$-nested, then $G$ contains one of the graphs listed in Figure~\ref{fig:forb_graphs} as an induced subgraph. 
\end{lema}

\begin{proof}
    We assume proven Lemma~\ref{lema:equiv_circle_2nested_4tent_sinLR} by simplicity. This is, we assume that the matrices $\mathbb B_1, \ldots,$ $\mathbb B_5$ are $2$-nested. In particular, any pair of vertices $v_1$ in $S_{ij}$ and $v_2$ in $S_{ik}$ such that $i\neq 6$ and $j \neq k$ are nested in $K_i$. Moreover, there is a vertex $v*_i$ in $K_i$ adjacent to both $v_1$ and $v_2$.
    Also notice that $B$ contains no $M_0$, $M_{II}(4)$, $M_V$, $S_0(k)$ for even $k \geq 4$ or any of the matrices in $\mathcal{F}$ for we find the homonymous forbidden induced subgraphs in each case.

  The structure of the proof is analogous as in Lemmas~\ref{lema:no2nested_prohibido_tent} and~\ref{lema:equiv_circle_2nested_4tent_sinLR}. The main difference is that $B$ may have some LR-rows and thus we also have to consider what happens if $B$ contains every subconfiguration with at least one LR-row in each case.

\begin{mycases}
    \case \textit{Suppose that $B$ is not admissible. }
    Hence, $B$ contains at least one of the matrices $D_0, D_1, \ldots, D_{13}$, $S_1(j), S_2(j), \ldots, S_8(j)$ for some $j \geq 3$ or $P_0(j,l)$, $P_1(j,l)$ for some $l \geq 0, j \geq 5$ or $P_2(j,l)$, for some $l \geq 0, j \geq 7$.
    \subcase \textit{$B$ contains $D_0$}.
    Let $v_1$ and $v_2$ be the vertices represented by the first and second row of $D_0$ respectively, and $k_{61}$, $k_{62}$ in $K_6$ represented by the first and second column of $D_0$, respectively.
    \subsubcase Suppose first that $v_1$ and $v_2$ are colored with the same color. Since the case is symmetric with regard of the coloring, we may assume that both rows are colored with red. Hence, either $v_1$ and $v_2$ lie in $S_{61} \cup S_{64} \cup S_{65}$, or $v_1$ and $v_2$ lie in $S_{36} \cup S_{46}$.
    If $v_1$ and $v_2$ lie in $S_{61}$ and $k_1$ in $K_1$ is adjacent to both $v_1$ and $v_2$, then we find a $M_{III}(3)$ induced by $\{ k_{61}$, $k_{62}$, $k_1$, $k_2$, $v_1$, $v_2$, $ s_{12} \}$. We find $M_{III}(3)$ if either $v_1$ and $v_2$ lie in $S_{64} \cup S_{65}$ changing $k_1$ for some $k_4$ in $K_4$ adjacent to both $v_1$ and $v_2$, $k_2$ for some $k_5$ in $K_5$ nonadjacent to both $v_1$ and $v_2$ and $s_{12}$ for $s_{45}$. We also find the same subgraph if $v_1$ and $v_2$ lie in $S_{36} \cup S_{46}$, changing $k_1$ for some $k_4$ in $K_4$ adjacent to both $v_1$ and $v_2$ and $s_{12}$ for $s_{24}$.
    If instead $v_1$ in $S_{61}$ and $v_2$ in $S_{64} \cup S_{65}$, since we defined $S_{65}$ as those vertices adjacent but not complete to $K_5$, then there are vertices $k_4$ in $K_4$ and $k_5$ in $K_5$ such that $v_1$ is nonadjacent to both, and $v_2$ is adjacent to $k_4$ and is nonadjacent to $k_5$. Thus, we find $F_{2}(5)$ induced by $\{ k_{62}$, $k_1$, $k_2$, $k_4$, $k_5$, $v_1$, $v_2$, $s_{45}$, $s_{12}$, $s_{24} \}$.

    \subsubcase Suppose now that $v_1$ is colored with red and $v_2$ is colored with blue. Hence, $v_1$ lies in $S_{62} \cup S_{63}$, and $v_2$ lies in $S_{61} \cup S_{64} \cup S_{65}$.
    If $v_2$ in $S_{61}$, then there is a vertex $k_4$ in $K_4$ nonadjacent to $v_1$ and $v_2$. Hence, we find $M_{III}(4)$ induced by $\{ k_{61}$, $k_{62}$, $k_1$, $k_2$, $k_4$, $v_1$, $v_2$, $s_{12}$ $s_{24} \}$. If instead $v_2$ in $S_{64} \cup S_{65}$, then we find $M_{III}(4)$ induced by $\{ k_{61}$, $k_{62}$, $k_2$, $k_4$, $k_5$, $v_1$, $v_2$, $s_{24}$, $s_{45} \}$.

    \subcase \textit{$B$ contains $D_1$. }
    Let $v_1$ and $v_2$ be the vertices that represent the rows of $D_1$ and let $k_6$ in $K_6$ be the vertex that represents the column of $D_1$. Suppose without loss of generality that both rows are colored with red, hence $v_1$ in $S_{36} \cup S_{46}$ and $v_2$ in $S_{61} \cup S_{64} \cup S_{65}$. Notice that $v_1$ is complete to $K_4$ since $S_{46} = S_{[46}$. 
     If $v_2$ in $S_{61}$, then we find a $4$-sun induced by $\{ k_6$, $k_1$, $k_2$, $k_4$, $v_1$, $v_2$, $s_{12}$, $s_{24} \}$. If $v_2$ in $S_{64}$ is not complete to $K_4$, then we find a tent induced by $\{ k_6$, $k_2$, $k_4$, $v_1$, $v_2$, $s_{24} \}$.
     If instead $v_2$ in $S_{64} \cup S_{65}$ is complete to $K_4$, then we find a $M_{II}(4)$ induced by $\{ k_2$, $k_4$, $k_5$, $k_6$, $v_1$, $v_2$, $s_{24}$, $s_{45} \}$.

    \subcase \textit{$B$ contains $D_2$. }
    Let $v_1$ and $v_2$ be the first and second row of $D_2$, and let $k_{61}$ and $k_{62}$ be the vertices corresponding to first and second column of $D_2$, respectively. Suppose that $v_1$ is colored with blue and $v_2$ is colored with red, thus $v_1$ lies in $S_{56} \cup S_{26} \cup S_{16}$ and $v_2$ lies in $S_{61} \cup S_{64} \cup S_{65}$.
    If $v_1$ in $S_{56}$ and $v_2$ in $S_{61}$, then we find a $5$-sun with center induced by $\{ k_{61}$, $k_{62}$, $k_1$, $k_2$, $k_4$, $k_5$, $v_1$, $v_2$, $s_{12}$, $s_{24}$, $s_{45} \}$.
    If instead $v_1$ in $S_{26} \cup S_{16}$, since $v_1$ is not complete to $K_1$ and we assume that $\mathbb B_1$ is admissible and thus contains no $D_1$, then there is a vertex $k_1$ in $K_1$ adjacent to $v_2$ and nonadjacent to $v_1$
    We find a tent induced by $\{ k_{61}$, $k_1$, $k_2$, $v_1$, $v_2$, $s_{12} \}$. The same holds if $v_1$ in $S_{56}$ and $v_2$ in $S_{65}$, for $\mathbb B_5$ is admissible and $v_2$ is adjacent but not complete to $K_5$.
    Moreover, if $v_1$ in $S_{56}$ and $v_2$ in $S_{64}$, then we find a tent induced by $\{ k_{61}$, $k_4$, $k_5$, $v_1$, $v_2$, $s_{45} \}$.
    Finally, if $v_1$ in $S_{26} \cup S_{16}$ and $v_2$ in $S_{64} \cup S_{65}$, then there are vertices $k_1$ in $K_1$ and $k_5$ in $K_5$ such that $k_1$ is nonadjacent to $v_1$ and adjacent to $v_2$, and $k_5$ is nonadjacent to $v_2$ and adjacent to $v_1$. Hence, we find $F_1(5)$ induced by $\{ k_1$, $k_2$, $k_4$, $k_5$, $v_1$, $v_2$, $s_{12}$, $s_{24}$, $s_{45} \}$.

\begin{remark} \label{obs:4tent_modelo0}
    If $G$ is circle and $S_{26} \cup S_{16} \neq \emptyset$, then $S_{64} \cup S_{65} = \emptyset$, and viceversa.
\end{remark}

    \subcase \textit{$B$ contains $D_3$.}
    Let $v_1$ and $v_2$ be the vertices corresponding to the rows of $D_3$ labeled with L and R, respectively, $w$ be the vertex corresponding to the LR-row, and $k_{61}$, $k_{62}$ and $k_{63}$ in $K_6$ be the vertices corresponding to the columns of $D_3$. An uncolored LR-row in $B$ represents a vertex in $S_{[16}$.
    Notice that there is no vertex $k_i$ in $K_i$ for some $i \in \{1, \ldots, 5\}$ adjacent to both $v_1$ and $v_2$, since $w$ is complete to $K_i$, thus we find $M_{III}(3)$ induced by $\{ k_{61}$, $k_{62}$, $k_{63}$, $k_i$, $v_1$, $v_2$, $w \}$.

    If $v_1$ is colored with red and $v_2$ is colored with blue, then $v_1$ in $S_{36} \cup S_{46}$ and $v_2$ in $S_{62} \cup S_{63}$. We find $M_{III}(4)$ induced by $\{ k_{62}$, $k_2$, $k_4$, $k_{61}$, $k_{63}$, $v_1$, $v_2$, $w$, $s_{24} \}$.

    Conversely, if $v_1$ is colored with blue and $v_2$ is colored with red, then $v_1$ in $S_{56} \cup S_{26} \cup S_{16}$ and $v_2$ in $S_{61} \cup S_{64} \cup S_{65}$.
    It follows by symmetry that it suffices to see what happens if $v_2$ in $S_{61}$ and $v_1$ in either $S_{56}$ or $S_{26}$.
    If $v_1$ in $S_{56}$, then we find $M_{III}(6)$ induced by $\{ k_{61}$, $k_1$, $k_2$, $k_4$, $k_5$, $k_{62}$, $k_{63}$, $v_1$, $v_2$, $s_{12}$, $s_{24}$, $s_{45}$, $w \}$. If instead $v_1$ in $S_{26}$, then we find $M_{III}(4)$ induced by $\{ k_{63}$, $k_{61}$, $k_1$, $k_2$, $k_{62}$, $v_1$, $v_2$, $s_{12}$, $w \}$.



    \subcase \textit{$B$ contains $D_4$.}
    Let $v_1$ and $v_2$ be the vertices represented by the rows labeled with L, $w$ be the vertex represented by the LR-row and $k_6$ in $K_6$ corresponding to the column of $D_4$. Since $B$ contains no $D_1$, suppose that $v_1$ is colored with red and $v_2$ is colored with blue. Thus, $v_1$ lies in $S_{61} \cup S_{64} \cup S_{65}$ and $v_2$ lies in $S_{62} \cup S_{63}$.
    In either case, we find $F_1(5)$: if $v_1$ in $S_{61}$, then it is induced by $\{ k_6$, $k_1$, $k_2$, $k_4$, $v_1$, $v_2$, $w$, $s_{12}$, $s_{24} \}$, and if $v_1$ in $S_{64} \cup S_{65}$, then it is induced by $\{ k_6$, $k_2$, $k_4$, $k_5$, $v_1$, $v_2$, $w$, $s_{24}$, $s_{45} \}$.

    \subcase \textit{$B$ contains $D_5$.} Let $v_1$ and $v_2$ be the vertices representing the rows labeled with L and R, respectively, $w$ be the vertex corresponding to the LR-row, and $k_6$ in $K_6$ corresponding to the column of $D_5$. Suppose $v_1$ is colored with blue and $v_2$ is colored with red.
    Notice that for any two vertices $x_1$ in $S_{ij}$ and $x_2$ in $S_{jk}$, we may assume that there are vertices $k_{j1}$ and $k_{j2}$ in $K_j$ such that $x_1$ is adjacent to $k_{j1}$ and is nonadjacent to $k_{j2}$ and $x_2$ is adjacent to $k_{j2}$ and is nonadjacent to $k_{j1}$, since we assume $\mathbb B_i$ admissible for each $i \in \{1, \ldots, 5\}$.
    It follows that, if $v_1$ in $S_{26} \cup S_{16}$ and $v_2$ in $S_{61}$, then there is a tent induced by $\{k_6$, $k_1$, $k_2$, $v_1$, $v_2$, $s_{12} \}$, where $k_1$ is a vertex nonadjacent to $v_1$. The same holds if $v_1$ in $S_{56}$ and $v_2$ in $S_{65}$, where the tent is induced by $\{ k_6$, $k_4$, $k_5$, $v_1$, $v_2$, $s_{45} \}$, with $k_5$ in $K_5$ adjacent to $v_1$ and nonadjacent to $v_2$.
    Finally, if $v_1$ in $S_{56}$ and $v_2$ in $S_{61}$, then we find a $5$-sun with center induced by $\{ k_5$, $k_6$, $k_1$, $k_2$, $k_4$, $v_1$, $v_2$, $w$, $s_{12}$, $s_{24}$, $s_{45} \}$.

    \begin{remark} \label{obs:4tent_modelo2}
    If $G$ is circle and contains no induced tent, then any two vertices $v_1$ in $S_{56}$ and $v_2$ in $S_{65}$ are disjoint in $K_5$. The same holds for any two vertices $v_1$ in $S_{16}$ and $v_2$ in $S_{61}$ in $K_1$.
    \end{remark}

    \subcase \textit{$B$ contains $D_6$}. Let $v_1$ and $v_2$ be the vertices represented by the rows labeled with L and R, respectively, $w$ be the vertex corresponding to the LR-row, and $k_{61}$ and $k_{62}$ in $K_6$ corresponding to the first and second column of $D_6$, respectively. Suppose without loss of generality that $v_1$ and $v_2$ are both colored with red, thus $v_1$ lies in $S_{36} \cup S_{46}$ and $v_2$ lies in $S_{61} \cup S_{64} \cup S_{65}$.
    Since $\mathbb B_i$ is admissible for every $i \in \{1, \ldots, 5\}$, then there is no vertex in $K_i$ adjacent to both $v_1$ and $v_2$. It follows that, since $v_1$ is complete to $K_4$, then $v_2$ in $S_{61}$. However, we find $F_2(5)$ induced by $\{ k_{62}$, $k_1$, $k_2$, $k_4$, $k_{61}$, $v_1$, $v_2$, $w$, $s_{12}$, $s_{24} \}$.

    The following remark is a consequence of the previous statement.
    \begin{remark} \label{obs:4tent_modelo1}
    If $G$ is circle, $v_1$ in $S_{36}\cup S_{46}$ and $v_2$ in $S_{61} \cup S_{64} \cup S_{65}$, then for every vertex $w$ in $S_{[16}$ either $N_{K_6}(v_1) \subseteq N_{K_6}(w)$ or $N_{K_6}(v_2) \subseteq N_{K_6}(w)$. The same holds for $v_1$ in $S_{56} \cup S_{26} \cup S_{16}$ and $v_2$ in $S_{62} \cup S_{63}$.
    \end{remark}


    \subcase \textit{$B$ contains $D_7$ or $D_{11}$}.
    Then, there is a vertex $k_i$ in some $K_i$ with $i\neq 6$ such that $k_i$ is adjacent to the three vertices corresponding to every row of $D_7$. We find a net ${}\vee{}K_1$. 
    \subcase \textit{$B$ contains $D_8$ or $D_{12}$}.
     There is a tent induced by all three rows and columns of $D_8$ or $D_12$. 
    \subcase \textit{$B$ contains $D_9$ or $D_{13}$.}
    It is straightforward that in this case we find $F_0$.

    \subcase \textit{$B$ contains $D_{10}$}.
    Let $v_1$ and $v_2$ be the vertices represented by the rows labeled with L and R, respectively, $w_1$ and $w_2$ be the vertices represented by the LR-rows and $k_{61}, \ldots, k_{64}$ in $K_6$ be the vertices corresponding to the columns of $D_{10}$.
    Suppose that $v_1$ is colored with red and $v_2$ is colored with blue, hence $v_1$ lies in $S_{36} \cup S_{46}$ and $v_2$ lies in $S_{62} \cup S_{63}$.
    Let $k_2$ in $K_2$ adjacent to $v_2$ and nonadjacent to $v_1$ and let $k_4$ in $K_4$ adjacent to $v_1$ and nonadjacent to $v_2$. We find $F_1(5)$ induced by $\{ v_1$, $v_2$, $w_1$, $w_2$, $s_{24}$, $k_2$, $k_4$, $k_{62}$, $k_{63} \}$.

    \subcase \textit{$B$ contains $S_1(j)$.}
    \subsubcase If $j \geq 4$ is even, let $v_1, v_2, \ldots, v_j$ be the vertices represented by the rows of $S_1(j)$, where $v_1$ and $v_j$ are labeled both with L or both with R, $v_{j-1}$ is a vertex corresponding to the LR-row, and $k_{61}, \ldots, k_{6(j-1)}$ in $K_6$ the vertices corresponding to the columns. Suppose without loss of generality that $v_1$ and $v_j$ are labeled with L. It follows that either $v_1$ and $v_j$ lie in $S_{36} \cup S_{46}$, or $v_1$ and $v_j$ lie in $S_{62} \cup S_{63}$ or $v_1$ lies in $S_{56} \cup S_{26} \cup S_{16}$ and $v_j$ lies in $S_{36} \cup S_{46}$.
In either case, there is $k_5$ in $K_5$ adjacent to both $v_1$ and $v_j$ and $k_5$ is also adjacent to $v_{j-1}$ since it lies in $S_{[16}$. Thus, this vertex set induces a $(j-1)$-sun with center.

\subsubcase If $j$ is odd, since $S_1(j)$ has $j-2$ rows (thus there are $v_1, \ldots, v_{j-2}$ vertices), then the subset of vertices given by $\{ v_1, \ldots, v_{j-2}$, $k_{61}, \ldots, k_{6(j-2)}$, $k_5 \}$ induces an even $(j-1)$-sun.

 \subcase \textit{$B$ contains $S_2(j)$}.
Let $v_1$ and $v_j$ be the vertices corresponding to the labeled rows, $k_{61}, \ldots, k_{6 (j-1})$ in $K_6$ be the vertices corresponding to the columns of $S_2(j)$, and suppose without loss of generality that $v_1$ and $v_j$ are labeled with R.

 \subsubcase $j$ is odd. Suppose first that $v_1$ is colored with red and $v_j$ is colored with blue. Thus, $v_1$ in $S_{61} \cup S_{64} \cup S_{65}$ and $v_j$ in $S_{62} \cup S_{63}$. If $v_1$ in $S_{61}$, then let $k_i$ in $K_i$ for $i = 1, 2, 4$ such that $k_1$ is adjacent to $v_1$ and $v_j$, $k_2$ is adjacent to $v_j$ and nonadjacent to $v_1$, and $k_4$ is nonadjacent to both $v_1$ and $v_j$. We find $F_2(j+2)$ induced by $\{ k_4$, $k_2$, $k_1$, $k_{61}, \ldots, k_{6 (j-1})$, $v_1, \ldots, v_j$, $s_{12}$, $s_{24} \}$.
    If $v_1$ in $S_{64} \cup S_{65}$, then we find $F_2(j)$ induced by $\{ k_5$, $k_{61}, \ldots, k_{6 (j-1)}$, $v_1, \ldots, v_j \}$, with $k_5$ in $K_5$ adjacent to $v_1$ and nonadjacent to $v_j$.
     Conversely, suppose $v_1$ is colored with blue and $v_2$ is colored with red, thus $v_1$ in $S_{62}\cup S_{63}$ and $v_j$ in $S_{61} \cup S_{64} \cup S_{65}$. If $v_j$ lies in $S_{64} \cup S_{65}$, then we find $F_2(j+2)$ induced by $\{ k_2$, $k_4$, $k_5$, $k_{61}, \ldots, k_{6 (j-1)}$, $v_1, \ldots, v_j$, $s_{24}$, $s_{45} \}$, with $k_i$ in $K_i$ for $i=2, 4, 5$ such that $k_2$ is adjacent to $v_1$ and $v_k$, $k_4$ is adjacent to $v_j$ and nonadjacent to $v_1$, and $k_5$ is nonadjacent to both $v_1$ and $v_j$. If instead $v_j$ in $S_{61}$, then it is induced by $\{ k_4$, $k_2$, $k_1$, $k_{61}, \ldots, k_{6 (j-1)}$, $v_1, \ldots, v_j$, $s_{12}$, $s_{24} \}$.
    \subsubcase $j$ is even. Hence, $v_1$ and $v_j$ are colored with the same color. Suppose without loss of generality that are both colored with red, and thus $v_1$ and $v_j$ lie in $S_{61} \cup S_{64} \cup S_{65}$.
    If $v_1$ and $v_j$ in $S_{61}$, then we find $F_2(j+1)$ induced by $\{ k_2$, $k_1$, $k_{61}, \ldots, k_{6 (j-1)}$, $v_1, \ldots, v_j$, $s_{12} \}$.
    We find $F_2(j+1)$ if $v_1$ and $v_j$ lie in $S_{64}$ or $S_{65}$, changing $s_{12}$ for $s_{45}$, and $k_1$ and $k_2$ for $k_4$ and $k_5$, where $k_5$ is nonadjacent to both $v_1$ and $v_j$ and $k_4$ is adjacent to both.
    If only $v_1$ lies in $S_{61}$, then we find $F_2(j+3)$ induced by $\{ k_1$, $k_2$, $k_4$, $k_5$, $k_{61}, \ldots, k_{6 (j-1)}$, $v_1, \ldots, v_j$, $s_{12}$,$s_{24}$, $s_{45} \}$, with $k_i$ in $K_i$ for $i=1, 2, 4, 5$.
    If only $v_j$ lies in $S_{61}$, then we find $F_2(5)$ induced by $\{ k_1$, $k_2$, $k_4$, $k_5$, $k_{62}$, $v_1$, $v_j$, $s_{12}$, $s_{24}$, $s_{45} \}$, with $k_i$ in $K_i$ for $i=1, 2, 4, 5$.

    \subcase \textit{Suppose that $B$ contains $S_3(j)$}.
Let $v_1$ and $v_j$ be the vertices corresponding to the labeled rows, $k_{61}, \ldots, k_{6 (j-1})$ in $K_6$ be the vertices corresponding to the columns of $S_3(j)$.

    \subsubcase $j$ is odd. Suppose that $v_1$ is labeled with L and colored with blue and $v_j$ is labeled with R and colored with red. In this case, $v_1$ in  $S_{56} \cup S_{26} \cup S_{16}$ and $v_j$ in $S_{61} \cup S_{64} \cup S_{65}$.
    If $v_1$ in $S_{56}$, then we find a $(j+3)$-sun if $v_j$ in $S_{61}$, induced by $\{ k_1$, $k_2$,  $k_4$, $k_5$, $k_{61}, \ldots, k_{6 (j-1)}$, $v_1, \ldots, v_j$, $s_{45}$, $s_{12}$, $s_{24} \}$.
    If $v_j$ in $S_{64} \cup S_{65}$, then we find a $(j+1)$-sun induced by $\{ k_4$, $k_5$, $k_{61}, \ldots, k_{6 (j-1)}$, $v_1, \ldots, v_j$, $s_{45} \}$.
    Moreover, if $v_j$ in $S_{61}$ and $v_1$ in $S_{26} \cup S_{16}$, then we find a $(j+1)$-sun induced by $\{ k_1$, $k_{61}, \ldots, k_{6 (j-1)}$, $k_2$, $v_1, \ldots, v_j$, $s_{12} \}$. Finally, it follows from Remark~\ref{obs:4tent_modelo0} that it is not possible that $v_j$ in $S_{64} \cup S_{65}$ and $v_1$ in $S_{26} \cup S_{16}$.

    \subsubcase $j$ is even. Suppose without loss of generality that $v_1$ and $v_j$ are both colored with red. Thus, $v_1$ in $S_{61} \cup S_{64} \cup S_{65}$ and $v_j$ in $S_{36} \cup S_{46}$.
    If $v_1$ in $S_{61}$, then we find $(j+2)$-sun induced by $\{ k_4$, $k_2$, $k_1$, $k_{61}, \ldots, k_{6 (j-1)}$, $v_1, \ldots, v_j$, $s_{12}$, $s_{24} \}$. If instead $v_1$ in $S_{64} \cup S_{65}$, then we find $j$-sun induced by $\{ k_4$, $k_{61}, \ldots, k_{6 (j-1)}$, $v_1, \ldots, v_j \}$.

    \subcase \textit{$B$ contains $S_4(j)$.}
    Let $v_1$ be the vertex corresponding to the row labeled with LR, $v_2$ corresponding to the row labeled with L, $v_j$ labeled with R and $k_{61}, \ldots, k_{6 (j-1})$ in $K_6$ the vertices corresponding to the columns of $S_4(j)$. 

    \subsubcase $j$ is even. Hence, $v_2$ and $v_j$ are colored with the same color. Suppose without loss of generality that they are both colored with red, thus $v_2$ in $S_{36} \cup S_{46}$ and $v_j$ in $S_{61} \cup S_{64} \cup S_{65}$.
    If $v_j$ lies in $S_{64} \cup S_{65}$, then we find a $(j-1)$-sun with center induced by $\{ k_4$, $k_{61}, \ldots, k_{6(j-1)}$, $v_1$, $2, \ldots, v_j \}$. If instead $v_j$ in $S_{61}$, then we find a $(j+1)$-sun with center induced by $\{ k_{61}, \ldots, k_{6(j-1)}$, $k_1$, $k_2$, $k_4$, $v_1$, $2, \ldots, v_j$, $s_{12}$, $s_{24} \}$.

    \subsubcase $j$ is odd. Suppose without loss of generality that $v_2$ is colored with red and $v_j$ is colored with blue. Hence, $v_2$ in $S_{36} \cup S_{46}$ and $v_j$ in $S_{62} \cup S_{63}$. We find a $j$-sun with center induced by $\{ k_4$, $k_{61}, \ldots, k_{6(j-1)}$, $k_2$,  $v_1$, $2, \ldots, v_j$, $s_{24} \}$.

    \subcase \textit{$B$ contains $S_5(j)$.}
    Let $v_1$ and $v_j$ be the vertices representing the rows labeled with L, $v_{j-1}$ the vertex corresponding to the row labeled with LR and $k_{61}, \ldots, k_{6 (j-2})$ in $K_6$ be the vertices corresponding to the columns of $S_4(j)$.

    \subsubcase $j$ is even. Hence $v_1$ and $v_j$ lie in $S_{36} \cup S_{46}$. We find $F_1(j+1)$ induced by $\{ k_2$, $k_4$, $k_{61}, \ldots, k_{6(j-2)}$, $v_1, \ldots, v_{j-1}, v_j$, $s_{24} \}$.

    \subsubcase $j$ is odd. Suppose $v_1$ is colored with red and $v_j$ is colored with blue, thus $v_1$ in $S_{36} \cup S_{46}$ and $v_j$ in $S_{56} \cup S_{26} \cup S_{16}$.
    If $v_j$ in $S_{56}$, then we find $F_1(j)$ induced by $\{ k_4$, $k_5$, $k_{61}, \ldots, k_{6(j-2)}$, $v_1, \ldots, v_{j-1}, v_j$, $s_{45} \}$. If instead $v_j$ lies in $S_{26} \cup S_{16}$, then we find $F_1(j+2)$ induced by $\{ k_1, $, $k_2$, $k_4$, $k_{61}, \ldots, k_{6(j-2)}$, $v_1, \ldots, v_{j-1}, v_j$, $s_{24}$, $s_{12} \}$.

    \subcase \textit{$B$ contains $S_6(j)$.}
    \subsubcase Suppose first that $B$ contains $S_6(3)$ or $S'_6(3)$. Let $v_1$, $v_2$ and $v_3$ be the vertices that represent the LR-row, the R-row and the unlabeled row, respectively. In\-de\-pen\-dent\-ly on where lies $v_2$, there is vertex $v$ in $K \setminus K_6$ such that $v$ is adjacent to $v_1$ and $v_2$ and is nonadjacent to $v_3$ and thus we find an induced $3$-sun with center.

    \subsubcase If $B$ contains $S_6(j)$ for some even $j\geq 4$, then we find $F_1(j)$ induced by every row and column of $S_6(j)$. If instead $j$ is odd, then we find $M_{II}(j)$ induced by every row and column of $S_6(j)$ and a vertex $k_i$ in some $K_i$ with $i \neq 6$. We choose such a vertex $k_i$ adjacent to $v_2$, and thus since $v_1$ in $S_{[16}$, $v_1$ is also adjacent to $k_i$ and $v_3, \ldots, v_j$ are nonadjacent to $k_i$ for they represent vertices in $S_{66}$.

    \subcase \textit{$B$ contains $S_7(j)$.}
    Suppose $B$ contains $S_7(3)$. It is straightforward that the rows and columns induce $M_{II}(4)$.
    Furthermore, if $j >3$, then $j$ is even. The rows and columns of $S_7(j)$ induce a $j$-sun.

    \subcase \textit{$B$ contains $S_8(2j)$.}
    If $j=2$, then we find a tent induced by the last three columns and the last three rows.
    If instead $j>2$, then we find a $(2j-1)$-sun with center induced by every unlabeled row, every column but the first and any vertex in $K_1$ --which will be the center--, since $K_1 \neq \emptyset$.

    \subcase \textit{$B$ contains $P_0(j,l)$.}
    Let $v_1, \ldots, v_j$ in $S$ and $k_{61}, \ldots, k_{6j}$ in $K_6$ be the vertices represented by the rows and the columns of $P_0(j,l)$, respectively. The rows $v_1$ and $v_j$ are labeled with L and R, respectively, and $v_{l+2}$ is an LR-row.

    Suppose $l=0$. If $j$ is even, then $v_1$ and $v_j$ are colored with distinct colors. Suppose without loss of generality that $v_1$ is colored with red, thus $v_1$ lies in $S_{36}\cup S_{46}$ and $v_j$ lies in $S_{62} \cup S_{63}$. In that case, there are vertices $k_i$ in $K_i$ for $i=2,4$ such that $k_2$ is adjacent to $v_j$ and nonadjacent to $v_1$ and $k_4$ is adjacent to $v_1$ and nonadjacent to $v_j$. We find $F_2(j+1)$ induced by $\{ k_2, $, $k_4$, $k_{62}, \ldots, k_{6j}$, $v_1, \ldots, v_j$, $s_{24} \}$

    If instead $j$ is odd, then $v_1$ and $v_j$ are colored with the same color. Suppose without loss of generality that they are both colored with red. Hence, $v_1$ lies in $S_{36} \cup S_{46}$ and $v_j$ lies in $S_{61} \cup S_{64} \cup S_{65}$. We find $F_2(j+2)$ induced by $\{ k_1, $, $k_2$, $k_4$, $k_{62}, \ldots, k_{6j}$, $v_1, \ldots, v_j$, $s_{24}$, $s_{12} \}$ if $v_j$ lies in $S_{61}$, and by $\{ k_2, $, $k_4$, $k_5$, $k_{62}, \ldots, k_{6j}$, $v_1, \ldots, v_j$, $s_{24}$, $s_{45} \}$ if $v_j$ lies in $S_{64} \cup S_{65}$.

    The proof is analogous if $l >0$.


    \subcase \textit{$B$ contains $P_1(j,l)$.}
        Let $v_1, \ldots, v_j$ and $k_{61}, \ldots, k_{6(j-1)}$ in $K_6$ be the vertices represented by the rows and the columns of $P_1(j,l)$, respectively, where $v_1$ and $v_j$ are labeled with L and R, respectively, and $v_{l+2}$ and $v_{l+3}$ are LR-rows.

    \subsubcase Suppose first that $l=0$.
    If $j$ is odd, then $v_1$ and $v_j$ are colored with the same color. Suppose without loss of generality that they are colored with red, thus, $v_1$ lies in $S_{36} \cup S_{46}$ and $v_j$ lies in $S_{61} \cup S_{64} \cup S_{65}$. In either case, $v_1$ is anticomplete to $K_1$. Hence, we find $F_1(j)$ induced by every row and column of $P_1(j,0)$ and an extra column that represents a vertex in $K_1$ adjacent to $v_j$, $v_2$ and $v_3$ and nonadjacent to $v_i$, for $1 \leq i \leq j-1$, $i \neq j,2,3$.
    If instead $j$ is even, then suppose $v_1$ and $v_j$ are colored with red and blue, respectively. Thus, $v_1$ lies in $S_{36} \cup S_{46}$ and $v_j$ lies in $S_{62} \cup S_{63}$. We find $F_1(j+1)$ induced by the vertices that represent every row and column of $P_1(j,0)$, $s_{24}$ and two vertices $k_2$ in $K_2$ and $k_4$ in $K_4$ such that $k_2$ is adjacent to $v_j$, $v_2$ and $v_3$ and is nonadjacent to $v_i$, and $k_4$ is adjacent to $v_1$, $v_2$ and $v_3$ and is nonadjacent to $v_i$, for each $1 \leq i \leq j-1$, $i \neq j, 2, 3$.

    \subsubcase Suppose $l>0$. The proof is analogous to the previous case if $j$ is even. If instead $j$ is odd, then $v_1$ lies in $S_{36} \cup S_{46}$ and $v_j$ lies in $S_{61} \cup S_{64} \cup S_{65}$. If $v_j$ in $S_{61}$, then we find $F_1(j+2)$ induced by $\{ k_4$, $k_{61}, \ldots, k_{6(j-2)}$, $k_1$, $k_2$, $v_1, \ldots, v_j$, $s_{12}$, $s_{24} \}$.
    If instead $v_j \not\in S_{61}$, then we find $F_1(j)$ induced by the vertices corresponding to every row and column of $P_1(j,l)$ and a vertex in $K_4$ adjacent to every vertex represented by a labeled row.

    \subcase \textit{$B$ contains $P_2(j,l)$.}
    Let $v_1, \ldots, v_j$ and $k_{61}, \ldots, k_{6(j-1)}$ in $K_6$ be the vertices represented by the rows and the columns of $P_2(j,l)$, respectively, where $v_1$ and $v_j$ are labeled with L and R, respectively, and $v_{l+2}$, $v_{l+3}$, $v_{l+4}$ and $v_{l+5}$ are LR-rows.

    Suppose $l=0$. If $j$ is even, then we find $F_1(j-1)$ induced by $\{ k_{62}$, $k_{65}, \ldots, k_{6(j-1)}$, $v_1$, $v_2$, $v_5, \ldots, v_j$, $s_{24} \}$. The same subgraph arises if $l>0$.

     If $j$ is odd, then $v_1$ and $v_j$ are colored with the same color, thus we assume that $v_1$ lies in $S_{36} \cup S_{46}$ and $v_j$ lies in $S_{61} \cup S_{64} \cup S_{65}$.
     If $v_j \not\in S_{61}$, then we find $F_1(j-2)$ induced by $\{ k_{61}$, $k_{62}$, $k_{65}, \ldots, k_{6(j-1)}$, $v_1$, $v_2$, $v_5, \ldots, v_j$, $k_4 \}$, where $k_4$ in $K_4$ is adjacent to $v_1$, $v_2$, $v_5$ and $v_j$. The same subgraph arises if $l>0$.
     If $v_j$ in $S_{61}$, then there are vertices $k_i$ in $K_i$, for $i=1,2,4$ such that $k_1$ is adjacent to $v_j$ and is nonadjacent to $v_1$, $k_2$ is nonadjacent to both and $k_4$ is adjacent to $v_1$ and nonadjacent to $v_j$. If $l=0$, we find $M_{II}(j)$ induced by $\{ k_{62}$, $k_{63}$, $k_{65}, \ldots, k_{6(j-1)}$, $v_1$, $v_2$, $v_5, \ldots, v_j$, $k_1$, $k_2$, $k_4$, $s_{12}$, $s_{24} \}$.
      If instead $l>0$, then we find $F_1(j)$ induced by  $\{ k_{61}$, $k_{62}$, $k_{64}, \ldots, k_{6(j-1)}$, $k_1$, $k_2$, $k_4$, $v_1$, $v_2$, $v_3$, $v_6, \ldots, v_j$, $s_{12}$, $s_{24} \}$.


    \vspace{3mm}
    \case Suppose now that $B$ is admissible but $B^*$ contains either a Tucker matrix, or one of the matrices in Figure~\ref{fig:forb_LR-orderable}. 
    It suffices to see that $B^*_\tagg$ contains no Tucker matrix, for in the case of the matrices listed in Figure~\ref{fig:forb_LR-orderable}, each labeled row admits a vertex belonging to the same subsets of $K$ considered in the analysis for a Tucker matrix having at least one LR-row. 
Towards a contradiction, let $M$ be a Tucker matrix contained in $B^*_\tagg$. Throughout the proof, when we refer to an LR-row in $M$, we refer to the row in $B$, this is, the complement of the row that appears in $M$.

    \subcase Suppose first that $M=M_I(j)$, for some $j \geq 3$. Let $v_1, \ldots, v_j$ and $k_{61}, \ldots, k_{6j}$ in $K_6$ be the vertices corresponding to the rows and the columns of $M$, respectively.

    \begin{remark} \label{obs:no_Di_in_MI}
    If two non-LR-rows in $M=M_I(j)$ are labeled with the same letter, then they induce $D_0$. Moreover, any pair of consecutive non-LR-rows labeled with distinct letters induce $D_1$ or $D_2$. 
        Since we assume $B$ admissible, then there are at most two labeled non-LR-rows in $M_I(j)$ and such rows are non-consecutive and labeled with distinct letters.
    Furthermore, 
    it is easy to see that there are at most two LR-rows in $M_I(j)$, for if not such rows induce $D_{11}$, $D_{12}$ or $D_{13}$.
    \end{remark}

    \subsubcase Suppose $M=M_I(3)$. Suppose first that $v_1$ is the only LR-row in $M$.

    If $v_2$ and $v_3$ are unlabeled, then we find 
    $M_{III}(3)$ induced by $\{v_1$, $v_2$, $v_3$, $k_{61}$, $k_{62}$, $k_{63}$, $ k_l \}$, where $k_l$ is any vertex in $K_l \neq K_6$. The same holds if either $v_2$ or $v_3$ are labeled rows, by accordingly replacing $k_l$ for some $l$ such that $k_l$ is nonadjacent to both $v_2$ and $v_3$ (there are no labeled rows complete to each partition $K_i \neq K_6$ of $K$).
    By Remark~\ref{obs:no_Di_in_MI}, if both $v_2$ and $v_3$ are labeled rows, then they are labeled with distinct letters. Thus, we find $F_0$ induced by $\{ v_1$, $v_2$, $v_3$, $k_{61}$, $k_{62}$, $k_{63}$, $k_1$, $k_5  \}$, where $k_1$ in $K_1$ is adjacent to $v_2$ and nonadjacent to $v_3$ and $k_5$ in $K_5$ is adjacent to $v_3$ and nonadjacent to $v_2$, or viceversa. Such vertices exist since $\mathbb B_i$ is admissible for every $i \in \{1, \ldots, 5\}$.
    If instead $v_1$ and $v_2$ are LR-rows, then we find a tent induced by $\{ v_1$, $v_2$, $v_3$, $k_{61}$, $k_{63}$, $ k_l \}$, considering $k_l$ in $K_l$ for some $l \in \{1, \ldots, 5\}$ such that $v_3$ is nonadjacent to $k_l$. Every other case is analogous by symmetry.
Moreover, if $v_1$, $v_2$ and $v_3$ are LR-rows, then there is a vertex $k_l$ in $K_l$ with $l \neq 6$ such that $v_1$, $v_2$ and $v_3$ are adjacent to $k_l$, hence we find a $M_{III}(3)$ induced by $\{ v_1$, $v_2$, $v_3$, $k_{61}$, $k_{62}$, $k_{63}$, $ k_l \}$.

    \subsubcase Suppose now that $M=M_I(j)$ for some $j \geq 4$, and let us suppose first that there is exactly one LR-row in $M$ and that $v_1$ is the such LR-row.
    Notice first that, if $j$ is odd, then we find $F_2(j)$ in $B$ induced by the vertices represented by every row and column of $M$. Hence, suppose $j$ is even.
    By Remark~\ref{obs:no_Di_in_MI}, there are at most two labeled rows in $M$ and they are labeled with distinct letters.
    If either there are no labeled rows or there is exactly one labeled row, then we find $M_{III}(j)$ induced by $\{ v_1, \ldots, v_j$, $k_{61}, \ldots, k_{6j}$, $ k_l \}$, where $k_l$ is any vertex in some $K_l \neq K_6$, nonadjacent to the labeled row.
    Suppose there are two labeled rows $v_i$ and $v_l$ in $M$. It suffices to see what happens if $v_i$ belongs to $S_{36} \cup S_{46}$ and $v_l$ belongs to either $S_{61}$, $S_{64} \cup S_{65}$ or $S_{62} \cup S_{63}$. If $v_l$ in $S_{61}$, then there is a vertex $k_2$ in $K_2$ nonadjacent to both $v_i$ and $v_l$, and thus we find $M_{III}(j)$ induced by the same vertex set from the previous paragraph.
If instead $v_l$ in $S_{64} \cup S_{65}$, then there are vertices $k_2$ in $K_2$ and $k_4$ in $K_4$ such that $k_4$ is adjacent to both $v_l$ and $v_i$. Hence, if $|l-i|$ is even, then we find an $(l-i)$-sun. If instead $|l-i|$ is odd, then we find a $(l-i)$-sun with center, where the center is given by the LR-vertex $v_1$.
Using a similar argument, if $v_l$ lies in $S_{62} \cup S_{63}$, then we find an even sun or an odd sun with center considering the same vertex set as before plus $s_{24}$.

    Suppose now that $v_1$ and $v_2$ are LR-rows. If $j \geq 4$ is even and every row $v_i$ with $i>2$ is unlabeled (or at most one is a labeled row), then we find $M_{II}(j)$ induced by $\{ v_1, \ldots, v_j$, $k_{61}$, $k_{63}, \ldots, k_{6j}$, $ k_l \}$, where $k_l$ is any vertex in some $K_l \neq K_6$ such that each $v_i$ is nonadjacent to $k_l$ for every $i \geq 3$.
    Moreover, if $j \geq 4$ is odd, then we find $F_1(j)$ induced by $\{ v_1, \ldots, v_j$, $k_{61}$, $k_{63}, \ldots, k_{6j}\}$. The same holds if there is exactly one labeled row since we can always find a vertex in some $K_l$ with $l \neq 6$ that is nonadjacent to such labeled vertex, if necessary.
    Let us suppose there are exactly two labeled rows $v_i$ and $v_l$. By Remark~\ref{obs:no_Di_in_MI}, these rows are non-consecutive and are labeled with distinct letters. As in the previous case, $v_i$ belongs to $S_{36} \cup S_{46}$ and $v_l$ belongs to either $S_{61}$ or $S_{64} \cup S_{65}$. If $v_l$ belongs to $S_{61}$, then there is a vertex $k_2$ in $K_2$ nonadjacent to both $v_i$ and $v_l$, and thus we find $\{ v_1, \ldots, v_j$, $k_{61}$, $k_{63}, \ldots, k_{6j}$, $k_2 \}$. If instead $v_l$ lies in $S_{64} \cup S_{65}$, then we find $k_4$ in $K_4$ adjacent to both $v_i$ and $v_l$ and hence we find either an even sun or an odd sun with center as in the previous case.
    Using a similar argument, if $v_l$ lies in $S_{62} \cup S_{63}$, then we find an even sun if $|l-i|$ is even or an odd sun with center if $|l-i|$ is odd.
    Finally, suppose $v_1$ and $v_i$ are LR-rows, where $i>2$. If $j=4$, then we find a $4$-sun as an induced subgraph, hence, suppose that $j>5$. In that case, we find a tent contained in the subgraph induced by $\{v_1$, $v_2$, $v_3 \}$ if $i=3$ and $\{v_1$, $v_{j-1}$, $v_j \}$ if $i= j-1$.
    Let $3<i<j-1$. However, in that case we find $M_{II}(i)$ induced by $\{v_1$, $v_2, \ldots, v_i$, $k_{62}, \ldots, k_{6(j-2)}, k_{6j} \}$. Therefore, there is no $M_I(j)$ in $B^*_\tagg$.

    \subcase Suppose that $B^*_\tagg$ contains $M=M_{II}(j)$. Let $v_1, \ldots, v_j$ and $k_{61}, \ldots, k_{6j}$ in $K_6$ be the vertices corresponding to the rows and the columns of $M$. If $j$ is odd and there are no labeled rows, then we find $F_1(j)$ by considering $\{ v_1, \ldots, v_j$, $k_{61} \ldots, k_{6(j-1)} \}$. Moreover, if there are no LR-rows and $j$ is odd, then we find $M_{II}(j)$ as an induced subgraph.
    Hence, we assume from now on that there is at least one LR-row.

\begin{remark} \label{obs:reduce_casos_MII_4tent}
    There are at most two rows labeled with L or R in $M=M_{II}(k)$, for any three LR-rows induce an enriched submatrix that contains either $D_0$, $D_1$ or $D_2$. Moreover, since $B$ is admissible, then there are at most three LR-rows.
    If $v_i$ and $v_l$ with $1<i<l<j$ are two rows labeled with either L or R, then they are labeled with distinct letters for if not we find $D_0$. Moreover, they are non-consecutive since we there is no $D_1$ or $D_2$ in $B$. Thus, since $v_i$ belongs to $S_{36} \cup S_{46}$ and $v_l$ belongs to either $S_{61}$ or $S_{64} \cup S_{65}$ or $S_{62} \cup S_{63}$, one of the following holds:
    \begin{itemize}
        \item If $v_l$ in $S_{61}$, then we find a $(l-i+2)$-sun if $l-i$ is even or a $(l-i+2)$-sun with center if $|l-i|$ is odd (the center is $k_{6j}$) induced by $\{ v_i, \ldots, v_l$, $s_{12}$, $s_{24}$, $k_{6(i+1)} \ldots, k_{6l}$, $k_1$, $k_2$, $k_4$, $k_{6j} \}$.
        \item If $v_l$ in $S_{64} \cup S_{65}$ (resp.\ $S_{62} \cup S_{63}$), then we find a $(l-i)$-sun if $|l-i|$ is even or a $(l-i)$-sun with center if $|l-i|$ is odd (the center is $k_{6j}$) induced by $\{ v_i, \ldots, v_l$, $k_{6(i+1)} \ldots, k_{6l}$, $k_4$, $k_{6j} \}$ (resp.\ $k_1$, $k_2$).
    \end{itemize}
    Furthermore, let $v_1$ and $v_i$ labeled with either L or R, where $1<i \leq j$. If $i=2,j$, then they are labeled with distinct letters for if not we find $D_0$. Moreover, they are colored with distinct colors for if not we find $D_1$. If instead $2<i<j$, then they are labeled with the same letter for if not we find $D_1$ or $D_2$.
\end{remark}

By Remark~\ref{obs:reduce_casos_MII_4tent}, we may assume without loss of generality that, if there are rows labeled with either L or R, then these rows are either $v_j$ and $v_{j-1}$, $v_1$ and $v_j$ or $v_{j-2}$ and $v_j$ for every other case is analogous or already covered. Moreover, if $v_j$ and $v_{j-1}$ (resp.\ $v_1$) are labeled rows, then we may assume they are colored with distinct colors.

    \subsubcase There is exactly one LR-row. Suppose first that $v_1$ is the only LR-row. If every non-LR row is unlabeled or $v_{j-2}$ and $v_j$ are labeled rows, then they are labeled with the same letter for if not we find $D_1$ or $D_5$ considering $v_1$, $v_{j-2}$ and $v_j$. Then, we find $M_{III}(j)$ induced by $\{ k_l$, $v_1, \ldots, v_j$, $k_{61}$, $\ldots$, $k_{6j} \}$, where $k_l$ is any vertex in $K_l \neq K_6$.
    Moreover, if $v_{j-1}$ is a labeled row, then we find either a $(j-1)$-sun or a $(j-1)$-sun with center, depending on whether $j$ is even or odd, induced by $\{ v_1, \ldots, v_{j-1}$, $k_l$, $k_{61}, \ldots, k_{6(j-2)}$, $k_{6j} \}$.
If $v_2$ is an LR-row, then we find $M_{II}(j-1)$ or $F_1(j-1)$ --depending on whether $j$ is odd or even-- induced by every column of $B$ and the rows $v_2$ to $v_j$. This holds disregarding on whether there are or not rows labeled with L or R.

Suppose $v_i$ is an LR-row for some $2<i<j-1$. Let $r_i$ be the first column in which $v_i$ has a $0$ and $c_i$ be column in which $v_j$ has a $0$, then we find a tent induced by $\{ k_{61}$, $k_{6(r_i)}$,$k_{6(c_i)}$, $v_1$, $v_i$, $v_j \}$.
If $v_{j-1}$ is an LR-row, then we find $M_{II}(j-1)$ induced by $\{ v_1, \ldots, v_{j-1}$, $k_{61}, \ldots, k_{6(j-2)}$, $k_{6j}\}$.

If $v_j$ is an LR-row and either every other row is unlabeled or there is exactly one labeled row, then we find $M_{III}(j)$ induced by $\{ k_l$, $v_1, \ldots, v_j$, $k_{61}$, $\ldots$, $k_{6j} \}$, where $k_l$ is any vertex in $K_l \neq K_6$ such that the vertex representing the only labeled row is nonadjacent to $k_l$.
If instead there are two labeled rows, then it follows from Remark~\ref{obs:reduce_casos_MII_4tent} that such rows are either $v_1$ and $v_2$ or $v_1$ and $v_i$ for some $2<i<j$. However, if $v_i$ is a labeled row for some $1<i<j-1$, then we find either an even sun or an odd sun with center analogously as we have in Remark~\ref{obs:reduce_casos_MII_4tent}. If instead $v_{j-1}$ and $v_1$ are labeled rows, then they are labeled with the same letter and thus we are in the same situation as if there were no labeled rows in $B$, since we can find a vertex that is nonadjacent to both $v_1$ and $v_{j-1}$.

\subsubcase There are two LR-rows. If $v_1$ and $v_2$ are LR-rows, then we find $M_{II}(j-1)$ as in the case where only $v_2$ is an LR-row.
Suppose $v_1$ and $v_3$ are LR-rows. If $j=4$, then we find $M_{II}(4)$ induced by $\{v_1, \ldots, v_4$, $k_{61}$, $k_{62}$, $k_{64}$, $ k_l \}$ where $k_l$ in $K_l \neq K_6$. Such a vertex exists, since $v_2$ and $v_4$ are either unlabeled rows or are rows labeled with the same letter, for if they were labeled with distinct letters we would find $D_0$ or $D_1$. Thus, there is a vertex that is nonadjacent to both $v_2$ and $v_4$ and is adjacent to $v_1$ and $v_3$.
If $j>4$, then we find a tent induced by $\{ v_3$, $v_{j-1}$, $v_j$, $k_{6(j-2)}$, $k_{6(j-1)}$, $k_{6j} \}$.
Moreover, if $v_i$ is an LR-row for $1<2<j-1$ and $v_{j-1}$ and $v_j$ are non-LR-rows, then we find a tent induced by $\{v_i$, $v_{j-1}$, $v_j$, $k_{6(j-2)}$, $k_{6(j-1)}$, $k_{6j} \}$.

It remains to see what happens if $v_1$ and $v_{j-1}$ and $v_1$ and $v_j$ are LR-rows.
If $v_1$ and $v_{j-1}$ are LR-rows, then we find $M_{II}(j)$ induced by every row of $M$ and every column except for column $j-1$, which is replaced by some vertex $k_l$ in $K_l \neq K_6$. This follows since, if there are two labeled rows, then they must be $v_i$ for some $1<i<j-1$ and $v_j$, hence they are labeled with the same letter and therefore there is a vertex $k_l$ nonadjacent to both.
Finally, if $v_1$ and $v_j$ are LR-rows, then we find a $j$-sun or a $j$-sun with center, depending on whether $j$ is even or odd, contained in the subgraph induced by $\{ v_1, \ldots, v_j$, $k_{61}, \ldots, k_{6j}$, $k_l \}$, where $k_l$ in $K_l \neq K_6$ is nonadjacent to every non-LR row. 
Therefore, there is no $M=M_{II}(j)$ in $B^*_\tagg$.

    \subcase Suppose that $B$ contains $M= M_{III}(j)$, let $v_1, \ldots v_j$ and $k_{61}, \ldots, k_{6(j+1)}$ be the rows and the columns of $M$.
    If there are no LR-rows, then we find $M_{III}(j)$, hence we assume there is at least one LR-row. It follows analogously as in the previous cases that there are at most two LR-rows in $M$, since $B$ is admissible.
    Notice that every pair of labeled rows $v_i$ and $v_l$ with $1\leq 1<i,l<j-1$ cannot be labeled with the same letter, since they induce $D_0$. Once more, if such rows are labeled with distinct letters, then they are non-consecutive, for if not we find $D_1$ or $D_2$. Furthermore, if such $v_i$ and $v_l$ are labeled rows, then we find either an even sun or an odd sun with center.
    It follows using the same arguments that, if $i=1,j-1$ and $l=j$, then $v_i$ and $v_l$ are not both labeled rows.
    Hence, if there are two labeled rows, then such rows must be $v_j$ and $v_i$ for some $i$ such that $2<i<j-1$.

\subsubcase There is exactly one LR-row. 
If $v_i$ is an LR-row for some $1\leq i< j-1$, then we find $M_{II}(j-i+1)$ induced by $\{ v_i, \ldots, v_j$, $k_{6(i+1)}, \ldots, k_{6(j+1)} \}$.
If $v_{j-1}$ is an LR-row, then we find $M_{II}(j)$, induced by $\{ v_1, \ldots, v_j$, $k_{62}, \ldots,  k_{6(j-1)}$, $k_{6(j+1)} \}$.
If instead $v_j$ is an LR-row, then we find an even $j$-sun or an odd $j$-sun with center $k_{6(j+1)}$.

\subsubcase There are two LR-rows $v_i$ and $v_l$. If $1\leq i<l<j-1$ and $v_i$ and $v_l$ are non-consecutive rows, then we find a tent induced by $v_i$, $v_l$, $v_j$, $k_s$ in $K_s \neq K_6$ adjacent to both $v_i$ and $v_l$ and nonadjacent to $v_j$, $k_{6i}$ (or $k_{6(i+1)}$ if $i=1$) and $k_{6l}$ (or $k_{6(l+1)}$ if $l=j-1$). The same holds if $l=i+1$ and $i>1$.
If instead $i=1$, or $i=j-1$ and $l=i+1$, then we find $F_0$ (or $M_{III}(3)$ if $j=3)$ induced by $\{ v_i$, $v_{i+1}$, $k_{6i}$, $k_{6(i+1)}$, $k_{6(i+2)}$, $k_6(j+1)$, $ k_s\}$ with $k_s$ in $K_s \neq K_6$ adjacent to both $v_i$ and $v_{i+1}$.
Finally, if $v_1$ and $v_j$ are LR-rows, then we find $M_{III}(j)$ induced by $\{ v_1$, $\ldots$, $v_j$, $k_{61}$, $\ldots$, $k_6(j+1)\}$. If instead $v_i$ and $v_j$ are LR-rows for some $i>1$, then we find $M_V$ induced by $\{ v_i$, $v_j$, $v_1$, $v_{j-1}$, $k_{61}$, $k_{62}$, $k_{6i}$, $k_{6(i+1)}$, $k_{6j}\}$, therefore there is no $M_{III}(j)$ in $B^*_\tagg$.

    \subcase Suppose $B^*_\tagg$ contains $M=M_{IV}$, let $v_1, \ldots, v_4$ and $k_{61}, \ldots, k_{66}$ be the rows and the columns of $M$. If there are no labeled rows, then we find $M_{IV}$ as an induced subgraph, and since $B$ is admissible and any three rows are not pairwise nested, then there are at most two LR-rows. 
    If $v_i$ is an LR-row for $i=1,2,3$, then we find $M_V$ induced by $\{ v_2$, $v_3$, $v_4$, $k_{62}, \ldots, k_{66} \}$. Moreover, if only $v_4$ is an LR-row, then we find $M_{IV}$ induced by all the rows and columns of $M$. Thus, we assume there are exactly two LR-rows.
    If $v_1$ and $v_4$ are LR-rows, then we find $M_V$ induced by $\{ v_1$, $v_2$, $v_3$, $v_4$, $k_{61}$, $k_{63}, \ldots, k_{66} \}$. The same holds if $v_i$ and $v_4$ are LR-rows, with $i=2,3$.
    Finally, if $v_1$ and $v_2$ are LR-rows, then we find a tent induced by $\{ v_1$, $v_2$, $v_4$, $k_{62}$, $k_{64}$, $k_{65} \}$. It follows analogously by symmetry if $v_1$ and $v_3$ or $v_2$ and $v_3$ are LR-rows, therefore there is no $M_{IV}$ in $B^*_\tagg$.

    \subcase Suppose $B^*_\tagg$ contains $M=M_V$, let $v_1, \ldots, v_4$ and $k_{61}, \ldots, k_{65}$ be the rows and the columns of $M$.
    Once more, if there are no LR-rows, then we find $M_V$ as an induced subgraph, thus we assume there is at least one LR-row. Moreover, since any three rows are not pairwise nested, there are at most two LR-rows.

    \subsubcase There is exactly one LR-row. If $v_1$ is the only LR-row, then we find a tent induced by $\{ v_1$, $v_3$, $v_4$, $k_{61}$, $k_{63}$, $k_{65} \}$. The same holds if $v_2$ is the only LR-row.
    If $v_3$ is the only LR-row and every other row is unlabeled (or are all labeled with the same letter), then we find $M_{IV}$ induced by$\{ v_1$, $v_2$, $v_3$, $v_4$, $k_{61}, \ldots, k_{65}$, $k_l \}$ where $k_l$ in $K_l \neq K_6$ adjacent only to $v_3$.
    Suppose there are at least two rows labeled with either L or R. Notice that, if $v_1$ and $v_2$ are labeled, then they are labeled with distinct letters for if we find $D_0$ in $B$. Moreover, $v_1$ (resp.\ $v_2$) and $v_4$ are not both labeled, for in that case we find either $D_0$, $D_1$ or $D_2$ in $B$. Hence, there are at most two rows labeled with either L or R, and they are necessarily $v_1$ and $v_2$.
    In that case, there is a vertex $k_l$ in some $K_l \neq K_6$ such that $v_2$ and $v_3$ are adjacent to $k_l$ and $v_4$ is nonadjacent to $k_l$. We find a tent induced by $v_2$, $v_3$, $v_4$, $k_l$, $k_{64}$ and $k_{65}$.
    If $v_4$ is the only LR-row and every other row is unlabeled or are (one, two or) all labeled with the same letter, then we find $M_V$ induced by $\{ v_1$, $v_2$, $v_3$, $v_4$, $k_{61}, \ldots, k_{64}$, $k_l \}$ where $k_l$ in $K_l \neq K_6$ adjacent only to $v_4$.

    \subsubcase There are exactly two LR-rows.
    If $v_1$ and $v_2$ are LR-rows, then we find a tent induced by $\{ v_1$, $v_2$, $v_3$, $k_{62}$, $k_{63}$, $k_{65} \}$.
    If instead $v_1$ and $v_3$ are LR-rows and every other row is unlabeled or (one or) all are labeled with the same letter, then we find $M_V$ induced by every row and column plus a vertex $k_l$ in some $K_l \neq K_6$ such that both $v_2$ and $v_4$ are nonadjacent to $k_l$.
    Moreover, since $v_2$ and $v_4$ overlap and there is a column in which both rows have a $0$, then they are not labeled with distinct letters --disregarding of the coloring-- for in that case we find $D_1$ or $D_2$ in $B$.
    If $v_1$ and $v_4$ are LR-rows and every other row is unlabeled or are (one or) all labeled with the same letter, then we find a tent induced by every row and column plus a vertex $k_l$ in some $K_l \neq K_6$ such that both $v_2$ and $v_4$ are nonadjacent to $k_l$. Notice that $v_2$ and $v_3$ are labeled with the same letter for if not we find either $D_1$ or $D_2$ in $B$.
    If $v_3$ and $v_4$ are LR-rows and every other row is unlabeled or $v_1$ (resp.\ $v_2$) is labeled with L or R, then we find $M_{IV}$ induced by every row and column plus a vertex $k_l$ in some $K_l \neq K_6$ such that both $v_1$ and $v_2$ are nonadjacent to $k_l$.
    Notice that $v_1$ and $v_2$ are labeled with distinct letters for if not they induce $D_0$. However, they cannot be labeled with distinct letters since in that case we find either $D_1$ or $D_2$.

    \subsubcase There are exactly three LR-rows.
    If $v_1$, $v_2$ and $v_3$ are LR-rows, since there is a vertex $k_l \in K_l$ with $l \neq 6$ such that $v_4$ is nonadjacent to $k_l$, then we find a tent induced by $\{ v_1$, $v_2$, $v_4$, $k_{61}$, $k_{64}$, $k_l \}$. Analogously, if $v_1$, $v_2$ and $v_4$ are LR-rows and $v_3$ is not, then the tent is induced by $\{ v_1$, $v_2$, $v_3$, $k_{61}$, $k_{64}$, $k_{65} \}$. The same holds if all $4$ rows are LR-rows, where the tent is induced by $\{ v_1$, $v_2$, $v_4$, $k_{62}$, $k_{63}$, $k_{65} \}$.
    Finally, if $v_2$, $v_3$ and $v_4$ are LR-rows, since there is a vertex $k_l \in K_l$ with $l \neq 6$ such that $v_1$ is nonadjacent to $k_l$, then we find $M_V$ induced by $\{ v_1$, $v_2$, $v_3$, $v_4$, $k_{61}$, $k_{62}$, $k_{63}$, $k_{65}$, $k_l \}$.

    \case $B$ is admissible and $B^*_\tagg$ is $\mathcal{M}$-free, but $B$ contains a monochromatic gem, or a monochromatic weak-gem or a badly-colored doubly weak-gem. Since there are no uncolored labeled rows and those colored rows are labeled with either L or R and do not induce any of the matrices in $\mathcal{D}$, then in particular no two pre-colored rows of $B$ induce a monochromatic gem or a monochromatic weak gem, and there are no badly-colored gems since every LR-row is uncolored, therefore this case is not possible.

    \case Finally, let us suppose that $B$ and $B^*_\tagg$ contain none of the matrices listed in Theorem~\ref{teo:2-nested_caract_bymatrices}, but $B$ is not $2$-nested. We consider $B$ ordered according to a suitable LR-ordering. Let $B'$ be a matrix obtained from $B$ by extending its partial pre-coloring to a total $2$-coloring. It follows from Lemma~\ref{lema:B_ext_2-nested} that, if $B'$ is not $2$-nested, then either there is an LR-row for which its L-block and R-block are colored with the same color, or $B'$ contains a monochromatic gem or a monochromatic weak gem or a badly-colored doubly weak gem.
    If $B'$ contains a monochromatic gem where the rows that induce such a gem are not LR-rows, then the proof is analogous as in the tent case. Thus, we may assume that at least one of the rows of the monochromatic gem is an LR-row.

    \subcase \textit{Let us suppose first that there is an LR-row $w$ for which its L-block $w_L$ and R-block $w_R$ are colored with the same color. } If these two blocks are colored with the same color, then there is either one odd sequence of rows $v_1, \ldots, v_j$ that force the same color on each block, or two distinct sequences, one that forces the same color on each block.

    \subsubcase Suppose first that there is one odd sequence $v_1, \ldots, v_j$ that forces the color on both blocks. If $k=1$, then notice this is not possible since we are coloring $B'$ using a suitable LR-ordering. If there is not a suitable LR-ordering, then it follows from Lemma~\ref{teo:hay_suitable_ordering} that $B$ is either not admissible or not LR-orderable, which is a contradiction. Thus, let $j>1$ and assume without loss of generality that $v_1$ intersects $w_L$ and $v_j$ intersects $w_R$. Moreover, we assume that each of the rows in the sequence $v_1, \ldots, v_j$ is colored with a distinct color and forces the coloring on the previous and the next row of the sequence.
    If $v_1, \ldots, v_j$ are all unlabeled rows, then we find an even $(j+1)$-sun.
    If instead $v_1$ is an L-row, then $w_L$ is properly contained in $v_1$. Thus, $v_2, \ldots, v_{j-1}$ are not contained in $v_1$, since at least $v_j$ intersects $w_R$. If $v_j$ is unlabeled or labeled with R, then we find an even $(j+1)$-sun. If instead $v_j$ is labeled with L, since $j$ is odd, then we find $S_1(j+1)$ in $B$ which is not possible since we are assuming $B$ admissible.

    \subsubcase Suppose there are two independent sequences $v_1, \ldots, v_j$ and $x_1, \ldots, x_l$ that force the same color on $w_L$ and $w_R$, respectively. Suppose without loss of generality that $w_L$ and $w_R$ are colored with red.
    If $j=1$ and $l=1$, then we find $D_6$, which is not possible. Hence, we assume that either $j>1$ or $l>1$.
    Suppose that $j>1$ and $l>1$. In this case, there is a labeled row in each sequence, for if not we can change the coloring for each row in one of the sequences and thus each block of $w$ can be colored with distinct colors. We may assume that $v_j$ is labeled with L and $x_l$ is labeled with R (for the LR-ordering used to color $B'$ is suitable and thus there is no R-row intersecting $w_L$, and the same holds for each L-block and $w_R$). As in the previous paragraphs, we assume that each row in each sequence forces the coloring on both the previous and the next row in its sequence. In that case, $v_2, \ldots, v_j$  is contained in $w_L$ and $x_2, \ldots, x_l$ is contained in $w_R$.
    Moreover, $w$ represents a vertex in $S_{[16}$, $v_j$ lies in $S_{46} \cup S_{36}$ or $S_{16} \cup S_{26} \cup S_{56}$ and $x_l$ lies in $S_{61} \cup S_{64} \cup S_{65}$ or $S_{62} \cup S_{63}$ (depending on whether they are colored with red or blue, respectively).
    Suppose first that they are both colored with red, thus $v_j$ lies in $S_{46} \cup S_{36}$ and $x_l$ lies in $S_{61} \cup S_{64} \cup S_{65}$. In this case $j$ and $l$ are both even.
    If $x_l$ lies in $S_{64} \cup S_{65}$, since there is a $k_i$ in some $K_i \neq K_6$ adjacent to both $v_j$ and $x_l$, then we find $F_2(j+l+1)$ contained in the subgraph induced by $k_i$ and each row and column on which the rows in $w$ and both sequences are not null. If instead $x_l$ lies in $S_{61}$, we find $F_2(k+l+3)$ contained in the same submatrix but adding three vertices $k_i$ in $K_i$ for $i=1,2,4$.
    The same holds if $v_j$ and $x_l$ are both blue.
    Suppose now that $v_j$ is colored with red and $v_l$ is colored with blue. Thus, $j$ is even and $l$ is odd. In this case, we find $F_2(j+l+2)$ contained in the submatrix induced by the row that represents $s_{24}$, two columns representing any two vertices in $K_2$ and $K_4$ and each row and column on which the rows in $w$ and both sequences are not null.
    The proof is analogous if either $j=1$ or $l=1$.

    Hence, we may assume there is either a monochromatic weak gem in which one of the rows is an LR-row or a badly-colored doubly-weak gem in $B'$, for the case of a monochromatic gem or a monochromatic weak gem where one of the rows is an L-row (resp.\ R-row) and the other is unlabeled is analogous to the tent case.  

    \subcase \textit{Let us suppose there is a monochromatic weak gem in $B'$}, and let $v_1$ and $v_2$ be the rows that induce such gem, where $v_2$ is an LR-row.
    \subsubcase Suppose first that $v_1$ is a pre-colored row. Suppose without loss of generality that the monochromatic weak gem is induced by $v_1$ and the L-block of $v_2$ and that $v_1$ and $v_2$ are both colored with red. We denote by $v_{2L}$ the L-block of $v_2$.
    If $v_1$ is labeled with R, then $v_2$ is the L-block of some LR-row $r$ in $B$ and $v_1$ is the R-block of itself. However, since the LR-ordering we are considering to color $B'$ is suitable, then the L-block of an LR-row has empty intersection with the R-block of a non-LR row and thus this case is not possible.
    If $v_1$ is labeled with L, since both rows induce a weak gem, then $v_{2L}$ is properly contained in $v_1$.
    Since $v_1$ is a row labeled with L in $B$, then $v_1$ is a pre-colored row. Moreover, since $v_{2L}$ is colored with the same color as $v_1$, then there is either a blue pre-colored row, or a sequence of rows $v_3, \ldots, v_j$ where $v_j$ forces the red coloring of $v_{2L}$.
    In either case, there is a pre-colored row in that sequence that forces the color on $v_{2L}$, and such row is either labeled with L or with R.
    Suppose first that such row is labeled with L. If $v_3$ is a the blue pre-colored row that forces the red coloring on $v_{2L}$, then $v_{2L}$ is properly contained in $v_3$. However, in that case we find $D_4$ which is not possible since $B$ is admissible.
    Hence, we assume $v_3, \ldots, v_{j-1}$  is a sequence of unlabeled rows and that $v_j$ is a labeled row such that this sequence forces $v_{2L}$ to be colored with red, and each row of the sequence forces the color on both its predecessor and its successor.
    If $j-3$ is even, then $v_j$ is colored with blue, and if $j-3$ is odd then $v_j$ is colored with red. In either case, we find $S_5(j)$ contained in the submatrix induced by rows $v_1, v_2, v_3, \ldots, v_j$.
    If instead the row $v$ that forces the coloring on $v_{2L}$ is labeled with R, since the LR-ordering used to color $B$ is suitable, then the intersection between $v_{2L}$ and $v$ is empty. Hence, $v \neq v_3$, thus we assume that $v_3, \ldots, v_{j-1}$ are unlabeled rows and $v_j = v$. If $j-3$ is odd, then $v_j$ is colored with red, and if $j-3$ is even, then $v_j$ is colored with blue. In either case we also find $S_5(j)$, which is not possible since $B$ is admissible.

    \subsubcase Suppose now that $v_1$ is an unlabeled row. Notice that, since $v_1$ and $v_2$ induce a weak gem, then $v_1$ is not nested in $v_2$.
    Hence, either the coloring of both rows is forced by the same sequence of rows or the coloring of $v_1$ and $v_2$ is forced for each by a distinct sequence of rows. As in the previous cases, we assume that the last row of each sequence represents a pre-colored labeled row.
    Suppose first that both rows are forced to be colored with red by the same row $v_3$. Thus, $v_3$ is a labeled row pre-colored with blue. Moreover, since $v_3$ forces $v_1$ to be colored with red, then $v_1$ is not contained in $v_3$ and thus there is a column $k_{61}$ in which $v_1$ has a $1$ and $v_3$ has a $0$.
We may also assume that $v_2$ has a $0$ in such a column since $v_1$ is also not contained in $v_2$.
Moreover, since $v_3$ forces $v_2$ to be colored with red, then $v_3$ is labeled with the same letter than $v_2$ and $v_3$ is not contained in $v_2$, thus we can find a column $k_{62}$ in which $v_2$ has a $0$ and $v_1$ and $v_3$ both have a $1$.
Furthermore, since $v_3$ and $v_2$ are both labeled with the same letter and the three rows have pairwise nonempty intersection, then there is a column $k_{63}$ in which all three rows have a $1$. Since $v_3$ is a row labeled with either L or R in $B$, then there are vertices $k_l \in K_l$, $k_m \in K_m$ with $l \neq m$, $l,m \neq 6$ such that $v_3$ is adjacent to $k_l$ and nonadjacent to $k_m$. Moreover, since $v_2$ is an LR-row, then $v_2$ is adjacent to both $k_l$ and $k_m$ and $v_j$ is nonadjacent to $k_l$ and $k_m$.
Hence, we find $F_0$ induced by $\{v_3$, $v_1$, $v_2$, $k_l$, $k_{61}$, $k_{63}$, $k_{62}$, $k_m \}$.
Suppose instead there is a sequence of rows $v_3, \ldots, v_j$ that force the coloring of both $v_1$ and $v_2$, where $v_3, \ldots, v_{j-1}$ are unlabeled rows and $v_j$ is labeled with either L or R and is pre-colored.
We have two possibilities: either $v_j$ is labeled with L or with R.
If $v_j$ is labeled with L and $v_j$ forces the coloring of $v_2$, then we have the same situation as in the previous case. Thus we assume $v_j$ is nested in $v_2$. In this case, since $v_j$ and $v_2$ are labeled with L, the vertices $v_3, \ldots, v_{j-1}$ are nested in $v_2$ and thus they are chained from right to left. Moreover, since $v_1$ and $v_2$ are colored with the same color, then there is an odd index $1 \leq l \leq j-1$ such that $v_1$ contains $v_3, \ldots, v_l$ and does not contain $v_{l+1}, \ldots, v_j$. Hence, we find $F_1(l+1)$ considering the rows $v_1, v_2, \ldots, v_{l+1}$.
Suppose now that $v_j$ is labeled with R. Since $B'$ is colored using a suitable LR-ordering, then $v_j$ and $v_2$ have empty intersection, thus there is a sequence of unlabeled rows $v_3, \ldots, v_{j-1}$, chained from left to right. Notice that it is possible that $v_1= v_3$.
Suppose first that $j$ is even. If $v_1 = v_3$, then there is an odd number of unlabeled rows between $v_1$ and $v_j$. In this case we find a $(j-2)$-sun contained in the subgraph induced by rows $v_2, v_1=v_3, v_4, \ldots, v_j$.
If instead $v_1 \neq v_3$, then $v_1$ and $v_3$ and $v_1$ and $v_5$ both induce a $0$-gem, and thus we find a $(j-2)$-sun in the same subgraph.
If $j$ is odd, then there is an even number of unlabeled rows between $v_2$ and $v_j$. Once more, we find a $(j-1)$-sun contained in the subgraph induced by rows $v_2, v_3, \ldots, v_j$.

Notice that these are all the possible cases for a weak gem. This follows from the fact that, if there is a pre-colored labeled row that forces the coloring upon $v_1$ then it forces the coloring upon $v_2$ and viceversa. Moreover, if there is a sequence of rows that force the coloring upon $v_2$, then one of these rows of the sequence also forces the coloring upon $v_1$, and viceversa.
Furthermore, since the label of the pre-colored row of the sequence determines a unique direction in which the rows overlap in chain, then there is only one possibility in each case, as we have seen in the previous paragraphs.
It follows that the case in which there is a sequence forcing the coloring upon each $v_1$ and $v_2$ can be reduced to the previous case.

    \subcase \textit{Suppose there is a badly-colored doubly-weak gem in $B'$. }
    Let $v_1$ and $v_2$ be the LR-rows that induce the doubly-weak gem. Since the suitable LR-ordering determines the blocks of each LR-row, then the L-block of $v_1$ properly contains the L-block of $v_2$ and the R-block of $v_1$ is properly contained in the R-block of $v_2$, or viceversa. Moreover, the R-block of $v_1$ may be empty. Let us denote by $v_{1L}$ and $v_{2L}$ (resp.\ $v_{1R}$ and $v_{2R}$) the L-blocks (resp.\ R-blocks) of $v_1$ and $v_2$.
    There is a sequence of rows that forces the coloring on both LR-rows simultaneously or there are two sequences of rows and each forces the coloring upon the blocks of $v_1$ and $v_2$, respectively.
    Whenever we consider a sequence of rows that forces the coloring upon the blocks of $v_1$ and $v_2$, we will consider a sequence in which every row forces the coloring upon its predecessor and its successor, a pre-colored row is either the first or the last row of the sequence, the first row of the sequence forces the coloring upon the corresponding block of $v_1$ and the last row forces the coloring upon the corresponding block of $v_2$. It follows that, in such a sequence, every pair of consecutive unlabeled rows overlap.
    We can also assume that there are no blocks corresponding to LR-rows in such a sequence, for we can reduce this to one of the cases.
        Suppose first there is a sequence of rows $v_3, \ldots, v_j$ that forces the coloring upon both LR-rows simultaneously. We assume that $v_3$ intersects $v_1$ and $v_j$ intersects $v_2$.
    If $v_3, \ldots, v_j$ forces the coloring on both L-blocks, then we have four cases: (1) either $v_3, \ldots, v_j$ are all unlabeled rows, (2) $v_3$ is the only pre-colored row, (3) $v_j$ is the only pre-colored row or (4) $v_3$ and $v_j$ are the only pre-colored rows.
    In either case, if $v_3, \ldots, v_j$ is a minimal sequence that forces the same color upon both $v_{1L}$ and $v_{2L}$, then $j$ is odd.

    \subsubcase Suppose $v_3, \ldots, v_j$ are unlabeled.
    If $j=3$, then we find $S_7(3)$ contained in the submatrix induced by $v_1$, $v_2$ and $v_3$. Suppose $j>3$, thus we have two possibilities. If $v_2 \cap v_3 \neq \emptyset$, since $j$ is odd, then we find a $(j-1)$-sun contained in the submatrix induced by considering all the rows $v_1$, $v_2$, $v_3, \ldots, v_j$.
    If instead $v_2 \cap v_3 = \emptyset$, then we find $F_2(j)$ contained in the same submatrix.

    \subsubcase Suppose $v_3$ is the only pre-colored row.
    Since $v_3$ is a pre-colored row and forces the color red upon the L-block of $v_1$, then $v_3$ contains $v_{1L}$ and $v_3$ is colored with blue. If $v_4 \cap v_{1L} \neq \emptyset$, then we find $F_0$ in the submatrix given by considering the rows $v_1$, $v_3$, $v_4$, having in mind that there is a column representing some $k_i$ in $K_i \neq K_6$ in which the row corresponding to $v_1$ has a $1$ and the rows corresponding to $v_3$ and $v_4$ both have $0$. This follows since $v_4$ is unlabeled and thus represents a vertex that lies in $S_{66}$, and $v_3$ is pre-colored and labeled with L or R and, thus it represents a vertex that is not adjacent to every partition $K_i$ of $K$.
    If instead $v_4 \cap v_{1L} = \emptyset$, then we find $F_2(j-2)$ contained in the submatrix induced by the rows $v_1, v_2, \ldots, v_{j-2}$ if $v_2 \cap v_{2R} = \emptyset$, and induced by the rows $v_1, v_2, v_5, \ldots, v_j$ if $v_2 \cap v_{2R} \neq \emptyset$.

    \subsubcase \label{subsubcase:B6_4.3.3.} Suppose $v_j$ is the only pre-colored row.
    In this case, $v_j$ properly contains $v_{2L}$ and we can assume that the rows $v_4, \ldots, v_{j-1}$ are contained in $v_{1L}$. If $v_3 \cap v_2 \neq \emptyset$, then we find an even $(j-1)$-sun in the submatrix induced by the rows $v_2, v_3, \ldots, v_j$.
    If instead $v_3 \cap v_2 = \emptyset$, then we find $F_2(j)$ in the submatrix given by rows $v_1, \ldots, v_j$.

    \subsubcase Suppose that $v_3$ and $v_j$ are the only pre-colored rows. Thus, we can assume that $v_j$ properly contains $v_{2L}$ and $v_3$ properly contains $v_{2L}$, thus $v_3$ properly contains $v_{2L}$. Hence, we find $D_9$ induced by the rows $v_1$, $v_2$ and $v_3$ which is not possible since $B$ is admissible.
    The only case we have left is when $v_3, \ldots, v_j$ forces the coloring upon $v_{1L}$ and $v_{2R}$. This follows from the fact that, if $v_3, \ldots, v_j$ forces the color upon $v_{2L}$ and $v_{1R} \neq \emptyset$, then this case can be reduced to case (4.3.\ref{subsubcase:B6_4.3.3.}).
     Hence, either (1) $v_3, \ldots, v_j$ are unlabeled rows, (2) $v_3$ is the only pre-colored row, or (3) $v_3$ and $v_j$ are the only pre-colored rows.
    Notice that in either case, $j$ is even and thus for (1) we find $S_8(j)$, which results in a contradiction since $B$ is admissible.
    Moreover, in the remaining cases, $v_3$ properly contains $v_{1L}$ and $v_{2L}$. Since $v_1$ and $v_2$ overlap, we find $D_9$ which is not possible for $B$ is admissible.
\end{mycases}
This finishes the proof.
\end{proof}

Let $G= (K,S)$, $H$ as in Section~\ref{subsec:4tent_partition} and the matrices $\mathbb B_i$ for $i= \{1 \ldots, 6\}$ as defined in the previous subsection.
Suppose $\mathbb B_i$ is $2$-nested for each $i \in \{1, 2, \ldots, 6\}$, let $\chi_i$ be a total block bi-coloring and $\Pi_i$ a suitable LR-ordering for $\mathbb B_i$, for each $i \in \{1, 2, \ldots, 6\}$.

Let $\Pi$ be the ordering of the vertices of $K$ given by concatenating the orderings $\Pi_1$, $\Pi_2$, $\ldots$, $\Pi_6$, as defined in Section~\ref{subsubsec:tent3}. 

\begin{defn} \label{def:matrices_B_por_colores}
We define the $(0,1)$-matrices $\mathbb B_r$, $\mathbb{B}_b$, $\mathbb B_{r-b}$ and $\mathbb B_{b-r}$ as in the tent case, but considering only those vertices of $S$ that are not in $S_{[16}$.
\end{defn}

Notice that the only nonempty subsets $S_{ij}$ with $i>j$ that we are considering are those where $i=6$. Hence, the rows of $\mathbb B_{r-b}$ are those representing vertices in $S_{61} \cup S_{64} \cup S_{65}$ and the rows of $\mathbb B_{b-r}$ are those representing vertices in $S_{62} \cup S_{63}$.

\begin{lema} \label{lema:matrices_union_son_nested_4tent}
    Suppose $\mathbb B_i$ is $2$-nested for each $i =1,2 \ldots, 6$. If $\mathbb B_r$, $\mathbb B_b$, $\mathbb B_{r-b}$ or $\mathbb B_{b-r}$ are not nested, then $G$ contains $F_0$, $F_1(5)$ or $F_2(5)$ as induced subgraphs.
\end{lema}

\begin{proof}
    Suppose first that $\mathbb B_r$ is not nested, thus there is a $0$-gem. Every row in $\mathbb B_r$ represents a vertex that belongs to one of the following subsets of $S$: $S_{12}$, $S_{13}$, $S_{35}$, $S_{36}$, $S_{45}$, $S_{61}$, $S_{62}$, $S_{63}$, $S_{64}$ or $S_{65}$. Recall that $S_{13}=S_{[13}$, $S_{35} = S_{35]}$, $S_{62} = S_{62]}$ and $S_{64}= S_{64]}$
    Let $f_1$ and $f_2$ be two rows that induce a gem in $\mathbb B_r$ and $v_1$ in $S_{ij}$ and $v_2$ in $S_{lm}$ be the corresponding vertices in $G$.
    Analogously as in the tent case we obtain the following claim, strongly using that no row in $\mathbb B_r$, $\mathbb B_b$, $\mathbb B_{r-b}$ and $\mathbb B_{b-r}$ represents a vertex of $S$ complete to $K$, for we do not consider $S_{[16]}$ to define them.
    \begin{claim}    \label{claim:4tent_nogeminnestedmatrices}
        Let $v_1$ and $v_2$ be two vertices that induce a $0$-gem in $\mathbb B_r$, $\mathbb B_b$, $\mathbb B_{r-b}$ or $\mathbb B_{b-r}$. If $v_1$ and $v_2$ lie in the same $S_{ij}$, then $v_1$ and $v_2$ are nested.
    \end{claim}
    Moreover, since $\mathbb B_i$ is $2$-nested for every $i\in \{1, 2, \ldots, 5,6\}$, in particular there are no monochromatic gems in each $\mathbb B_i$. Moreover, if $j=l$, then we find $D_1$ in $K_i$ or $K_j$, respectively.
     It follows from these remarks that we have only two possibilities for the gem: (1) $v_1$ in $S_{36}$ and $v_2$ in $S_{35]} \cup S_{[46}$ and (2) $v_1$ in $S_{62]} \cup S_{63}$ and $v_2$ in $S_{[13}$.
    \begin{mycases}
    \case $v_1$ in $S_{36}$ and $v_2$ in $S_{35]} \cup S_{[46}$.
    Suppose that $v_2$ in $S_{35}$. Let $k_2$ in $K_2$ nonadjacent to both. There are vertices $k_{31}$, $k_{32}$ in $K_3$ such that $k_{31}$ is adjacent only to $v_2$ and $k_{32}$ is adjacent to both. Moreover, there are vertices $k_5$ in $K_5$ and $k_6$ in $K_6$ such that $k_5$ is adjacent to both and $k_6$ is adjacent only to $v_1$. We find $F_0$ induced by $\{ v_1$, $v_2$, $s_{24}$, $k_{31}$, $k_{32}$, $k_{5}$, $k_{6}$, $k_2 \}$.
     If instead $v_2$ in $S_{46}$, then we also find $F_0$ changing $k_{32}$ for some vertex $k_4$ in $K_4$ in the same subset.

    \case$v_1$ in $S_{62]} \cup S_{63}$ and $v_2$ in $S_{[13}$.
    Since $v_2$ is also complete to $K_1$, then one of the columns of the $0$-gem is induced by the column $c_L$ of $\mathbb B_r$. Thus, there is a vertex $k_6$ in $K_6$ adjacent to $v_1$ and nonadjacent to $v_1$. Moreover, the gem is induced by a column corresponding to a vertex $k_3$ in $K_3$ nonadjacent to $v_1$ and adjacent to $v_2$. We find $F_0$ induced by $\{ v_1$, $v_2$, $s_{24}$, $k_{6}$, $k_{1}$, $k_{2}$, $k_{3}$, $k_{4} \}$.
    \end{mycases}
Hence $\mathbb B_r$ is nested.
Suppose now that $\mathbb B_b$ is not nested, and let $v_1$ in $S_{ij}$ and $v_2$ in $S_{lm}$ two vertices for which its rows in $\mathbb B_b$ induce a $0$-gem.
Every row in $\mathbb B_b$ represents a vertex belonging to either $S_{23}$, $S_{24}$, $S_{34}$, $S_{14}$, $S_{25}$, $S_{15}$, $S_{56}$, $S_{16}$, $S_{61}$, $S_{62}$, $S_{63}$, $S_{64}$ or $S_{65}$.
Recall that $S_{14}=S_{14]}$, $S_{25} = S_{[25}$, $S_{26} = S_{[26}$, $S_{62} = S_{62]}$ and $S_{64}= S_{64]}$. It follows from this and Claim~\ref{claim:4tent_nogeminnestedmatrices} that the cases are: (1) $v_1$ in $S_{23} \cup S_{34}$ and $v_2$ in $S_{24}$, (2) $v_1$ in $S_{14} \cup S_{15}$ and $v_2$ in $S_{15} \cup S_{16}$, or $v_2$ in $S_{25} \cup S_{26}$, or $v_2$ in $S_{64} \cup S_{65}$, (3) $v_1$ in $S_{56}$ and $v_2$ in $S_{26} \cup S_{16}$ and (4) $v_1$ in $S_{16}$ and $v_2$ in $S_{64} \cup S_{65}$.

    \begin{mycases}
        \case $v_1$ in $S_{23} \cup S_{34}$ and $v_2$ in $S_{24}$. Suppose $v_1$ in $S_{23}$. We find $F_2(5)$ induced by $\{ v_1$, $v_2$, $s_{12}$, $s_{24}$, $s_{45}$, $k_{1}$, $k_{2}$, $k_3$, $k_{4}$, $k_{5} \}$, where $v_1$ and $v_2$ are nonadjacent to $k_1$ and $k_5$. It follows analogously by symmetry if $v_1$ in $S_{34}$.

        \case $v_1$ in $S_{14} \cup S_{15}$.
        \subcase $v_2$ in $S_{15} \cup S_{16}$
        If $v_1$ in $S_{14}$ and $v_2$ in $S_{15}$, then we find $F_1(5)$ induced by $\{ v_1$, $v_2$, $s_{12}$, $s_{24}$, $s_{45}$, $k_{1}$, $k_{2}$, $k_{4}$, $k_{5} \}$. The same holds if instead $v_2$ in $S_{16}$. Moreover, we find $F_1(5)$ induced by the same subset if $v_1$ in $S_{15}$ and $v_2$ in $S_{16}$, since there is a vertex in $K_5$ that is nonadjacent to $v_1$.

        \subcase $v_2$ in $S_{25} \cup S_{26}$. We find $F_1(5)$ induced by $\{ v_1$, $v_2$, $s_{12}$, $s_{24}$, $s_{45}$, $k_1$, $k_{2}$, $k_{4}$, $k_{5} \}$.

        \subcase $v_2$ in $S_{64} \cup S_{65}$.
         Since $S_{64}$ is complete to $K_1$ and $K_4$, the $0$-gem cannot be induced by $v_2$ and a vertex $v_1$ in $S_{14}$ since $v_1$ is also complete to $K_4$. If $v_1$ in $S_{15}$, then we find $F_0$ induced by $\{ v_1$, $v_2$, $s_{45}$, $k_1$, $k_{2}$, $k_{41}$, $k_{42}$, $k_{5} \}$.

        \case $v_1$ in $S_{56}$ and $v_2$ in $S_{26} \cup S_{16}$. Since $v_2$ is not complete to $K_1$ by definition, we find $F_2(5)$ induced by $\{ v_1$, $v_2$, $s_{12}$, $s_{24}$, $s_{45}$, $k_{1}$, $k_{2}$, $k_4$, $k_{5}$, $k_{6} \}$, where $k_6$ represents the third column of the gem and thus, $k_6$ is adjacent to $v_1$ and nonadjacent to $v_2$.

        \case $v_1$ in $S_{16}$ and $v_2$ in $S_{64} \cup S_{65}$. We find $F_1(5)$ induced by $\{ v_1$, $v_2$, $s_{12}$, $s_{24}$, $s_{45}$, $k_1$, $k_{2}$, $k_{4}$, $k_{5} \}$ since $v_1$ is not complete to $K_1$ and $v_2$ is not complete to $K_5$ by definition.
    \end{mycases}
    Hence $\mathbb B_b$ is nested.
    Suppose that $\mathbb B_{b-r}$ is not nested. Every row in $\mathbb B_{b-r}$ represents a vertex in $S_{62}$ or $S_{63}$. It follows from Claim~\ref{claim:4tent_nogeminnestedmatrices} that it suffices to consider a $0$-gem induced by two vertices $v_1$ in $S_{62}$ and $v_2$ in $S_{63}$. Let $k_{61}$, $k_{62}$ and $k_3$ be the vertices represented by the columns of the gem. We find $F_0$ induced by $\{ v_1$, $v_2$, $s_{24}$, $k_{61}$, $k_{62}$, $k_{2}$, $k_3$, $k_4 \}$.

    Finally, suppose that $\mathbb B_{r-b}$ is not nested. Every row in $\mathbb B_{r-b}$ represents a vertex in $S_{61}$, $S_{64}$ or $S_{65}$. Let $v_1$ and $v_2$ in $S_{61} \cup S_{64} \cup S_{65}$ be two vertices whose rows induce a $0$-gem.
    If $v_1$ in $S_{61}$ and $v_2$ in $S_{64} \cup S_{65}$, then we find $F_2(5)$ induced by $\{ v_1$, $v_2$, $s_{12}$, $s_{24}$, $s_{45}$, $k_{61}$, $k_{1}$, $k_{2}$, $k_{4}$, $k_5 \}$.
    Similarly, we find $F_0$ induced by $\{ v_1$, $v_2$, $s_{45}$, $k_{61}$, $k_{2}$, $k_{4}$, $k_{51}$, $k_{52} \}$ if $v_1$ in $S_{64}$ and $v_2$ in $S_{65}$, and this finishes the proof.
\end{proof}

\begin{teo} \label{teo:finalteo_4tent}
    Let $G=(K,S)$ be a split graph containing an induced $4$-tent. Then, the following are equivalent:
\begin{enumerate}
\item $G$ is circle;
\item $G$ is \fsc-free;
\item $\mathbb B_1,\mathbb B_2,\ldots,\mathbb B_6$ are $2$-nested and $\mathbb B_r$, $\mathbb B_b$, $\mathbb B_{r-b}$ and $\mathbb B_{b-r}$ are nested.
\end{enumerate}
\end{teo}

\begin{proof}
It is not hard to see that $(1) \Rightarrow (2)$, and that
$(2) \Rightarrow (3)$ is a consequence of the previous lemmas. We will show $(3) \Rightarrow (1)$. Suppose that each of the matrices $\mathbb B_1,\mathbb B_2,\ldots,\mathbb B_6$ is $2$-nested and the matrices $\mathbb B_r$, $\mathbb B_b$, $\mathbb B_{r-b}$ or $\mathbb B_{b-r}$ are nested.
Let $\Pi$ be the ordering for all the vertices in $K$ obtained by concatenating each suitable LR-ordering $\Pi_i$ for $i \in \{1, 2,\ldots, 6\}$.

\begin{figure}[h!]
    \centering
     \includegraphics[scale=.45]{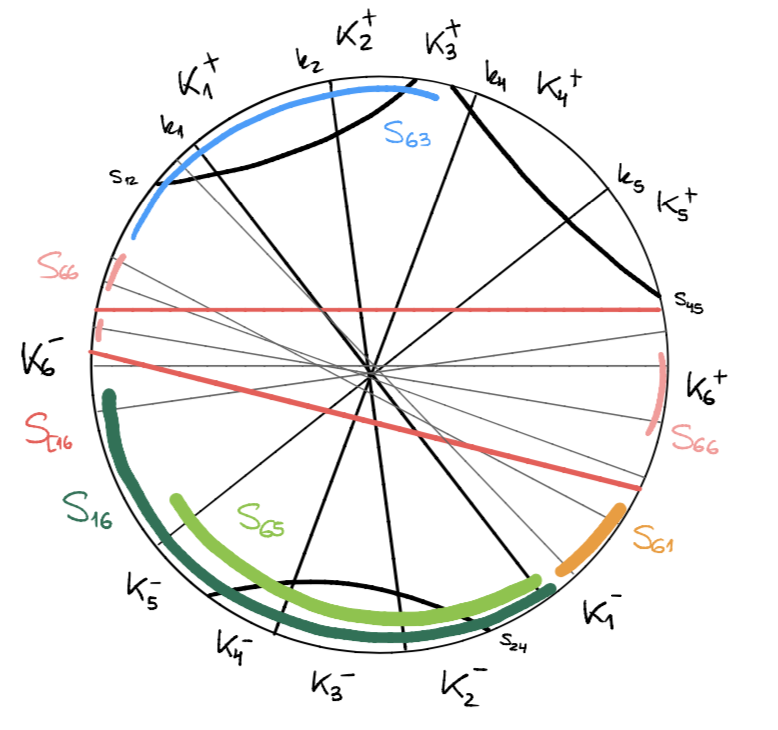}
    \caption{Sketch model of $G$ with some of the chords associated to the rows in $\mathbb{B}_6$.}
     \label{fig:4tent_model}
\end{figure}

Consider the circle divided into twelve pieces as in Figure~\ref{fig:4tent_model}. For each $i\in\{1,2,\ldots,6\}$ and for each vertex $k_i \in K_i$ we place a chord having one endpoint in $K_i^+$ and the other endpoint in $K_i^-$, in such a way that the ordering of the endpoints of the chords in $K_i^+$ and $K_i^-$ is $\Pi_i$.

Let us see how to place the chords for each subset $S_{ij}$ of $S$. First, some useful remarks.

\begin{remark} The following assertions hold:
    \begin{itemize}
        \item By Lemma~\ref{lema:matrices_union_son_nested_4tent}, all the vertices in $S_{ij}$ are nested, for every pair $i, j = \{1, 2, \ldots, 6\}$, $i \neq j$. This follows since any two vertices in $S_{ij}$ are nondisjoint. Moreover, in each $S_{ij}$, all the vertices are colored with either one color (the same), or they are colored red-blue or blue-red. Hence, these vertices are represented by rows in the matrices $\mathbb B_{r-b}$ and $\mathbb B_{b-r}$ and therefore they must be nested since each of these matrices is a nested matrix.
        \item It follows from the previous and Claim~\ref{claim:tent_0} that, if $i \geq k$ and $j\leq l$, then every vertex in $S_{ij}$ is nested in every vertex of $S_{kl}$.
        \item Also as a consequence of the previous and Lemma~\ref{lema:matrices_union_son_nested_4tent}, if we consider only those vertices labeled with the same letter in some $\mathbb B_i$, then there is a total ordering of these vertices. This follows from the fact that, if two vertices $v_1$ and $v_2$ are labeled with the same letter in some $\mathbb B_i$, since $\mathbb B_i$ is --in particular-- admissible, then they are nested in $K_i$. Moreover, if $v_1$ and $v_2$ are labeled with L in $\mathbb B_i$, then they are either complete to $K_{i-1}$ or labeled with R in $\mathbb B_{i-1}$. Thus, there is an index $j_l$ such that $v_i$ is labeled with R  in $\mathbb B_{j_l}$, for $l=1,2$. Therefore, we can find in such a way a total ordering of all these vertices.
        \item If $v_1$ and $v_2$ are labeled with distinct letters in some $\mathbb B_i$, then they are either disjoint in $K_i$ (if they are colored with the same color) or $N_{K_i}(v_1) \cup N_{K_i}(v_2) = K_i$  (if they are nondisjoint and colored with distinct colors), for there are no $D_1$ or $D_2$ in $\mathbb B_i$ for every $i \in \{1, 2, \ldots, 6\}$.
    \end{itemize}
\end{remark}

Notice that, when we define the matrix $\mathbb B_6$, we pre-color every vertex in $S_{[15]}$ with the same color. Since we are assuming $\mathbb B_6$ is $2$-nested and thus in particular is admissible, the subset $S_{[15]} \neq \emptyset$ if and only if the vertices represented in $\mathbb B_6$ are either all vertices in $S_{66} \cup S_{[16}$ and vertices that are represented by labeled rows $r$, all of them colored with the same color and labeled with the same letter L or R.
Moreover, since $\mathbb B_6$ is admissible, the sets $N_{K_6}(S_{6i}) \cap N_{K_6}(S_{j6})$ are empty, for $i = 1, 4, 5$ and $j =3,4$. The same holds for the sets $N_{K_6}(S_{6i}) \cap N_{K_6}(S_{j6})$, for $i = 2, 3$ and $j =2, 5, 1$.

If $S_{[16}= \emptyset$, then the placing of the chords that represent vertices with one or both endpoints in $K_6$ is very similar as in the tent case. Suppose that $S_{[16} \neq \emptyset$.
Before proceeding with the guidelines to draw the circle model, let us see a few remarks on the relationship between the vertices in $S_{ij}$ with either $i=6$ or $j=6$ and those vertices in $S_{[16}$, which follow from the claims stated throughout the proof of Lemma~\ref{lema:B6_2nested_4tent}:

\begin{remark} \label{obs:4tent_guidelines_model}
Let $G$ be a circle graph that contains no induced tent but contains an induced $4$-tent, and such that each matrix $\mathbb B_i$ is $2$-nested for every $i=1,2, \ldots, 6$. Then, all of the following statements hold:
\begin{itemize}
\item If $S_{26} \cup S_{16} \neq \emptyset$, then $S_{64} \cup S_{65} = \emptyset$, and viceversa.
\item If $v_1$ in $S_{36}\cup S_{46}$ and $v_2$ in $S_{61} \cup S_{64} \cup S_{65}$, then for every vertex $w$ in $S_{[16}$ either $N_{K_6}(v_1) \subseteq N_{K_6}(w)$ or $N_{K_6}(v_2) \subseteq N_{K_6}(w)$. Moreover, $v_1$ and $v_2$ are disjoint in $K_6$. The same holds for $v_1$ in $S_{56} \cup S_{26} \cup S_{16}$ and $v_2$ in $S_{62} \cup S_{63}$.
\item If $v$ in $S_{56}$ and $w$ in $S_{65}$, then $N_{K_5}(v) \cap N_{K_5}(w) = \emptyset$. The same holds for $v$ in $S_{61}$ and $w$ in $S_{16}$.
\end{itemize}
\end{remark}

Let $v$ in $S_{ij} \neq S_{[16}$ and $w$ in $S_{[16}$, with either $i=6$ or $j=6$.
Suppose first that $i=j=6$. Since $\mathbb B_6$ is $2$-nested, the submatrix induced by the rows that represent $v$ and $w$ in $\mathbb B_6$ contains no monochromatic gems or monochromatic weak gems.
If instead $i<j$, since $\mathbb B_6$ is admissible, then the submatrix induced by the rows that represent $v$ and $w$ in $\mathbb B_6$ contains no monochromatic weak gem, and thus we can place the endpoint of $w$ corresponding to $K_6$ in the arc portion $K^+_6$ and the $K_6$ endpoint of $v$ in $K^-_6$, or viceversa.


Remember that, since we are considering a suitable LR-ordering, there is an L-row $m_L$ such that any L-row and every L-block of an LR-row are contained in $m_L$ and every R-row and R-block of an LR-row are contained in the complement of $m_L$. Moreover, since we have a block bi-coloring for $\mathbb B_6$, then for each LR-row one of its blocks is colored with red and the other is colored with blue.
Hence, for any LR-row, we can place one endpoint in the arc portion $K^+_6$ using the ordering given for the block that colored with red, and the other endpoint in the arc portion $K^-_6$ using the ordering given for the block that is colored with blue.
Notice that, if $\mathbb B_6$ is $2$-nested, then all the rows labeled with L (resp.\ R) and colored with the same color and those L-blocks (resp.\ R-blocks) of LR-rows are nested. In particular, the L-block (resp.\ R-block) of every LR-row contains all the L-blocks of those rows labeled with L (resp.\ R) that are colored with the same color.
Equivalently, let $r$ be an LR-row in $\mathbb B_6$ with its L-block $r_L$ colored with red and its R-block $r_R$ colored with blue, $r'$ be a row labeled with L and $r''$ be a row labeled with R. Hence, if $r_L$, $r'$ and $r''$ are colored with the same color, then $r$ contains $r'$ and $r \cap r'' = \emptyset$. This holds since we are considering a suitable LR-ordering and a total block bi-coloring of the matrix $\mathbb B_6$, thus it contains no $D_0$, $D_1$, $D_2$, $D_8$ or $D_9$.
Since every matrix $\mathbb B_r$, $\mathbb B_b$, $\mathbb B_{r-b}$ and $\mathbb B_{b-r}$ is nested, there is a total ordering for the nondisjoint rows in each of these matrices. In other words, there is a total ordering for all the rows that intersect that are colored with the same color, or with red-blue or with blue-red, respectively. Moreover, if $v$ and $w$ are two vertices in $S$ such that they both have rows representing them in one of these matrices --hence, they are colored with the same color or sequence of colors--, then either $v$ and $w$ are disjoint or nested.

Notice that there are no other conditions besides being disjoint or nested outside each of the following subsets: $S_{11}$, $S_{22}$, $S_{33}$, $S_{44}$, $S_{55}$, $S_{66}$.
For the subset $S_{12}$, we only need to consider if every vertex in $S_{12} \cup S_{11} \cup S_{22}$ are disjoint or nested. The same holds for the subsets $S_{24}$ and $S_{45}$, considering $S_{22} \cup S_{44}$ and $S_{44} \cup S_{55}$, respectively.
Since each matrix $\mathbb B_i$ is $2$-nested for every $i=1, 2, \ldots, 6$, if there are vertices in both $S_{23}$ and $S_{34}$, then they are disjoint in $K_3$. The same holds for vertices in $S_{62}$ and $S_{63}$, and $S_{61}$ and $S_{14} \cup S_{15} \cup S_{16}$.
This is in addition to every property seen in Remark~\ref{obs:4tent_guidelines_model}.

With this in mind, we give guidelines to build a circle model for $G$.
We place first the chords corresponding to every vertex in $K$, using the ordering $\Pi$.
For each subset $S_{ij}$, we order its vertices with the inclusion ordering of the neighborhoods in $K$ and the ordering $\Pi$. When placing the chords corresponding to the vertices of each subset, we do it from lowest to highest according to the previously stated ordering given for each subset.

First, we place those vertices in $S_{ii}$ for each $i=1,2, \ldots, 6$, considering the ordering given by inclusion. If $v$ in $S_{ii}$ and the row that represents $v$ is colored with red, then both endpoints of the chord corresponding to $v$ are placed in $K_i^+$. If instead the row is colored with blue, then both endpoints are placed in $K_i^-$.

For each $v$ in $S_{ij} \neq S_{[16}$, if the row that represents $v$ in $\mathbb B_i$ is colored with red (resp.\ blue) , then we place the endpoint corresponding to $K_i$ in the portion $K_i^+$ (resp.\ $K_i^-$) . We apply the same rule for the endpoint corresponding to $K_j$.

Let us consider now the vertices in $S_{[15]}$. If $G$ is circle, then all the rows in $\mathbb B_6$ are colored with the same color. Moreover, if $S_{[15]} \neq \emptyset$, then either every row labeled with L or R in $\mathbb B_6$ is labeled with L and colored with red or labeled with R and colored with blue, or viceversa.
Suppose first that every row labeled with L or R in $\mathbb B_6$ is labeled with L and colored with red or labeled with R and colored with blue. In that case, every row representing a vertex $v$ in $S_{[15]}$ is colored with blue, hence we place one endpoint of the chord corresponding to $v$ in $K_6^+$ and the other endpoint in $K_6^-$. In both cases, the endpoint of the chord corresponding to $v$ is the last chord of a vertex of $S$ that appears in the portion of $K_6^+$ and is the first chord of a vertex of $S$ that appears in the portion of $K_6^-$. We place all the vertices in $S_{[15]}$ in such a manner.
If instead every row labeled with L or R in $\mathbb B_6$ is labeled with L and colored with blue or labeled with R and colored with red, then every row representing a vertex in $S_{[15]}$ is colored with red. We place the endpoints of the chord in $K_6^-$ and $K_6^+$, as the last and first chord that appears in that portion, respectively.

Finally, let us consider now a vertex $v$ in $S_{[16}$. Here we have two possibilities: (1) the row that represents $v$ has only one block, (2) the row that represents the row that represents $v$ has two blocks of $1$'s.
Let us consider the first case. If the row that represents $v$ has only one block, then it is either an L-block or an R-block. Suppose that it is an L-block. If the row in $\mathbb B_6$ is colored with red, then we place one endpoint of the chord as the last of $K_6^-$ and the other endpoint in $K_6^+$, considering in this case the partial ordering given for every row that has an L-block colored with red in $\mathbb B_6$.
If instead the row in $\mathbb B_6$ is colored with blue, then we place one endpoint of the chord as the first of $K_6^+$ and the other endpoint in $K_6^-$, considering in this case the partial ordering given for every row that has an L-block colored with blue in $\mathbb B_6$. The placement is analogous for those LR-rows that are an R-block.
Suppose now that the row that represents $v$ has an L-block $v_L$ and an R-block $v_R$. If $v_L$ is colored with red, then $v_R$ is colored with blue. We place one endpoint of the chord in $K_6^+$, considering the partial ordering given by every row that has an L-block colored with red in $\mathbb B_6$, and the other endpoint of the chord in $K_6^-$, considering the partial ordering given by every row that has an R-block colored with blue in $\mathbb B_6$. The placement is analogous if $v_L$ is colored with blue.

This gives a circle model for the given split graph $G$.
\end{proof}

\subsection{Split circle graphs containing an induced co-4-tent} \label{subsec:circle4}

In this section we will address the last case of the proof of Theorem~\ref{teo:circle_split_caract}, which is the case where $G$ contains an induced co-$4$-tent. This case is mostly similar to the $4$-tent case, with two particular difference.
First of all, the co-$4$-tent is not a prime graph. This implies that there is more than one possible circle model for this graph. Moreover, if a non-circle graph $G$ is also not a prime graph, then for any split decomposition one of its factors is a non-circle graph. This follows from the fact that circle graphs are closed under split decomposition~\cite{Bou-circle}. Hence, a minimally non-circle graph is necessarily a prime graph. However, this problem was solved in Section~\ref{section:partitions}, more precisely in Remark~\ref{obs:co4tent_1}. We assume $K_2 \neq \emptyset$ and $K_4 \neq \emptyset$ in order to work with a minimally non-circle graph throughout the proof.
Second of all, even when considering these extra hypothesis to work with a prime graph, we have to divide the proof in $5$ cases. However, we will see that one of these $5$ cases is indeed more general than the other four.



Let $G=(K,S)$ and $H$ as in Section~\ref{subsec:co4tent_partition}. For each $i\in\{1,2,\ldots,8\}$, let $\mathbb C_i$ be a $(0,1)$-matrix having one row for each vertex $s\in S$ such that $s$ belongs to $S_{ij}$ or $S_{ji}$ for some $j\in\{1,2,\ldots,8\}$ and one column for each vertex $k\in K_i$ and such that such that the entry corresponding to row $s$ and column $k$ is $1$ if and only if $s$ is adjacent to $k$ in $G$.
For each $j\in\{1,2,\ldots,8\}-\{i\}$, we label those rows corresponding to vertices of $S_{ji}$ with L and those corresponding to vertices of $S_{ij}$ with R, with the exception of those rows in $\mathbb{C}_7$ that represent vertices in $S_{76]}$ and $S_{[86]}$which are labeled with LR.
Notice that we have considered those vertices that are complete to $K_1, \ldots, K_5$ and $K_8$ and are also adjacent to $K_6$ and $K_7$ divided into two distinct subsets. Thus, $S_{76}$ are those vertices that are not complete to $K_6$ and therefore the corresponding rows are labeled with R in $\mathbb C_6$ and with L in $\mathbb C_7$.
As in the $4$-tent case, there are LR-rows in $\mathbb C_7$. Moreover, there may be some empty LR-rows, which represent those vertices of $S$ that are complete to $K_1, \ldots, K_6$ and $K_8$ and are anticomplete to $K_7$. These vertices are all pre-colored with the same color, and that color is assigned depending on whether $S_{74} \cup S_{75} \cup S_{76} \neq \emptyset$ or $S_{17} \cup S_{27} \neq\emptyset$.
We color some of the remaining rows of $\mathbb C_i$ as we did in the previous sections, to denote in which portion of the circle model the chords have to be drawn.
In order to characterize the forbidden induced subgraphs of $G$ and using an argument of symmetry, we will only analyse the properties of the matrices $\mathbb C_1$, $\mathbb C_2$, $\mathbb C_3$, $\mathbb C_6$ and $\mathbb C_7$, since the matrices $\mathbb C_i$ $i=4,5,8$ are symmetric to $\mathbb C_2$, $\mathbb C_3$ and $\mathbb C_6$, respectively.


We consider $5$ distinct cases, according to whether the subsets $K_6$, $K_7$ and $K_8$ are empty or not, for the matrices we need to define may be different in each case.
Using the symmetry of the subclasses $K_6$ and $K_8$, the cases we need to study are the following: (1) $K_6, K_7, K_8 \neq \emptyset$, (2) $K_6, K_7 \neq \emptyset$, $K_8 = \emptyset$, (3) $K_6, K_8 \neq \emptyset$, $K_7 = \emptyset$, (4) $K_6 \neq \emptyset$, $K_7, K_8 = \emptyset$, (5) $K_7 \neq \emptyset$, $K_6, K_8 = \emptyset$.
We have the following lemma. For more details on its proof, see~\cite{P20}.

\begin{lema}
Let $\mathbb{C}^{j}_1, \ldots, \mathbb{C}^{j}_8$ be the matrices defined for each of the previously stated cases $j =1, \ldots, 5$. Then, for each $i \in \{1, \ldots, 8\}$, the matrix $\mathbb{C}^{j}_i$ is a submatrix of $\mathbb{C}^{1}_i$ for every $j =2,3,4,5$. 
\end{lema}

In (1), the subsets are given as described in Table~\ref{fig:tabla_co4tent_1}, and thus the matrices $\mathbb C_i$ as are follows:
\[ \mathbb C_1 = \scriptsize{ \bordermatrix{ & K_1\cr
                S_{12}\ \textbf{L} & \cdots \cr
                 S_{11}\            & \cdots \cr
                S_{16]}\ \textbf{L} & \cdots \cr
                S_{17}\ \textbf{L} & \cdots }\
                \begin{matrix}
                \textcolor{red}{\bullet} \\ \\ \textcolor{blue}{\bullet} \\ \textcolor{blue}{\bullet}
                \end{matrix} } \qquad
   \mathbb C_2 = \scriptsize{ \bordermatrix{ & K_2\cr
                S_{12}\ \textbf{R} & \cdots \cr
                S_{22}\            & \cdots \cr
                S_{23}\ \textbf{L} & \cdots \cr
                S_{25]}\ \textbf{L} & \cdots \cr
                S_{26}\ \textbf{L} & \cdots }\
                \begin{matrix}
                \textcolor{red}{\bullet} \\ \\ \textcolor{blue}{\bullet} \\ \textcolor{blue}{\bullet} \\ \textcolor{blue}{\bullet}
                \end{matrix} }\qquad
   \mathbb C_3 = \scriptsize{ \bordermatrix{ & K_3\cr
                S_{13}\ \textbf{R} & \cdots \cr
                S_{34}\ \textbf{L} & \cdots \cr
                S_{33}\            & \cdots \cr
                S_{35}\ \textbf{L} & \cdots \cr
                S_{36}\ \textbf{L} & \cdots \cr
                S_{23}\ \textbf{R} & \cdots }\
                \begin{matrix}
                \textcolor{red}{\bullet} \\  \textcolor{red}{\bullet} \\ \\ \textcolor{blue}{\bullet} \\ \textcolor{blue}{\bullet} \\ \textcolor{blue}{\bullet}
                \end{matrix} }  \]
\[   \mathbb C_6 = \scriptsize{ \bordermatrix{ & K_6\cr
                S_{66}\            & \cdots \cr
                S_{26}\ \textbf{R} & \cdots \cr
                S_{36}\ \textbf{R} & \cdots \cr
                S_{[46}\ \textbf{R} & \cdots \cr
                S_{76}\ \textbf{R} & \cdots \cr
                S_{[86}\ \textbf{R} & \cdots }\
                \begin{matrix}
                \\ \\  \textcolor{blue}{\bullet} \\  \textcolor{blue}{\bullet} \\ \textcolor{blue}{\bullet} \\ \textcolor{red}{\bullet} \\ \textcolor{red}{\bullet} \\ \\
                \end{matrix} } \qquad
    \mathbb C_7 = \scriptsize{ \bordermatrix{ & K_7\cr
                S_{17}\ \textbf{R} & \cdots \cr
                S_{[27}\ \textbf{R} & \cdots \cr
                S_{77}\            & \cdots \cr
                S_{74]}\ \textbf{L} & \cdots \cr
                S_{75}\ \textbf{L} & \cdots \cr
                S_{76}\ \textbf{L} & \cdots \cr
                S_{87}\ \textbf{R} & \cdots \cr
                S_{[86]}\ \textbf{LR} & \cdots \cr
                S_{76]}\ \textbf{LR} & \cdots }\
                \begin{matrix}
                \textcolor{blue}{\bullet} \\  \textcolor{blue}{\bullet} \\ \\ \textcolor{blue}{\bullet} \\ \textcolor{blue}{\bullet} \\ \textcolor{blue}{\bullet}  \\ \textcolor{blue}{\bullet} \\ \\ \\
                \end{matrix} } \]

Let us suppose that $K_6, K_7, K_8 \neq \emptyset$.
The Claims in Section~\ref{section:partitions} and the following prime circle model allow us to assume that some subsets of $S$ are empty.
We denote by $S_{87}$ the set of vertices in $S$ that are complete to $K_1, \ldots, K_6$, are adjacent to $K_7$ and $K_8$ but are not complete to $K_8$, and analogously $S_{76}$ is the set of vertices in $S$ that are complete to $K_1, \ldots, K_5, K_8$, are adjacent to $K_6$ and $K_7$ but are not complete to $K_6$. Hence, $S_{76]}$ denotes the vertices of $S$ that are complete to $K_1, \ldots, K_6, K_8$ and are adjacent to $K_7$.

\begin{figure}[h!]
\centering
\includegraphics[scale=.45]{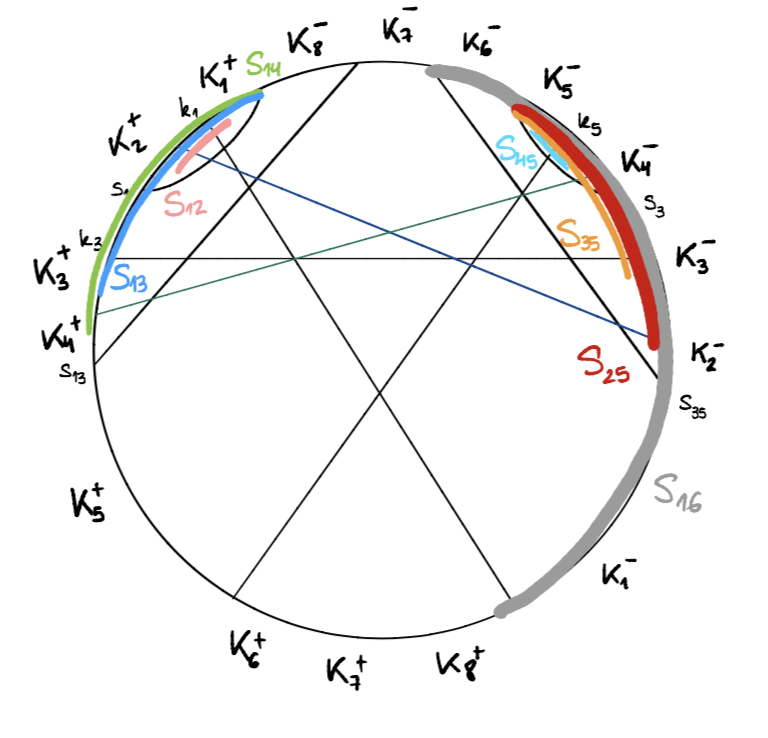}
      \caption{A circle model for the co-$4$-tent graph.} \label{fig:co4tent_model}
\end{figure}

We state some results analogous as the ones seen in the previous sections for a graph that contains an induced co-$4$-tent and contains no $4$-tent or tent. These results and its proofs are analogous to Lemmas~\ref{lema:equiv_circle_2nested_4tent_sinLR},~\ref{lema:B6_2nested_4tent} and Theorem~\ref{teo:finalteo_4tent}, thus we will just state the results (for the complete proof of each of these results, see~\cite{P20}).|

\begin{lema}[\cite{P20}] \label{lema:equiv_circle_2nested_co4tent}
    If $\mathbb C_1$, $\mathbb C_2, \ldots, \mathbb C_8$ are not $2$-nested, then $G$ contains one of the graphs listed in Figure~\ref{fig:forb_graphs} as induced subgraphs.
\end{lema}

We define the matrices $\mathbb{C}_r$, $\mathbb{C}_b$, $\mathbb{C}_{r-b}$ and $\mathbb C_{b-r}$ as in Section~\ref{subsubsec:4tent3}. Similarly, we have the following lemma for these matrices.

\begin{lema}[\cite{P20}] \label{lema:matrices_union_son_nested_co4tent}
    Suppose that $\mathbb C_i$ is $2$-nested for each $i =1,2 \ldots, 8$. If $\mathbb C_r$, $\mathbb C_b$, $\mathbb C_{r-b}$ or $\mathbb C_{b-r}$ are not nested, then $G$ contains $F_0$ as a minimal forbidden induced subgraph for the class of circle graphs.
\end{lema}

The main result of this section is the following theorem.

\begin{teo}[\cite{P20}] \label{teo:finalteo_co4tent}
    Let $G=(K,S)$ be a split graph containing an induced co-$4$-tent. Then, the following are equivalent:
\begin{enumerate}
\item $G$ is circle;
\item $G$ is \fsc-free;
\item $\mathbb C_1,\mathbb C_2,\ldots,\mathbb C_8$ are $2$-nested and $\mathbb C_r$, $\mathbb C_b$, $\mathbb C_{r-b}$ and $\mathbb C_{b-r}$ are nested.
\end{enumerate}
\end{teo}

\subsection{Split circle graphs containing an induced net and proof of Theorem~\ref{teo:circle_split_caract}} \label{subsec:circle5}

Let $G=(K,S)$ be a split graph. If $G$ is a minimally non-circle graph, then it contains either a tent, or a $4$-tent, or a co-$4$-tent, or a net as induced subgraphs.
In the previous sections, we addressed the problem of having a split minimally-non-circle graph that contains an induced tent, $4$-tent and co-$4$-tent, respectively. Let us consider a split graph $G$ that contains no induced tent, $4$-tent or co-$4$-tent, and suppose there is a net subgraph in $G$.

We define $K_i$ as the subset of vertices in $K$ that are adjacent only to $s_i$ if $i=1,3,5$, and if $i=2,4,6$ as those vertices in $K$ that are adjacent to $s_{i-1}$ and $s_{i+1}$. We define $K_7$ as the subset of vertices in $K$ that are nonadjacent to $s_1$, $s_3$ and $s_5$.
Let $s$ in $S$. We denote by $T(s)$ the vertices that are false twins of $s$.

\begin{remark}
The net is not a prime graph. Moreover, if $K_i = \emptyset$, $K_j = \emptyset$ for any pair $i, j \in \{2,4,6\}$, then $G$ is not prime. For example, if $K_2 = \emptyset$ and $K_4 = \emptyset$, then a split decomposition can be found considering the induced subgraphs $H_1 = \{u\} \cup K_3 \cup T(s_3)$ and $H_2 = V(G) \setminus T(s_3)$, where $u \not\in V(G)$ is complete to $K_3$ and anticomplete to $V(G) \setminus T(s_3)$.
\end{remark}

Since in the proof we consider a minimally non-circle graph $G$, it follows from the previous remark that at least two of $K_2$, $K_4$ and $K_6$ must be nonempty so that $G$ results prime. However, in either case we find a $4$-tent as an induced subgraph. Therefore, as a consequence of this and the previous sections, we have now proven that any graph containing no graph in \fsc\ as an induced subgraph is a circle graph. Hence, since the class of circle graphs is hereditary, in order to complete the proof of Theorem~\ref{teo:circle_split_caract}, it suffices to verify that no graph in \fsc\ is a circle graph. We do so in the lemma below.

\begin{lema}\label{lem:nocircle} None of the graphs in \fsc\ is a circle graph.\end{lema}
\begin{proof} We show this for the odd $k$-suns with center and the even $k$-suns; the proofs for the remaining graphs are similar (for more details, see~\cite{P20}). We use throughout the proof that $BW_3$, $W_5$ and $W_7$ (see Figure~\ref{fig:bouchet_lc}) are non-circle graphs~\cite{Bou-circle-obs}, for every $k \geq 2$.

\begin{figure}[h]
\centering
\includegraphics[scale=.4]{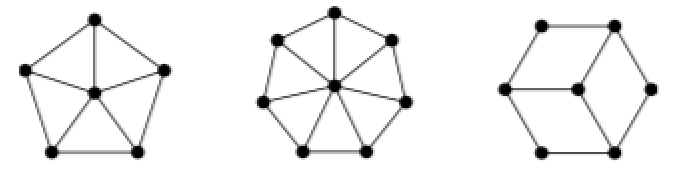}
\caption{The graphs $W_5$, $W_7$ and $BW_3$.} \label{fig:bouchet_lc}
\end{figure}


Let us first consider an odd $k$-sun with center, where $v_1, \ldots, v_k$ are the vertices of the clique of size $k$, $w_1, \ldots, w_k$ are the \emph{petals} (the vertices of degree $2$) and $x$ is the \emph{center} (the vertex adjacent to every vertex of the clique of size $k$). If $k=3$, we consider the local complement with respect to the center and we obtain $BW_3$. If $k=5$ or $k=7$, we apply local complementation with respect to $x$, $w_1$, $\ldots$ and $w_k$, obtaining $W_5$ and $W_7$, respectively. 
If instead $k=9$, we first apply local complementation with respect to $x$, $w_1$, $\ldots$ and $w_k$ and we obtain $W_9$. Once we have this wheel, we apply local complement with respect to $v_1$, $v_2$ and $v_9$ and obtain $W_6$ induced by $\{v_3,\ldots,v_8\}$. If we now consider the local complement with respect to $v_3$, $v_6$ and $x$ (in that order), we find $\overline{C_6}$, which is not a circle graph because it is locally equivalent to $W_5$.
More in general, for every $k \geq 11$ we can obtain a $W_k$ considering the sequence described at the beginning of the paragraph. Once we have a wheel, if we apply local complementation by $v_1$, $v_2$ and $v_k$, then we obtain $W_{k-3}$. We can repeat this until $k-3 < 8$, in which case either $k=5$, $k=6$ or $k=7$ and we reduce this to one of the previous cases.


Let us consider now an even $k$-sun, where $v_1, \ldots, v_k$ are the vertices of the clique of size $k$ and $w_1, \ldots, w_k$ are the \emph{petals} (the vertices of degree $2$), where $w_i$ is adjacent to $v_i$ and $v_{i+1}$. If $k=4$, then we apply local complementation with respect to the sequence $w_1$, $w_2$, $w_3$, $w_4$, $v_1$, $w_4$, $w_1$, $w_3$, $v_3$, $w_2$ and $v_1$, and we obtain $\overline{C_6}$, which is locally equivalent to $W_5$.
If $k=6$ and we apply local complementation to the sequence $w_1$, $w2, \ldots, w_k$, then we obtain $\overline{C_6}$.
Let us consider an even $k\geq 8$, thus $k=2j$ for some $j\geq 4$ and let $l=\frac{k-8}{2}$. We apply local complementation with respect to the sequence $v_1$, $v_{j+1}$, $v_2$, $v_k$, $\ldots$, $v_{2+l}$, $v_{k-l}$ and we find $W_5$ or $W_7$ induced by $\{v_1$, $v_{j+1}$, $v_{j+1-2}$, $v_{j+1-1}$, $v_{j+1+1}$, $v_{j+1+2} \}$ or by $\{v_1$, $v_{j+1}$, $v_{j+1-3}$, $v_{j+1-2}$, $v_{j+1-1}$, $v_{j+1+1}$, $v_{j+1+2}$, $v_{j+1+3} \}$, depending on whether $k \equiv 2 \mod 4$ or $k \equiv 0 \mod 4$, respectively.
\end{proof}

This completes the proof of Theorem~\ref{teo:circle_split_caract}. Moreover, will now prove that the characterization given in this theorem is a characterization by minimal forbidden induced subgraphs.

\begin{teo} \label{teo:caract_is_minimal}
All the graphs in \fsc\ are minimally non-circle.
\end{teo}

\begin{proof}
Since all the graphs depicted in Figure~\ref{fig:forb_graphs} are non-circle graphs, it suffices to see that none of these graphs is an induced subgraph of any of the others. Moreover, if we consider the matrix $A(S,K)$ for each of these graphs, then it suffices to see that none of these matrices is a subconfiguration of any other corresponding to a graph on the list. Let $\mathcal{U}$ be the set containing all these matrices. 
We call the matrix $A(S,K)$ for each graph in $\mathcal{U}$ with the name as the corresponding graph, and we denote with $M_I(k)$ the matrix $A(S,K)$ for the $k$-sun, for each $k \geq 3$. Notice that the matrices corresponding to $k$-suns, $M_{II}(2j)$, $M_{III}(3)$, $M_{III}(2j)$, $M_{IV}$ and $M_V$ are all Tucker matrices, for every $j\geq 2$. In particular, the matrix $A(S,K)$ that represents an even $k$-sun is $M_I(k)$. For its part, the matrices corresponding to the graphs $F_0$, $F_1(2j+1)$ and $F_2(2j+1)$ for every $j\geq 2$ are depicted in Figure~\ref{fig:forb_F}.
Moreover, the $A(S,K)$ matrix corresponding to an odd $k$-sun with center is $S_0(k)$, which is the Tucker matrix $M_I(k)$ plus one row having a $1$ in each column, and the $A(S,K)$ matrix of tent${}\vee{}K_1$ is $M_0$, which is the Tucker matrix $M_I(3)$ plus one column having a $1$ in every row (see Figure~\ref{fig:forb_M_chiquitas}).

If $M$ is a Tucker matrix, then $M$ is not a subconfiguration of any other Tucker matrix since these matrices are minimal forbidden subconfigurations for the C$1$P.  It follows that $M$ is not a subconfiguration of $F_0$, $F_1(2j+1)$ and $F_2(2j+1)$ for every $j \geq 2$ since these matrices do have the C$1$P.
Since $M_0$ and $S_0(2j)$ contain $M_I(3)$ and $M_I(2j-1)$ as subconfigurations, respectively, it follows using the same arguments that neither $M_0$ nor $S_0(2j)$ is a subconfiguration of any of the matrices in $\mathcal{U}$, for every $j \geq 2$. Finally, $S_0(2j)$ is not a subconfiguration of $M_0$ since $S_0(2j)$ always has more rows that $M_0$, and conversely, $M_0$ is not a subconfiguration of $S_0(2j)$ for in that case $M_I(3)$ would be a subconfiguration of $M_I(k)$ for some odd $k \geq 5$ and this contradicts the minimality of the Tucker matrices regarding the C$1$P.
Hence, it suffices to see that $F_0$, $F_1(2j+1)$ and $F_2(2j+1)$ are not subconfigurations of any other matrix in $\mathcal{U}$, for every $j \geq 2$.

First of all, notice that no matrix in $\mathcal{U}$ has three or more rows that have at least three $1$'s, hence it is not possible for $F_0$ to be a subconfiguration of any of these matrices.
Toward a contradiction, suppose there is $M$ in $\mathcal{U}$ such that $F_1(2j+1)$ is a subconfiguration of $M$, for some $j \geq 2$. Since $j\geq 2$, $M$ has at least five rows in total, and in particular, at least two distinct rows that have at least three $1$'s. It follows that $M$ is either $M_{IV}$, $M_V$ or $M_{II}(2j)$ for some $j \geq 2$.
Moreover, $F_1(2j+1)$ is not a subconfiguration of $M_{IV}$ or $M_V$ since these matrices have only four rows and $F_1(2j+1)$ has at least five rows.
Let us suppose that $M=M_{II}(k)$ for some even $k \geq 4$. Notice that $k \geq 6$ since $k\geq 2j+1$ and $k$ is even. Thus, the first and second rows of $F_1(2j+1)$ correspond to the first and last rows of $M$ since these are the only rows in $M$ having more than two $1$'s. Furthermore, since the first and last row of $M$ have only one $0$, then columns $1$ and $k-1$ of $M$ correspond to columns $1$ and $2j$ of $F_1(2j+1)$. Hence, rows $2$ to $2j-1$ of $F_1(2j+1)$ correspond either to rows $2$ to $2j-1$ of $M$ or to rows $k-1-2j$ to $k-1$. However, $k \geq 6$ and thus either row $2j-1$ or $k-1-2j$ of $M$ has a $0$ in the first or last column, and therefore $F_1(2j+1)$ cannot be a subconfiguration of $M$.
It follows analogously for $F_2(2j+1)$ and $M_{III}(k)$, and since none of the matrices has at least one row with at least three $1$'s and at least two $0$'s, with the exception of $M_{III}(k)$, this finishes the proof.
\end{proof}


\section{Final remarks and future challenges}\label{sec:final}


To conclude, we leave some possible future challenges about structural char\-ac\-ter\-i\-za\-tions on circle graphs.

\begin{itemize}

\item Recall that split graphs are those chordal graphs whose complement is also a chordal graph. Moreover, the graph depicted in Figure~\ref{fig:example_A''3} 
is a chordal graph that is neither circle nor a split graph. 
It follows that the list of graphs given in Theorem~\ref{teo:circle_split_caract} is not enough to characterize those chordal graphs that are also circle. However, Theorem~\ref{teo:circle_split_caract} is indeed a good first step to characterize circle graphs by forbidden induced subgraphs within the class of chordal graphs, which remains as an open problem.

\begin{figure}[h]
\centering
\includegraphics[scale=.15]{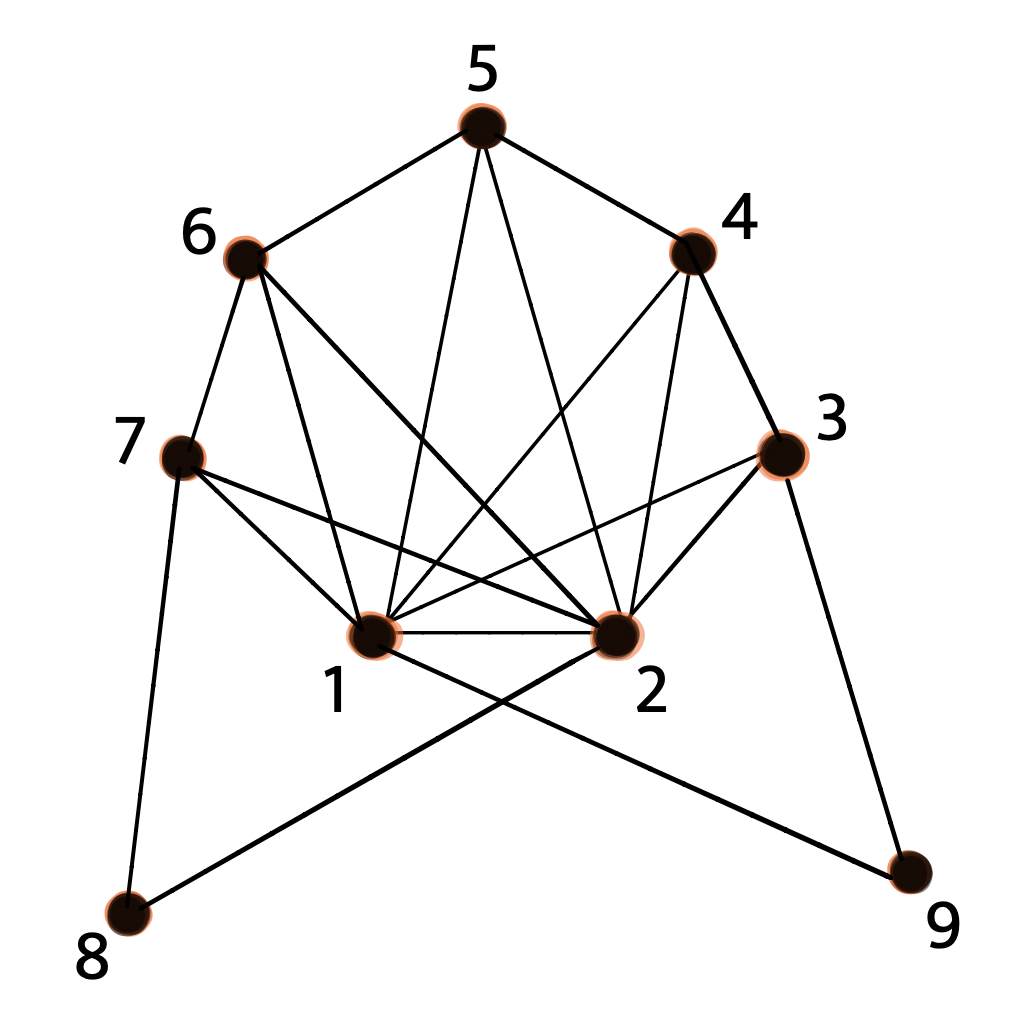}
\caption{Example of a chordal graph that is neither circle nor split.} \label{fig:example_A''3}
\end{figure}


\item Bouchet~\cite{B-99} showed that if a circle graph is the
complement of a bipartite graph, then its complement is also a
circle graph. Another possible con\-ti\-nua\-tion of this work
would be studying the characterization of those circle graphs
whose complement is a circle graph as well.

\item Characterize Helly circle graphs (graphs that have a model of chords with the Helly property) by forbidden induced subgraphs. The class of Helly circle graphs was characterized by forbidden induced subgraphs within circle graphs~\cite{DGR10}, but the problem of finding such a characterization for the whole class of Helly circle graphs is still open. Moreover, it would be interesting to find a decomposition analogous as the split decomposition is for circle graphs, this is, such that Helly circle graphs are closed under this decomposition.

\end{itemize}

\section*{Acknowledgements}
\label{sec:ack}

This work was partially supported by ANPCyT PICT-2015-2218,
UBACyT Grants 20020170100495BA and 20020160100095BA,
PIO CONICET UNGS-144-20140100011-CO, Universidad Nacional del Sur Grants PGI 24/L103
and PGI 24/L115, and ANPCyT PICT-2017-1315 (Argentina), and ISCI CONICYT PIA FB0816; ICM P-05-004-F (Chile).
Nina Pardal is partially supported by a CONICET doctoral fellowship.



\end{document}